\documentclass[aps,pra,
            twocolumn,10pt,superscriptaddress,
            longbibliography,floatfix]{revtex4-2}
\usepackage[utf8]{inputenc}
\usepackage[american]{babel}
\usepackage[normalem]{ulem}
\usepackage{textcomp}
\usepackage{amsmath}
\usepackage{mathtools}
\usepackage{nicefrac}
\usepackage{graphicx}
\usepackage{physics}
\usepackage{xcolor}
\usepackage[unicode=true,
 bookmarks=false,
 breaklinks=false,pdfborder={0 0 1},backref=false,colorlinks=false]
 {hyperref}
\usepackage{dsfont}
\usepackage{upgreek}
\usepackage{enumitem}
\usepackage[most]{tcolorbox}
\usepackage{cancel}
\usepackage[export]{adjustbox}
\usepackage{amsfonts}
\usepackage{amsthm}
\usepackage{latexsym}
\usepackage{bbm}

\usepackage{indentfirst}
\usepackage{tikz}
\usetikzlibrary{arrows.meta}
\usepackage{multirow}
\usepackage{cellspace}
\makeatother
\makeatletter

\newcommand{\newreptheorem}[2]{\newtheorem*{rep@#1}{\rep@title}\newenvironment{rep#1}[1]{\def\rep@title{#2 \ref*{##1}}\begin{rep@#1}}{\end{rep@#1}}}
\makeatother

\newtheorem{theorem}{Theorem}
\newreptheorem{theorem}{Theorem}
\newtheorem{algo}{Algorithm}
\newtheorem{lemma}[theorem]{Lemma}
\newtheorem{definition}[theorem]{Definition}

\newtheorem{proposition}[theorem]{Proposition}
\newtheorem{corollary}[theorem]{Corollary}

\newcommand{\ii}{\mathrm{i}}
\newcommand{\LD}{\mathcal{L}_\mathcal{D}(\mathcal{H}_d)}

\newcommand{\TII}{\affiliation{Quantum Research Center, Technology Innovation Institute, Abu Dhabi, UAE}}

\begin{document}
\title{Robustness of Quantum Signal Processing and Resource-aware Polynomial Design}
\author{Thiago O. Maciel}
\TII
\thanks{Corresponding author}
\email{thiago.maciel@tii.ae}
\author{Giancarlo Camilo}
\TII
\author{Allan Tosta}
\TII
\author{Daniel Stilck-Fran\c{c}a}
\affiliation{Department of Mathematical Sciences, University of Copenhagen, Universitetsparken 5, 2100 Denmark}
\author{Leandro Aolita}
\TII
\author{Thais de Lima Silva}
\TII


\begin{abstract}
Quantum signal processing (QSP) realizes operator functions through polynomial transformations of a finite degree. In addition to function approximation errors, realistic implementations incur algorithmic errors also when encoding the input operator and in synthesizing gates into a discrete fault-tolerant gate set. We provide an end-to-end analysis of these errors and show that the input error propagation is governed by the operator-Lipschitz constant $L_\star$ of the implemented function.  We then derive an efficient classical algorithm to compute tight lower and upper bounds on $L_\star$, and we show that these bounds induce in turn lower and upper bounds on the total quantum-resource cost (for instance, $T$-gate count) needed for QSP implementations compatible with a target total error.
We further devise a procedure to search for a polynomial with minimum cost based solely on its coefficients, yielding a resource-aware polynomial-design procedure for QSP. We illustrate the resulting method for ground state energy estimation via eigenvalue thresholding for the H$_6$ molecular Hamiltonian. In the tested hardware profiles, the resource-aware family improves the total-error budget and reduces the modeled $T$-gate count by $51.9\%$--$56.7\%$ relative to the standard analytical polynomial obtained from an error-function approximation.
\end{abstract}

\maketitle

\section{Introduction}\label{sec:introduction}

Quantum signal processing (QSP) and quantum singular-value transformation (QSVT) provide a recipe for implementing functions of operators via polynomial transformations encoded in quantum circuits~\cite{Low_Chuang_2019,Gilyen_2019,Martyn_2021}. 
For Hermitian inputs, QSP transforms the eigenvalues of a normalized operator accessed through either a block encoding \cite{Low_Chuang_2019} or the unitary time-evolution it generates \cite{silva2022,Low_2024}, while QSVT extends the construction to singular-value transformations of general matrices \cite{Gilyen_2019}. 
A practical QSP/QSVT implementation accumulates errors at several distinct layers: $(i)$ the polynomial incurs an approximation error; $(ii)$ the encoded operator can differ from the ideal input because of truncation or imperfect block encoding; and $(iii)$ the QSP rotations must be synthesized into a discrete fault-tolerant gate set.

Conventional polynomial design aims at minimizing the uniform approximation error over the domain, which determines how accurately the target function is represented. The resulting degree determines the query complexity under the chosen QSP/QSVT scheme. 
These quantities do not, however, determine how the implemented polynomial amplifies perturbations of the encoded operator. 
Two polynomials with comparable degree and approximation error can require different block-encoding precision. 
The missing quantity is a sensitivity of the implemented operator function to general, and in particular noncommuting, operator perturbations.

Here we use operator-Lipschitz (OL) continuity to make this sensitivity explicit. For a function $f$ which is Lipschitz on the relevant operator domain, its optimal OL constant $L_\star[f]$ is the smallest positive real number for which $\norm{f(A)-f(B)}\leq L_\star[f]\norm{A-B}$ holds for all admissible operators $A$ and $B$ in all dimensions. While not all scalar Lipschitz functions are OL, QSP implements only polynomial functions, which are OL.
This gives a direct robustness bound for QSP/QSVT and separates the scalar approximation error from the amplification of errors in the encoded operator. 
This bound is linear in the input error 
and improves over previous bounds for QSVT, which either carry a logarithmic factor that grows unbounded as the perturbation vanishes~\cite{aleksandrov2010holder,Gilyen_2018_arxiv} or have a square root dependence on the input error~\cite{Gilyen_2018_arxiv}. 
As a result, it tolerates larger input errors in the high-precision regime relevant to fault-tolerant implementations.

The remaining difficulty is that $L_\star[f]$ is generally not available in closed form.
In QSP/QSVT, however, the implemented function is always polynomial. We show that when $f$ is a polynomial, computing operator-Lipschitz bounds reduces to a finite-dimensional active-set optimization directly in terms of its polynomial coefficients and prove that it converges asymptotically to its semi-infinite formulation.
We derive an exact finite coefficient-space formulation for the divided difference of polynomials, an exchange algorithm with certified two-sided bounds, and a sufficient criterion for when the OL constant differs from its scalar counterpart.
Moreover, each iteration of our algorithm requires only $\mathcal{O}(N^2)$ memory, compared with the $\mathcal{O}(N^4)$ memory of a single interior-point step for the equivalent semidefinite program, which already exceeds $16$~GB at $N=123$.
In contrast, our certified operator-Lipschitz bounds reaches polynomial degrees of order $10^4$, which is in line with high-precision QSVT subroutines with degrees scaling as $\mathcal{O}(\Delta^{-1}\log(1/\varepsilon))$, such as eigenvalue thresholding with a small spectral gap $\Delta$ when compiled for fault-tolerant quantum computers~\cite{Gilyen_2018_arxiv,Martyn_2021}.

Once these bounds are available, they can be coupled directly to an end-to-end resource model. 
We develop this connection explicitly for a linear combination of unitaries (LCU) access model \cite{Babbush_2018,Low_2024}, in which the error of the polynomial approximation, the total error of encoding the input matrix as an LCU, and the precision of the parameters of the QSP sequence implementing the polynomial compete for a common total error budget.
For any given resource metric, such as $T$-gate counts, we define the minimum implementation cost compatible with that budget, and propagating the operator-Lipschitz bounds through it yields lower and upper bounds in the resource cost itself.
Polynomial selection can therefore be based on an actual implementation cost, given in terms of physical resources for fault-tolerant quantum computation, rather than on approximation error, degree, or sensitivity alone.
We then place the same objective inside a search over polynomial coefficients, obtaining a resource-aware polynomial-design procedure for QSP/QSVT.

We illustrate the resulting design procedure using eigenvalue thresholding for an H$_6$ molecular Hamiltonian. 
We compare an analytical polynomial family derived from an error-function approximation~\cite{wan2022} with a resource-aware family optimized using the end-to-end criterion.
Across the tested hardware profiles, the resource-aware family achieves a lower operator-sensitivity measure while reducing the modeled $T$-gate count by $51.9\%$--$56.7\%$ relative to the analytical polynomial.

The paper is organized as follows. In Sec.~\ref{sec:robustness}, we formulate the QSP/QSVT error model and show how the operator-Lipschitz constant controls the amplification of perturbations in the encoded operator. 
In Sec.~\ref{sec:end-to-end-budget}, we use this result to develop an end-to-end, fault-tolerant error budget analysis, specialize the input error to an LCU access model, derive the corresponding resource bounds, formulate a framework for resource-aware polynomial design, and apply it to find low-cost polynomials for eigenvalue thresholding for an H$_6$ molecular Hamiltonian.
In Sec.~\ref{sec:computing-Lstar}, we show how $L_\star[f]$ can be computed directly from the coefficients of a polynomial and develop an active-set algorithm that computes certified lower and upper bounds on it. Sec.~\ref{sec:prior} compares our robustness bound with previous robustness bounds.
Finally, Sec.~\ref{sec:conclusion} summarizes the implications for robustness and resource-aware QSP/QSVT design.


\section{Robustness of QSVT}\label{sec:robustness}


In this section we present a detailed robustness analysis for a fault-tolerant quantum implementation of an operator function $f_{\mathrm{T}}(A)$ using QSVT. 
We assume access to a faulty input matrix $B$, and to a function $f$ (typically a polynomial) that approximates $f_{\mathrm{T}}$ in its domain. The robustness of the implemented function $f(B)$ is measured by the spectral norm error with respect to the ideal $f_{\mathrm{T}}(A)$, namely the implementation is $\varepsilon$-robust if
\begin{align}\label{eq:qsvt-robustness-objective}
    \norm{f_\text{T}(A)-f(B)}\leq \varepsilon\,.
\end{align} 
In order to properly discuss the sources of input errors, in
Sec.~\ref{sec:robustness-gen} we first recall the different input access models for QSVT. 
Sec.~\ref{subsec:error-decomposition} then presents a detailed breakdown of the possible sources of error. In particular, in addition to input and function approximation errors, there is a further error contribution that is specific to QSVT: approximate synthesis of single-qubit gates into the discrete fault-tolerant gate set implementable by the hardware. 
Sec.~\ref{sec:robustness-qsvt} then uses the operator-Lipschitz constant to combine these errors and provide a complete error robustness analysis for the QSVT implementation. 
For simplicity, we focus here on normal input matrices, where applying a function to an operator reduces to transforming its eigenvalues. Nevertheless, all of our results extend to singular-value transformations of general operators, as implemented by QSVT (see App.~\ref{app:ext-QSVT} for details).



\subsection{Oracle models}\label{sec:robustness-gen}

QSP and QSVT work under the assumption that the input operator $A$ is encoded into a unitary operator (the oracle). 
Two types of oracles have been considered in the literature: projected unitaries and time-evolution oracles.
For Hermitian operators, the projected unitary oracle reduces to the so-called block-encoding $U_H$, which is a unitary containing the input operator $H\in \mathds{C}^{d\times d}$ as one of its matrix blocks. More specifically, $U_H$ acts on a larger space $\mathds{C}^{d}\otimes \mathds{C}^a$ and satisfies $\norm{\bra{0}U_H\ket{0}-H}\leq \delta$, where $\ket{0}\in \mathds{C}^a$ is a reference state, for some dimension $a$ and some precision $\delta$. Common techniques for building $U_H$ include linear combinations of unitaries (LCU) and sparse access \cite{Low_Chuang_2019,Gilyen_2019,Babbush_2018,Low_2024}.
In turn, the time evolution oracle $U_{H,t}$ is defined as the action of the evolution $e^{-\ii Ht}$ generated by $H$ at time $t$ controlled by an auxiliary qubit. $U_{H,t}$ can be implemented, for instance, using product formulas \cite{dong2022ground}.

The oracle model also determines the method of function approximation to be used. When accessing a block-encoding oracle, Chebyshev series approximation and the QSVT framework \cite{Gilyen_2019} are the natural choice, whereas for a time-evolution oracle it is more natural to work with Fourier expansions \cite{silva2022}. In both cases, the number of calls to the oracle and the function approximation error are determined by the number of terms in the corresponding basis expansion. 

To unify the discussion regardless of the access model, we denote the ideal input by 
\begin{align}
\label{eq:oracles}
    A\coloneq\begin{cases}
        H, &\text{for block-encoding oracle}\\
        e^{-\ii Ht}, & \text{for time-evolution oracle.}
    \end{cases}
\end{align}
On the other hand, the implemented oracle realizes an operator $B$ that is a noisy version of $A$ with (spectral norm) error $\norm{A-B}$. 
Here the spectra of $A$ and $B$ are assumed to lie in the same domain. Although the goal is usually a function of $H$ itself, it is often more convenient to take the unitary evolution as the input for robustness, since in this access model the error is typically controlled at the level of the time-evolution operator.

\subsection{Implementation error breakdown}\label{subsec:error-decomposition}


Fault-tolerant implementations of $f_\text{T}(A)$ are of course only approximate. We identify three independent sources of algorithmic errors (disregarding hardware noise):

\begin{enumerate}
    \item[$(i)$] function approximation error $\varepsilon_{\text{app}}\coloneq\norm{f_\text{T}-f}_{\infty,\mathcal{D}}$;
    \item[$(ii)$] input errors $\varepsilon_{\text{input}}\coloneq\norm{A-B}$;
    \item[$(iii)$] single-qubit gate synthesis errors $\varepsilon_{\rm{synth}}$.
\end{enumerate}

The uniform approximation error $\varepsilon_{\text{app}}$ depends only on the choice of approximating function $f$. 
For the two oracle models of Sec.~\ref{sec:robustness-gen}, $f$ is either a Chebyshev or Fourier series of a fixed degree. If the relevant domain is $\mathcal{D}$, the approximation error is measured by the $\ell_\infty$-norm over $\mathcal{D}$, namely $\norm{f_\text{T}-f}_{\infty,\mathcal{D}} = \sup_{z\in\mathcal{D}}\vert f_\text{T}(z)-f(z)\vert$, and can be reduced by increasing the degree of $f$. 

The input error $\varepsilon_{\text{input}}$ depends on the oracle access model. Although they are intended to perfectly encode $H$, in practice the oracles themselves suffer from imperfections even in the absence of hardware noise. 
For instance, in the time-evolution oracle model, the oracle can be implemented using standard product formulas and the input error receives two contributions,  $\varepsilon_{\text{input}}=\varepsilon_{\text{Trotter}}+\varepsilon_{\text{comp}}$, where $\varepsilon_{\text{Trotter}}$ refers to Trotterization errors and $\varepsilon_{\text{comp}}$ to errors due to compilation of each Trotter term into the discrete fault-tolerant gate set. 
For the block-encoding model based on the linear combination of unitaries (LCU) we have a similar input error decomposition (see App.~\ref{subsec:lcu-input-error} for the complete analysis). Namely, in LCU the input matrix $A$ is of the form $A = \frac{1}{\alpha}\sum_{j\in[M]} a_j U_j$ for unitaries $U_j$, and the block-encoding is built using two subroutines $\textsc{prepare}$ and $\textsc{select}$: the former encodes the coefficients $a_j$ into a quantum register, while the latter applies each $U_j$ controlled by the computational basis state corresponding to $a_j$ in that encoding. In this case, $\varepsilon_{\text{input}}=\varepsilon_{\text{trunc}}+\varepsilon_{\text{fp}}$ where the truncation error $\varepsilon_{\text{trunc}}$ is due to discarding some of the $a_j$ in the LCU while $\varepsilon_{\text{fp}}$ accounts for finite precision errors in the $\textsc{prepare}$ subroutine. 

Finally, the synthesis errors $\varepsilon_{\rm{synth}}$ arise because the QSP/QSVT sequence interleaves calls to the oracle with arbitrary single-qubit rotations which, in a fault-tolerant setting, must be compiled into a discrete gate set (\emph{e.g.}, Clifford$+T$) with finite accuracy. For a polynomial $f$ of degree $N$ and arbitrary parity, there is a sequence of single-qubit rotations of length up to $S=2N+1$ that exactly implements $f(B)$ in the standard QSVT circuit \cite{Gilyen_2019}. Denoting by $\widetilde{f}(B)$ the actual polynomial implemented by the compiled circuit, and assuming that each compiled single-qubit rotation differs from its ideal gate by at most $\varepsilon_{\rm{synth}}$ in operator norm, 
the error after $S$ such steps is $\big\|{f(B)-\widetilde{f}(B)\big\|} \leq S\,\varepsilon_{\rm{synth}}$.

\subsection{Operator Lipschitzness and QSVT robustness}\label{sec:robustness-qsvt}

Here we recall the notion of operator-Lipschitz functions \cite{aleksandrov2016} and use it to state our QSVT robustness bounds. 

Given a $d$-dimensional Hilbert space $\mathcal{H}_d$ and a compact domain $\mathcal{D}\subset\mathbb{C}$, we denote by $\LD$ the space of bounded and normal linear operators acting on $\mathcal{H}_d$ whose eigenvalue spectra lie in the domain $\mathcal{D}$. 
Here we will be interested only in two special cases relevant for Chebyshev- and Fourier-based QSP: respectively, normalized Hermitian operators $A=H$, in which case $\mathcal{D}=[-1,1]$ such that $\norm{H}\le1$; and unitary operators $A=e^{-\ii Ht}$, in which case $\mathcal{D}=\mathbb{T}\coloneq\{z\in\mathbb{C}\text{ s.t. } \vert z\vert=1\}$ (the unit circle). 
An operator function $g$ is operator-Lipschitz (OL) on the domain $\mathcal{D}$ if there exists $L>0$ such that, for all $A,B\in\LD$ and any dimension $d$,
\begin{align}\label{eq:op-lip}
\norm{g(A)-g(B)}\leq L\,\norm{A-B}\,.
\end{align}
The operator-Lipschitz constant of $g$ in $\mathcal{D}$ is the minimal such constant $L$ over all dimensions $d$, 
\begin{align}\label{eq:Lstardef}
 L_\star[g] \coloneq \sup_{d\geq 2}\,\,\,\sup_{\substack{A,B\,\in\LD\\A\ne B}} \frac{\|g(A)-g(B)\|}{\|A-B\|}\,.
\end{align}

For commuting $A,B$, the operator relation in Eq.~\eqref{eq:op-lip} reduces to standard Lipschitz-continuity of the scalar function $g(z)$ for $z\in\mathcal{D}$, 
implying that $L_\star[g] \ge \ell_\star[g]$, where $\ell_\star[g]$ is the scalar Lipschitz constant of $g$ in $\mathcal{D}$. 
While the two constants always coincide in the Frobenius norm \cite{kitaneh1985}, in the operator norm, scalar Lipschitzness does not even imply operator Lipschitzness (\emph{e.g.} $|x|$ is Lipschitz but not OL \cite{Kato1973}). Although the equality holds for many target functions of interest for QSVT applications (see Table~\ref{tab:lipschitz_examples}), this is rather the exception. In fact, even smooth functions like low-degree polynomials can display this difference \cite{Aleksandrov2020}.  In App. \ref{app:criteria} we show an operational sufficient criterion for $g$ to test if $L_\star[g] \neq \ell_\star[g]$. For functions of Hermitian operators, the best-known upper bound relating $\ell_\star[g]$ and perturbations to the operator function is (Ref.~\cite[Thm.~11.4]{aleksandrov2011}) 
\begin{align}\label{eq:op-lip-bound}
\norm{g(A)-g(B)} \le C(1+\log(d))\,\ell_\star[g]\,\norm{A-B}\,,
\end{align}
which depends on the dimension $d$ of the input operator and carries an unknown absolute constant $C>0$. 

The key fact here is that the polynomial functions implemented via Chebyshev- or Fourier-based QSP are always operator-Lipschitz (see App.~\ref{app:trivial_bounds_scalar} for the proof as well as for simple upper bounds on $L_\star$). 
This allows us to derive the following tight robustness bound in terms of the optimal OL constant $L_\star$. 
The proof follows simply from triangle inequalities and is presented in App.~\ref{app:proof_main}.

\begin{proposition}[QSVT robustness bound]\label{prop:robustness}
    Let $\mathcal{D}$ be either $[-1,1]$ or the unit circle $\mathbb{T}$, $\mathcal{H}_{2^n}$ a $n$-qubit Hilbert space, $f_\text{T}$ a target function, $f$ a degree-$N$ polynomial approximation of $f_\text{T}$ on $\mathcal{D}$, and $\widetilde{f}$ the polynomial actually implemented by the QSVT circuit after compilation into a discrete gate set. 
    Moreover, let $\varepsilon_{\rm{app}},\varepsilon_{\rm{input}},\varepsilon_{\rm{synth}}$ be the approximation, input, and synthesis errors as introduced in Sec.~\ref{subsec:error-decomposition}, and define
    \begin{align}
        L_\star \coloneq 
        \begin{cases}
            \min\{L_\star[f],\,L_\star[f_\text{T}]\}, & \text{if } f_\text{T} \text{ is OL on } \mathcal{D}\\
            L_\star[f], & \text{otherwise} \\
        \end{cases}\,,
    \end{align}
    and $S\coloneq2N+1$. 
    Then, for every $A,B\in\mathcal{L}_\mathcal{D}(\mathcal{H}_{2^n})$, 
    \begin{align}\label{eq:robustness_bound}
        \big\|{f_\text{T}(A)-\widetilde{f}(B)\big\|}\le \varepsilon_{\rm{app}} +L_\star\,\varepsilon_{\rm{input}} +S\,\varepsilon_{\rm{synth}}\,.
    \end{align}
\end{proposition}


Although we focus our discussion on eigenvalue transformation, a simple adaptation of the proof in Ref.~\cite[Cor.~21]{Gilyen_2018_arxiv} shows that, for functions implementable by QSVT, $\varepsilon$-robustness in eigenvalue transformation of Hermitian matrices implies $\varepsilon$-robustness also in singular value transformation of arbitrary matrices. In particular, for operator-Lipschitz functions, their induced singular value transformation are also Lipschitz with the same constant, as stated below.  
 \begin{lemma}[Extension to QSVT]\label{lem:ext-QSVT} Let $A,B\in\mathds{C}^{d^\prime \times {d}}$ be rectangular matrices such that $\norm{A},\norm{B}\leq 1$. Given a target function $f_T$ and a polynomial approximation $f$ both with definite parity, denote by $f_T^{(SV)}(A)$ the singular value transformation of $A$ by $f_T$ as implemented by QSVT, and similarly for $f^{(SV)}(A)$. Then, Proposition \ref{prop:robustness} is satisfied for $f_T^{(SV)}$ with the same Lipschitz constant.

\end{lemma}

\noindent The proof appears in App.~\ref{app:ext-QSVT} and is based on constructing a Hermitian dilation whose eigenvalue transformation is related to the singular value transformation of $A$. Lemma~\ref{lem:ext-QSVT} allows us to restrict ourselves to the Hermitian case without loss of generality.

Moreover, Eq.~\eqref{eq:robustness_bound} improves prior QSVT robustness bounds from \cite{Gilyen_2019,Chia_2020}.
The best such bound has been obtained for degree-$N$ polynomial functions, $\norm{f(A)-f(B)} \leq 4N\sqrt{\norm{A-B}}$ \cite[Lemma 22]{Gilyen_2018_arxiv}, which  
has a quadratically worse dependence in the input error $\varepsilon_{\mathrm{input}}=\norm{A-B}$ compared to Eq.~\eqref{eq:robustness_bound}. 
In Sec.~\ref{sec:prior} we present a thorough comparison with other prior bounds.
Notice also that the 
$N$ dependence is replaced by $L_\star$, which typically gives rise to further improvement.

Proposition~\ref{prop:robustness} provides two decompositions of the QSVT approximation error, depending on whether the perturbation is propagated through the target function's operator-Lipschitz constant or through the polynomial approximation actually implemented. An analytical expression for $L_\star[f_\text{T}]$, if it exists, can quantify the robustness of the implementation without a polynomial-specific calculation.
 However, when $L_\star[f_\text{T}]$ is not available, or the target function is not even OL, the polynomial approximation $f$ remains operator-Lipschitz in the bounded domains considered, and characterizing its OL constant $L_\star[f]$ suffices to obtain a tighter input-error bound. We could, in principle, use the upper bound in Eq.~\eqref{eq:op-lip-bound} to obtain the robustness bound 
 $\varepsilon_{\rm{app}} +C(1+\log(d))\ell_\star\,\varepsilon_{\rm{input}} +S\,\varepsilon_{\rm{synth}}$
 that bypasses the need to compute the OL constant $L_\star$ but has the disadvantage of depending on the Hilbert space dimension. To overcome this, in Sec.~\ref{sec:computing-Lstar} we devise an efficient numerical optimization method to compute tight lower and upper bounds on $L_\star[f]$ for arbitrary fixed-degree polynomials.

Notwithstanding, selecting the polynomial implemented by the QSP/QSVT sequence fixes several quantities that determine the final quantum resource budget, as it will become clear in Sec.~\ref{sec:end-to-end-budget} when we analyze a particular oracle construction. Its approximation quality determines $\varepsilon_{\rm{app}}$, while its degree determines the number of queries to the oracle and the number of single-qubit rotations to be compiled in an eventual fault-tolerant setting. The polynomial also carries the OL constant $L_\star[f]$, which governs the amplification of oracle errors. In principle, distinct polynomials can have comparable approximation error and degree while imposing different requirements on the encoded input through different OL constants. This motivates us to build a polynomial approximation method presented in Sec.~\ref{sec:res-aware-poly-sel} that accounts for input-error sensitivity in addition to uniform approximation errors.


\begin{table*}[t]
\centering
\renewcommand{\arraystretch}{1.6}
\setlength{\tabcolsep}{4pt}
\begin{tabular}{c|ccccccccccc}
\hline
$f_\text{T}$ &
$H^m$ &
$T_m(H)$ &
$e^{-itH}$ &
$e^{-\beta H}$ &
$\text{erf}(\beta H)$ &
$(\mu-H)^{-1}$ &
$\log(\mu-H)$ &
$\arcsin H^\prime$ &
$A^{-1}$ &
$U^{\pm m}$ &
$(\nu-U)^{-1}$
\\ \hline
$L_\star=\ell_\star$ &
$m$ &
$m^2$ &
$\vert t\vert$ &
$\beta\,e^\beta$ &
$2\beta/\sqrt{\pi}$ &
$\vert\text{Im}(\mu)\vert^{-2}$ &
$\vert\text{Im}(\mu)\vert^{-1}$ &
$(1-\delta^2)^{-1/2}$ &
$\kappa^2$ &
$m$ &
$(\vert \nu\vert-1)^{-2}$
\\ \hline
\end{tabular}
\caption{{\bf Examples of target functions $f_{\text{T}}$ with $L_\star=\ell_\star$.} Here $H$ and $H^\prime$ denote Hermitian matrices with spectra respectively in $[-1,1]$ and $[-\delta,\delta]$ for $\delta\in(0,1)$; $A\in\mathbb{C}^{d\times d}$ has singular values in $[1/\kappa,1]$ for $\kappa>1$; and $U$ is unitary. The parameter ranges are $m\in\mathbb{Z}^+$
, $t\in\mathbb{R}$, $\beta\in\mathbb{R^+}$, $\mu\in\mathbb{C}\setminus\mathbb{R}$, and $\nu\in\mathbb{C}\setminus\mathbb{T}$. See App.~\ref{app:trivial_bounds_scalar} for details.}
\label{tab:lipschitz_examples}
\end{table*}

\section{End-to-end robustness analysis in a fault-tolerant scenario}\label{sec:end-to-end-budget}


Proposition~\ref{prop:robustness} provides a general robustness statement for a fault-tolerant operator function implementation based on QSVT. The input and synthesis errors, $\varepsilon_{\rm{input}}$ and $\varepsilon_{\rm{synth}}$, can be freely chosen, while the approximation error $\varepsilon_{\rm{app}}$ and QSVT sequence size $S=2N+1$ are readily computed from the polynomial $f$. 
However, as mentioned in Sec.~\ref{sec:robustness-qsvt}, computing the OL constant $L_\star$ of the polynomial $f$ exactly is highly non-trivial and it is often easier to compute an upper bound 
\begin{align}
L_{\mathrm{up}}\geq L_\star[f]    \,.
\end{align}
This will be precisely the focus of Sec.~\ref{sec:computing-Lstar} below, and here we assume such an $L_{\mathrm{up}}$ has already been obtained. 
Consequently, the implementation is guaranteed to achieve total error at most $\varepsilon$ whenever
\begin{align}\label{eq:end-to-end-budget}
    \varepsilon_{\rm{app}} +L_{\mathrm{up}}\,\varepsilon_{\rm{input}} +S\,\varepsilon_{\rm{synth}} \leq \varepsilon.
\end{align}
Importantly, the way in which the total error $\varepsilon$ is allocated among each individual source is a design choice, and different allocations lead to polynomials with different implementation costs. 
For instance, for a fixed degree $N$, a steeper polynomial can reduce $\varepsilon_{\rm{app}}$ at the price of a larger input error amplification due to larger $L_{\mathrm{up}}$. 
Motivated by this, in the following we 
turn this design choice into a well-posed optimization that seeks an error allocation that minimizes the total resource cost (here the $T$ gate count).

\subsection{Fault-tolerant resource bounds}\label{sec:resource-bounds}
Let us assume that the hardware architecture, the oracle access model, and the oracle implementation scheme are fixed. 
The goal is an end-to-end implementation of a polynomial $f$ using the QSVT algorithm. Let 
\begin{align}\label{eq:ft-res}
R \left( f, \varepsilon_{\rm{input}}, \varepsilon_{\rm{synth}} \right)
\end{align}
be the resource cost of an implementation that accepts an 
input error $\varepsilon_{\mathrm{input}}$ 
and per-query synthesis error
$\varepsilon_{\rm{synth}}$. 
We will be mostly interested in $R\left( f, \varepsilon_{\rm{input}}, \varepsilon_{\rm{synth}} \right)$ representing the total number of $T$ gates, but the general construction below extends to any other resource metric (\emph{e.g.}, number of qubits, total runtime, circuit depth, etc.). 
In the following we assume that the functional form of $R \left( f, \varepsilon_{\rm{input}}, \varepsilon_{\rm{synth}} \right)$ is known. 
As an illustration, in App.~\ref{sec:lcu-res-model} we model the full functional form of the $T$-gate count using the LCU access model of Sec.~\ref{sec:end-to-end-budget} with the \textsc{prepare} oracle of Ref.~\cite{Low_2024}. The full expression is presented in App.~\ref{sec:lcu-res-model} and will be the one used for applications in Sec.~\ref{sec:h6-bench}. 

For a polynomial $f$ and any bound $L$ (upper or lower) on its OL constant $L_\star[f]$, we define the minimum resource cost compatible with $L$ and the total error $\varepsilon$ as 
\begin{gather}\label{eq:min-cost}
\begin{aligned}
\mathcal{C}[f,L,\varepsilon]
\coloneq
&\inf_{\substack{ \varepsilon_{\rm{input}}\geq 0,\\ \varepsilon_{\rm{synth}}\geq 0 }}  \quad R \left( f, \varepsilon_{\rm{input}}, \varepsilon_{\rm{synth}} \right) \\
&\textrm{s.t.}  \,\, \varepsilon_{\rm{app}}[f] +
L\,\varepsilon_{\rm{input}}+ S[f]\varepsilon_{\rm{synth}} \leq \varepsilon .
\end{aligned}
\end{gather}
Notice that, although the choice $\varepsilon_{\mathrm{input}}=\varepsilon_{\mathrm{synth}}=0$ trivially satisfies the inequality constraint and therefore is a feasible point, it may not minimize the resource cost. In fact, we show in App.~\ref{sec:lcu-res-model} that the $T$-gate count for a LCU-based implementation scales logarithmically with the inverse precisions $\varepsilon_{\mathrm{input}}$ and $\varepsilon_{\mathrm{synth}}$.

Since increasing $L$ in Eq.~\eqref{eq:min-cost} only makes the constraint stricter without affecting the objective, it follows that $\mathcal{C}$ is a monotonically non-decreasing function of $L$. In particular, if $L_\mathrm{low}[f]\le L_\star[f]\le L_\mathrm{up}[f]$ are lower and upper bounds for the OL constant of $f$, this implies
\begin{align}\label{eq:cost-interval}
\mathcal{C} \big[f,L_{\mathrm{low}}[f],\varepsilon\big]
\leq
\mathcal{C}\big[f,L_\star[f],\varepsilon\big]
\leq
\mathcal{C}\big[f,L_{\mathrm{up}}[f],\varepsilon\big]\,.
\end{align}
Therefore, whenever the exact OL constant $L_\star[f]$ is not known exactly, Eq.~\eqref{eq:cost-interval} provides lower and upper bounds on the total cost of implementing a polynomial $f$ based on computable lower ($L_{\mathrm{low}}[f]$) and upper ($L_{\mathrm{up}}[f]$) bounds for $L_\star[f]$ itself. 
Notice that, since the resource functions may depend nonlinearly and even discontinuously on the precisions, a small gap in the OL bounds does not necessarily imply a small gap in resource cost. 
This motivates selecting the polynomial only after propagating its approximation error, degree, and operator-Lipschitz bound through the complete resource model.


\subsection{Resource-aware polynomial selection and design}\label{sec:res-aware-poly-sel}

The cost interval in Eq.~\eqref{eq:cost-interval} turns the end-to-end budget into a criterion for choosing which polynomial to implement. Standard polynomial constructions aim at minimizing the approximation error at a given degree, as is the case, \emph{e.g.}, for the various approximants of $x^{-1}$ used in matrix inversion \cite{childs2017,gribling2021,sunderhauf2025} or for the error-function-based approximants of the Heaviside function used in eigenvalue thresholding \cite{wan2022}. 
Under a fixed total error $\varepsilon$, however, the cheapest implementation need not be the one with the smallest $\varepsilon_{\rm app}$ or the smallest degree $N$. As it is clear from Eq.~\eqref{eq:end-to-end-budget}, the polynomial also dictates via the OL constant $L_\star[f]$ the amount by which the input error gets amplified. This is particularly relevant when the target function has no known OL constant, or is not OL at all, as is the case for the Heaviside function on a domain containing its discontinuity. 
In the following we show how to use the cost $\mathcal{C}[f,L_{\rm up}[f],\varepsilon]$ to select the cheapest polynomial from a given list of candidates, $\mathcal{P} = \left\{f_1,\ldots,f_J \right\}$, and leverage this to design resource-efficient polynomials of fixed degree by searching directly in coefficient space.


Selecting a resource-optimal polynomial from this list goes as follows: 
we apply Algorithm~\ref{alg:active-set} independently to each $f$ to obtain certified bounds on its OL constant.
Eq.~\eqref{eq:cost-interval} then uses these bounds and associates to every candidate a resource cost interval within the error model of Eq.~\eqref{eq:end-to-end-budget}, where the optimization over $\varepsilon_{\rm{input}}$ and $\varepsilon_{\rm{synth}}$ is done separately for every candidate. Finally, the resource-optimal polynomial compatible with the OL upper bounds in the set of candidates $\mathcal{P}$ is then
\begin{align}\label{eq:opt-poly}
f_\star^{(\mathcal{P})} \coloneq \arg\min_{f\in\mathcal P} \mathcal{C}\big[f,L_{\mathrm{up}}[f],\varepsilon\big].
\end{align}

The selection in Eq.~\eqref{eq:opt-poly} assumes that the list of candidate polynomials $\mathcal{P}$ has already been constructed. 
We now present a similar construction to search for resource-efficient approximating polynomials of a fixed degree $N$ directly in the space of allowed coefficients $\mathbf{a}=\{a_t\}_{t\in[N]}$ in a given basis $\{\Phi_t\}_{t=0}^{N}$ of degree-$N$ polynomials. Any candidate polynomial can be written as
\begin{align}\label{eq:designed-poly}
    f_{\mathbf{a}}(x) = \sum_{t=0}^{N} a_t\,\Phi_t(x), \quad \mathbf{a} \in \mathcal{A}_{N,\varepsilon}\,,
\end{align}
where $\mathcal{A}_{N,\varepsilon}$ denotes the set of all possible coefficient vectors such that $i)$ $f_{\mathbf{a}}$ is a feasible QSP/QSVT polynomial (\emph{i.e.}, it satisfies the necessary parity and normalization conditions); $ii)$ $\norm{f_{\mathrm{T}}-f_{\mathbf{a}}}\le\varepsilon$ for $f_{\mathrm{T}}$ the target function. 
For fixed $N$, the QSP sequence size is constant: $S[f_{\mathbf{a}}]=N$ or $S[f_{\mathbf{a}}]=2N+1$ for $f_{\mathbf{a}}$ of definite parity or not, respectively. 
Therefore, varying the coefficient vector $\mathbf{a} \in \mathcal{A}_{N,\varepsilon}$ changes only the approximation error and the OL contributions to the total error in Eq.~\eqref{eq:min-cost}. 

For each candidate $f_{\mathbf{a}}$, Alg.~\ref{alg:active-set} outputs bounds 
on $L_\star[f_{\mathbf{a}}]$. 
Using the upper bound $L=L_{\mathrm{up}}[f_{\mathbf{a}}]$ in the cost \eqref{eq:min-cost}, one can formalize the search for the cost-optimal approximating polynomial of degree $N$ as 
\begin{align}\label{eq:design-objective}
\min_{ \mathbf{a} \in \mathcal{A}_{N,\varepsilon}}\mathcal{C}\!\big[f_{\mathbf{a}},L_{\mathrm{up}}[f_{\mathbf{a}}],\varepsilon\big].
\end{align}
In Alg.~\ref{alg:res-aware-poly-design} below, we present an heuristic that finds an surrogates solutions to this minimization problem. 
At each step, the algorithm calls a subroutine, here denoted $\mathsf{Search}_N$, that searches a new candidate in the coefficient space $\mathcal A_{N,\varepsilon}$, \emph{i.e.}, it takes a coefficient vector $\mathbf{a}$ and proposes a new one $\mathbf{a}^\prime$. 
In our implementation (see App.~\ref{sec:searchN-details} for details), $\mathsf{Search}_N$ receives a polynomial satisfying the $\varepsilon$-condition $\varepsilon_{\mathrm{app}}[f_{\mathbf{a}}]+L_{\rm{up}}[f_{\mathbf{a}}]\left(\varepsilon_{\rm{trunc}}+\varepsilon_{\rm{fp}}\right)+S[f_{\mathbf{a}}]\varepsilon_{\rm{synth}}\le\varepsilon$ and iteratively modifies its coefficients so as to find a new polynomial $f_{\mathbf{a}^\prime}$ having a smaller OL upper bound, $L_{\mathrm{up}}[f_{\mathbf{a}^\prime}]\le L_{\mathrm{up}}[f_{\mathbf{a}}]$, while still satisfying the $\varepsilon$-condition. 
Alg.~\ref{alg:res-aware-poly-design} then iterates this procedure to produce a list $\mathcal{P}_N$ of candidate degree-$N$ polynomials having different approximations errors and OL constants, from which the optimal cost degree-$N$ polynomial $f_{\star,N}$ is selected by comparing their total costs as in Eq.~\eqref{eq:opt-poly}.

\begin{algo}\label{alg:res-aware-poly-design}
\quad\\
\quad\\
\indent Input: errors $\varepsilon_{\mathrm{input}},\varepsilon_{\mathrm{synth}},\varepsilon>0$, degree $N$, tolerance $\xi>0$, and a coefficient-search subroutine $\mathsf{Search}_N$ over space $\mathcal A_{N,\varepsilon}$ in Eq.~\eqref{eq:designed-poly}
\\
\indent Output: the optimal cost polynomial $f_\star$ and its cost

\begin{enumerate}[leftmargin=*]
    \item Initialize a coefficient vector $\mathbf{a} \in \mathcal{A}_{N,\varepsilon}$ and the list $\mathcal{P}_N$ of candidate degree-$N$ polynomials as $\mathcal{P}_N=\{f_{\mathbf{a}}\}$;

    \item Use $\mathsf{Search}_N$ to propose a new coefficient vector $\mathbf{a}^\prime$, build the corresponding polynomial $f_{\mathbf{a}^\prime}$ in Eq.~\eqref{eq:designed-poly}, and compute $\varepsilon_{\mathrm{app}}[f_{\mathbf{a}^\prime}],L_{\mathrm{up}}[f_{\mathbf{a}^\prime}]$. \\
    If $\varepsilon_{\mathrm{app}}[f_{\mathbf{a}^\prime}]+L_{\mathrm{up}}[f_{\mathbf{a}^\prime}]\varepsilon_{\mathrm{input}}+S[f_{\mathbf{a}}]\varepsilon_{\mathrm{synth}} \le\varepsilon$, add $f_{\mathbf{a}^\prime}$ to the list $\mathcal{P}_N$; 

    \item If $\vert\varepsilon_{\mathrm{app}}[f_{\mathbf{a}^\prime}]-\varepsilon_{\mathrm{app}}[f_{\mathbf{a}}]\vert\le\xi$ and $\vert L_{\mathrm{up}}[f_{\mathbf{a}^\prime}]-L_{\mathrm{up}}[f_{\mathbf{a}}]\vert\le\xi$ (\emph{i.e.}, $f_{\mathbf{a}^\prime}$ is not significantly different from $f_{\mathbf{a}}$), go to Step~4. 
    Else, set $\mathbf{a}=\mathbf{a}^\prime$ and return to Step 2;

    \item Select from $\mathcal{P}_N$ the polynomial with minimal cost, $f_\star^{(\mathcal{P}_N)}$, using Eq.~\eqref{eq:opt-poly}. 
    Return $f_{\star,N}\coloneq f_\star^{(\mathcal{P}_N)}$, its properties $\varepsilon_{\mathrm{app}}[f_{\star,N}],L_{\mathrm{up}}[f_{\star,N}]$, and the cost $\mathcal{C}\big[f_{\star,N},L_\mathrm{up}[f_{\star,N}],\varepsilon\big]$). 
\end{enumerate}
\end{algo}

The objective function $\mathcal{C}$ in Eq.~\eqref{eq:min-cost} need not be smooth. For instance, the $T$-gate count resource model in App.~\ref{sec:lcu-res-model} contains ceilings and other discretizations that are not smooth. Therefore, we place no global-convergence requirement on $\mathsf{Search}_N$. The search is restricted to $\mathcal{A}_{N,\varepsilon}$ and explores only a finite set of fully certified points $\mathcal{P}_N$. Consequently, $f_{\star,N}$ is the lowest-resource polynomial found by the specified search at degree $N$, rather than a global minimizer over $\mathcal{A}_{N,\varepsilon}$. 
Finally, the procedure can be repeated independently for a finite set of degrees $N$, with the global optimal obtained after a final cross-degree selection step at the end. 

Next, we apply this polynomial design to a concrete application.



\subsection{Application: molecular ground state energy estimation}\label{sec:h6-bench}

Here we illustrate our resource-aware polynomial design framework introduced in Sec. \ref{sec:res-aware-poly-sel} for the task of ground state energy estimation via eigenvalue thresholding \cite{Martyn_2021} for a molecular Hamiltonian.

We chose the H$_6$ Hamiltonian on $12$ qubits and $919$ Pauli terms, obtained from PennyLane~\cite{azad2023}. For eigenvalue thresholding, the target transformation is the Heaviside function
\begin{align}\label{eq:heaviside-domain}
    f_{\rm{T}}(x)=\Theta(x), \quad \mathcal{D}=[-1,1].
\end{align}
The polynomial is built to approximate $f_{\rm{T}}(x)$ up to $\varepsilon_{\rm{app}}$ in the gapped domain $D_{\nu}=[-1,-\nu]\cup[\nu,1]$ that excludes the discontinuity at $x=0$. We set $\nu = 0.015$, which is sufficient to resolve the energy gap $3.02\times10^{-2}$ between the ground and first excited states of H$_6$'s Hamiltonian (after normalization to $\norm{H}=1$). Since the eigenvalue thresholding protocol shifts the position of the discontinuity scanning across the whole spectral domain, the OL constant is calculated over the full domain $\mathcal{D}$.

\begin{figure*}[!ht]
    \centering
    \includegraphics[width=.96\textwidth]{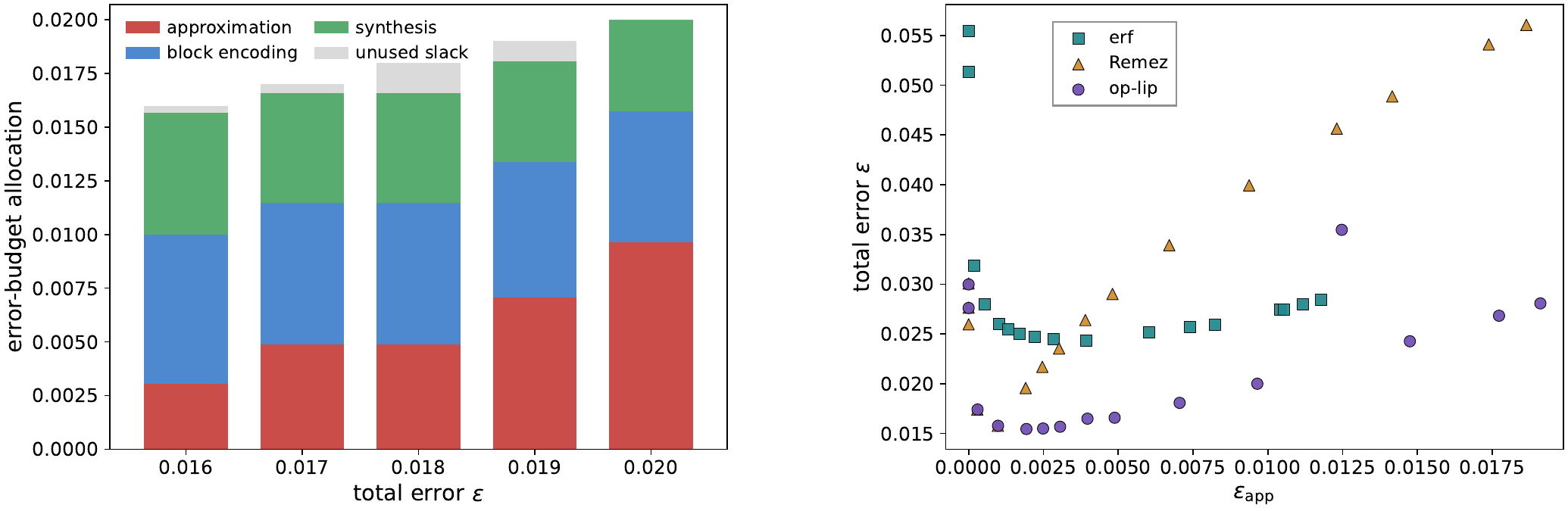}
    \caption{
\textbf{Fixed-hardware profiles at  $\varepsilon_{\rm fp}=2.44\times 10^{-4}$ and $\varepsilon_{\rm synth}=4.15\times 10^{-5}$.}
Left: error-budget decomposition of the minimum-$T$ polynomial over the feasible range $0.016 \le \varepsilon \le 0.020$, splitting the allowed total error into approximation, amplified block-encoding, synthesis, and unused slack.
Right: Pareto frontiers in the $(\varepsilon_{\rm{app}},\varepsilon)$ plane for the same fixed hardware profile, shown for the erf, Remez, and op-lip families.
Taken together, the two panels show why polynomial selection is not controlled by approximation error alone under fixed hardware constraints. A substantial part of the admissible error budget is consumed by implementation terms, and the op-lip family consequently occupies the most favorable approximation--total-error tradeoff region, which is why it supplies the minimum-$T$ feasible winner throughout this feasible range.
}
    \label{fig:h6-fixed-profile-budget-frontier}
\end{figure*}

We compare two polynomial families. First, as a baseline, we chose the analytical polynomial from Ref.~\cite{wan2022} obtained by first approximating $\Theta(x)$ by the error function $\mathrm{erf}(x)$, and then by finding an analytical approximating polynomial for $\mathrm{erf}(x)$. For the second family, we used the Remez exchange algorithm~\cite{remez1934}. The third family was designed following Sec.~\ref{sec:res-aware-poly-sel}, denoted as \emph{op-lip}, therefore it is not a single closed-form approximant and is not chosen by the approximation error alone, it is a family over degrees that explores the tradeoff between both $\varepsilon_{\rm{app}}[f_{\mathbf{a}}]$ and $L_{\mathrm{up}}[f_{\mathbf{a}}]$ while selecting the polynomial with the minimum the resource cost.

For simplicity, instead of allowing $\varepsilon_{\rm{trunc}},\varepsilon_{\rm{fp}}, \varepsilon_{\rm{synth}}$ to be optimized for every polynomial using Eq.~\eqref{eq:min-cost} with $\varepsilon_{\mathrm{input}}=\varepsilon_{\mathrm{trunc}}+\varepsilon_{\mathrm{fp}}$, we consider fixed hardware profiles. Each profile represents a
different machine setting: a limited number of auxiliary qubits constrains the precision attainable by the \textsc{prepare} oracle and therefore fixes $\varepsilon_{\rm{fp}}$, while a finite synthesis-depth budget constrains the precision of the compiled QSP rotations and fixes $\varepsilon_{\rm{synth}}$. In addition, we retain the complete Pauli expansion, so $\varepsilon_{\rm{trunc}}=0$.  
Thus, for a hardware profile $(\varepsilon_{\rm{fp}}, \varepsilon_{\rm{synth}})$, the optimization
of $\mathcal{C}[f,L_{\mathrm{up}}[f],\varepsilon]$ collapses to a feasibility test in Eq.~\eqref{eq:end-to-end-budget}. 

Thus, within a fixed profile, op-lip selects the feasible candidate with the smallest resource cost. Varying the profile then probes how the preferred polynomial changes across machines with different qubit and circuit-depth limitations. For each profile and polynomial family, we can also ask what is the smallest total-error budget $\varepsilon$ for which at least one polynomial satisfies Eq.~\eqref{eq:end-to-end-budget}.

Fig.~\ref{fig:h6-fixed-profile-budget-frontier} makes the fixed-profile mechanism explicit.  The left panel decomposes the error budget of the minimum-$T$ feasible candidate over $0.016 \leq \varepsilon \leq 0.020$, where we can see a substantial fraction of the available tolerance consumed by block-encoding and synthesis errors, so the scalar approximation error $\varepsilon_{\rm{app}}$ cannot by itself determine the preferred polynomial.  The right panel shows the corresponding Pareto frontiers in the $(\varepsilon_{\rm{app}},\varepsilon)$ plane.  Although erf and Remez can attain smaller $\varepsilon_{\rm{app}}$, their larger amplification of the fixed hardware errors shifts the required total error upward, whereas op-lip occupies the lower part of the attainable frontier over the accuracy range.  Consequently, the minimum-$T$ solution can deliberately tolerate a larger approximation error while spending less of the total budget on implementation error.  This illustrates the resource-aware character of the design.

\begin{figure*}[!ht]
    \centering
    \includegraphics[width=\textwidth]{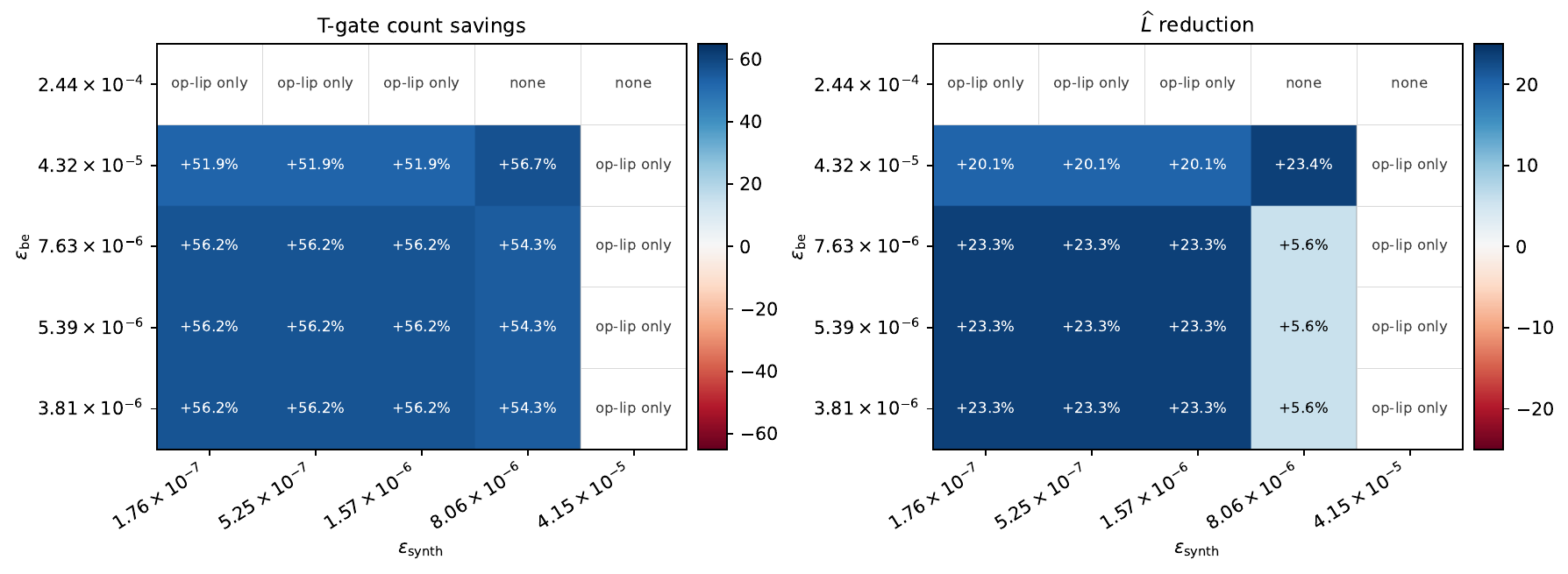}
\caption{
\textbf{Resource advantage and robustness advantage of op-lip over erf at fixed hardware profiles.}
For the H$_6$ Hamiltonian, each entry corresponds to a fixed hardware profile $(\varepsilon_{\rm{fp}},\varepsilon_{\rm{synth}})$ at the fixed total error budget $\varepsilon=0.010$.
The left panel shows the percentage reduction in the $T$-gate count (see Sec.~\ref{sec:lcu-res-model}) achieved by the op-lip family relative to the erf family.
The right panel shows the corresponding percentage reduction in the OL upper bound, computed for the polynomials used in the left panel using Algorithm~\ref{alg:active-set}.
Thus the right panel measures the robustness improvement attached to the very polynomials responsible for the $T$-count advantage on the left.
Entries labeled ``op-lip only'' indicate profiles for which only an op-lip polynomial is feasible, and ``none'' indicates when no family has a feasible polynomial.
}
    \label{fig:h6-erf-oplip-t-lhat-advantage}
\end{figure*}

The improvement can be further explained using Fig.~\ref{fig:h6-erf-oplip-t-lhat-advantage}. At fixed total error budget $\varepsilon=0.01$, the left panel shows the percentage $T$-gate savings, $1-(T_{\rm op\text{-}lip})/T_{\rm erf}$, of the op-lip relative to the corresponding erf  on profiles where both families are feasible, where the resources are counted as in Sec.~\ref{sec:lcu-res-model}. 
In the right panel, we have the corresponding reduction in the OL upper bound, $1-(L_{{\mathrm{up}}_{\rm{op-lip}}})/L_{{\mathrm{up}}_{\rm erf}}$.
Among the $16$ profiles where both families are feasible, op-lip yields  $T$-gate savings ranging from $51.9\%$ to $56.7\%$, and a seven profiles that are only feasible for op-lip. So, the designed polynomial is less sensitive to the implementation errors. Thus, the savings are the direct resource manifestation of the improved robustness seen in Fig.~\ref{fig:h6-erf-oplip-t-lhat-advantage}.

In conclusion, for this H$_6$ instance, the evidence shows that, under fixed hardware constraints, designing a polynomial to reduce the resource cost not only reduced the T-gate count, but also produced robust polynomials with lower OL constants.

\section{Bounding $L_\star[f]$ from polynomial coefficients}\label{sec:computing-Lstar}



In this section, we show our main technical results, namely the translation of results from operator-space theory literature to the QSP setting followed by an operational method to calculate lower and upper bounds on the operator Lipschitz constant $L_\star$ of polynomial functions. 

Although $L_\star[f]$ is formally defined by Eq.~\eqref{eq:Lstardef}, the definition is not straightforward to operationalize since the supremum ranges over a continuum of non-commuting matrix pairs $A$ and $B$. Although a finite $\varepsilon$-net or brute force random search can probe this space and produce \emph{lower} bounds on $L_\star[f]$, no finite collection of tested pairs can by itself certify that a larger ratio does not occur elsewhere. A more challenging direction is therefore the construction of a computable \emph{upper} bound $L_{\mathrm{up}}\geq L_\star[f]$ 
that is valid uniformly over all admissible $A,B\in\LD$ and is sufficiently tight to be useful in the robustness estimates below. Here, we show that when $f$ is a polynomial, this problem reduces to finite-dimensional optimization directly in terms of its polynomial coefficients. Our contribution is an exact finite coefficient-space formulation for polynomial kernels (Lemma~\ref{lem:semiinfinite_reformulation}) and an exchange algorithm with certified two-sided bounds in Sec.~\ref{subsec:active-set-alg}.

\subsection{Haagerup factorization and semi-infinite optimization}

The operator Lipschitz constant of a differentiable function $f$ is related to its divided difference kernel $K_f:\mathcal{D}\times\mathcal{D}\rightarrow\mathbb{C}$ defined as 
\begin{align}\label{eq:df-kernel}
K_f(x,y) \coloneq
\begin{cases}
    \frac{f(x)-f(y)}{x-y}, &(x \neq y ) \\
    f^\prime(x),  &(x = y)
\end{cases}.
\end{align}  
In App.~\ref{app:proof-lem-lips-schur}, we give an alternative proof of the following lemma, which compiles various results in Ref.~\cite{aleksandrov2016}.
\begin{lemma}[$L_\star$ as a Schur multiplier norm] \label{lem:lips-schur}
The operator-Lipschitz constant of a differentiable function $f$ in a domain $\mathcal{D}=[-1,1]$ or $\mathcal{D}=\mathbb{T}$ is equal to the Schur multiplier norm of its divided difference kernel in that domain, i.e
    \begin{align}\label{eq:Schur}
L_\star[f] & =\norm{K_f}_{\rm{Schur}} \coloneq \sup_{d\ge2}\,\sup_{\substack{\mathsf{X},\mathsf{Y}\in\mathcal{D}^d\\ Z\in\mathbb{C}^{d\times d}\\ \norm{Z}=1}}\big\|K_{f,d}(\mathsf{X},\mathsf{Y})\circ Z\big\|\,,
\end{align}
where $\circ$ denotes the Hadamard (element-wise) product, $K_{f,d}(\mathsf{X},\mathsf{Y})$ is the $d\times d$ matrix with components $K_f(x_i,y_j)$.
\end{lemma}

Lemma~\ref{lem:lips-schur} gives an alternative way of calculating $L_\star$ without scanning pairs of matrices of various dimensions. However, the quantity in Eq.~\eqref{eq:Schur} is still challenging to compute in practice. 
Here we will focus on the particular class of divided difference kernels for which the $x$- and $y$-dependence factorizes in the following sense. 

\begin{definition}[Haagerup-factorizable kernel]
We say that the divided difference kernel \eqref{eq:df-kernel} is \emph{Haagerup-factorizable}~\cite{aleksandrov2016,peller2016} if it admits a separable factorization of the form 
\begin{align}\label{eq:haagerup}
    K_f(x,y) = 
    \mathbf{p}(x)^\top\!\!\cdot\mathbf{q}^*(y) =  
    \sum_{k \geq 0} p_k(x)q^*_k(y)\,, 
\end{align}
where, for any $x,y\in\mathcal{D}$, the vectors 
$\mathbf{p}(x)\coloneq\{p_k(x)\}_{k\ge0}$ and $\mathbf{q}(y)\coloneq\{q_k(y)\}_{k\ge0}$ have bounded $\ell_2$-norms 
$\norm{\mathbf{p}(x)}_2=\big(\sum_{k\geq 0}\vert p_k(x)\vert^2\big)^{1/2}$ and $\norm{\mathbf{q}(y)}_2=\big(\sum_{k\geq 0}\vert q_k(y)\vert^2\big)^{1/2}$. 
\end{definition}

\noindent Notice from Eq.~\eqref{eq:haagerup} that Haagerup decompositions are not unique: any pair of linear transformations $\mathbf{p}\to W\mathbf{p},\mathbf{q}\to V\mathbf{q}$ with $W$ and $V$ satisfying $W^\top V^*=\mathds{1}$ leaves $K_f$ invariant. The operator-Lipschitz constant $L_\star$ of functions that admit Haagerup-factorizable kernels is given by an optimization over all possible such decompositions, as stated in the following theorem (see 
App. \ref{app:proof-main} for the proof):

\begin{theorem}[$L_\star$ as an optimization over Haagerup factorizations] \label{thm:OL-haagerup}
Let 
$\mathcal{D}\subset \mathbb{C}$ be a compact domain and $f:\mathcal{D}\to\mathbb{C}$ be a bounded function whose divided difference 
kernel $K_f$ is Haagerup-factorizable as in Eq.~\eqref{eq:haagerup}. 
Then $f$ is operator-Lipschitz on $\LD$ 
with optimal constant 
\begin{align}\label{eq:Lstar}
L_\star[f] = 
\inf_{\mathbf{p},\mathbf{q}} \left\{\sup_{x\in\mathcal{D}} \big\|\mathbf{p}(x)\big\|_2 \cdot \sup_{y\in\mathcal{D}} \big\|\mathbf{q}(y)\big\|_2 \right\},
\end{align}
where the infimum is over all $\mathbf{p},\mathbf{q}$ such that Eq.~\eqref{eq:haagerup} holds. 
\end{theorem}

\noindent Notice that the $\inf_{\mathbf{p},\mathbf{q}}$ in Eq.~\eqref{eq:Lstar} is in general intractable since the sum in Eq.~\eqref{eq:haagerup} can have infinitely many terms and, moreover, there are infinitely many $\mathbf{p},\mathbf{q}$ satisfying Eq. \eqref{eq:haagerup}. Nevertheless, we are interested in the particular case of polynomial functions. As we show next, the feasibility issues in Eq.~\eqref{eq:Lstar} simplify significantly in this case. We start by stating in the next lemma that polynomials are Haagerup factorizable with a finite number of terms (see proof in App.~\ref{app:haagerup-fact}).

\begin{lemma}[Haagerup factorization of polynomials]\label{lem:haagerup_poly}
     Let $\{\Phi_j(x)\}_{j=0}^{N}$ be a basis for the space of polynomials of degree up to $N$ in a domain $\mathcal{D}$, with $\deg(\Phi_j)= j$, and let $f(x) = \sum^{N}_{j=0} a_j \Phi_j(x)$ be an element of this space.
    Then the divided-difference kernel $K_f(x, y)$ is Haagerup factorizable. Moreover, any
    \begin{gather}\label{eq:valid-fact}
        \begin{aligned}
        p_k(x) &= \sum_{t=0}^{N-1} h_{p_k,t}\Phi_t(x) = \mathbf{h}_{p_k}^{\,\top}\mathbf{\Phi}(x), \\
        q_k(y) &= \sum_{t=0}^{N-1} h_{q_k,t}\Phi_t(y) = \mathbf{h}_{q_k}^{\,\top}\mathbf{\Phi}(y),
        \end{aligned}
    \end{gather}
   where
$\mathbf{h}_{p_k},\mathbf{h}_{q_k}\in\mathbb C^{N}$ and $\mathbf{\Phi}(x)\coloneq\left[ \Phi_k(x) \right]_{k=0}^{N-1} $,  define a valid factorization
   $K_f(x,y) = \sum_k p_k(x)q_k(y)^* = \mathbf{\Phi}(x)^\top G \,\mathbf{\Phi}(y)^*$ if $G \coloneq \sum_k \mathbf{h}_{p_k}\mathbf{h}_{q_k}^\dagger $ satisfies
   \begin{align}
       G=(D^{-1})^\top G'D^{-1},
   \end{align}
   with $G,G',D\in\mathbb{C}^{N\times N}$ such that $\mathbf{\Phi}(x)=D\cdot\left[ x^k \right]_{k=0}^{N-1}$ is the change from the basis of monomials of degree at most $N-1$, and $G'_{k,k'}=\sum_{j=0}^N \widehat{D}_{j,k+k'+1}\,a_j$ if $k+k'+1\leq N$ and $G'_{k,k'}=0$ otherwise. Here $\left[\Phi_j(x)\right]_{j=0}^{N}=\widehat{D}\cdot\left[x^s\right]_{s=0}^{N}$, and $D$ is the leading $N\times N$ block of the lower-triangular matrix $\widehat{D}$.
\end{lemma}

Lemma~\ref{lem:haagerup_poly} reduces the problem of finding Haagerup factorizations of polynomials to finding the finitely many coefficients $\mathbf{h}_{p_k},\mathbf{h}_{q_k}$ that hold the correct relation with a matrix $G$ that only depends on the basis and on $f(x)$ itself.
In particular, using the Chebyshev basis $\mathbf{\Phi}(x)\coloneq\left[ T_k(x) \right]_{k=0}^{N-1} \in [-1,1]^{N}$ in the domain $\mathcal{D}=[-1,1]$, we have
\begin{align}\label{eq:Hcheb}
    G_{k,k'} = (2-\delta_{k,0})(2-\delta_{k',0})\!\!\sum_{j=0}^{\lfloor \frac{N-k-k'-1}{2}\rfloor}\!\! a_{k+k'+1+2j}\,,
    \end{align}
for $k,k'\in\{0,\ldots,N-1\}$ with $k+k'+1\leq N$, and $G_{k,k'}=0$ otherwise. Since $K_f$ has degree $N-1$ in each variable, an $N$-dimensional basis suffices, and no smaller one exists when $\deg f=N$: the entries with $k+k'=N-1$ equal $(2-\delta_{k,0})(2-\delta_{k',0})\,a_N$ and all entries with $k+k'>N-1$ vanish, so $G$ is anti-triangular with nonzero anti-diagonal, hence invertible, and $K_f$ has rank exactly $N$.
 A similar statement holds for Laurent polynomials as implemented by QSP with a time-evolution oracle (see App.~\ref{app:haagerup-fact})). In this case, using the basis of monomials $\Phi_k(x)=x^{k-N/2}$, $k=0, \ldots,N-1$ and $N$ even, we have (here $ j=k+k'-N+1$)
(here $ j=k+k'-N+1$)
    \begin{align}\label{eq:Htrig}
        G_{k,k'}&=\begin{cases}
        \phantom{-}a_j, & 0<j\leq N/2,\quad 0\leq k-N/2 \leq j-1 \\
        -a_j, & -N/2\leq j<0,\quad j\leq k-N/2 \leq -1 \\
        \phantom{-}0, & \text{otherwise}
    \end{cases}\!.
    \end{align}
    
Lemma \ref{lem:haagerup_poly} also gives a clear operational way of finding the optimal factorization in the sense of Eq.~\eqref{eq:Lstar} and, consequently, computing the operator Lipschitz constant of a polynomial $f$. Given valid factorization terms as in Eq.~\eqref{eq:valid-fact}, 
if we define the matrices
\begin{align*}
P \coloneq \sum_k \mathbf{h}_{p_k}\,\mathbf{h}_{p_k}^\dagger,
\quad
Q \coloneq \sum_k \mathbf{h}_{q_k}\,\mathbf{h}_{q_k}^\dagger,
\end{align*}
then, by construction, they satisfy the positive-semidefinite constraint 
\begin{align}\label{eq:G-sip}
\Gamma \coloneq \begin{bmatrix}
P & G\\
G^\dagger & Q
\end{bmatrix}\succeq 0.    
\end{align}
In this parametrization, the quantities in Eq.~\eqref{eq:Lstar} are
\begin{align}\label{eq:s1s2}
\norm{\mathbf{p}( x)}_2^2 = \mathbf{\Phi}( x)^\top P\,\mathbf{\Phi}( x)^*,
\,\,
\norm{\mathbf{q}( y)}_2^2 = \mathbf{\Phi}( y)^\top Q\,\mathbf{\Phi}( y)^*.  \end{align}

The apparently unbounded number of terms in the Haagerup decomposition does not introduce an infinite-dimensional optimization for polynomial $f$, but two reductions are needed. First, every Haagerup decomposition of $K_f$ with an arbitrary number of terms, including non-polynomial factors and polynomial factors of degree larger than that of $K_f$ can be reduced, without increasing $\norm{\mathbf{p}(x)}_2$ or $\norm{\mathbf{q}(y)}_2$ at any point, to a decomposition of the form of Eq.~\eqref{eq:valid-fact} with at most $\mathrm{rank}(G)$ terms, whose Gram data $P,Q$ satisfy Eq.~\eqref{eq:G-sip} (Lemma~\ref{lem:compression} in App.~\ref{app:proof-sip}). Second, conversely, every pair $P,Q$ satisfying Eq.~\eqref{eq:G-sip} defines a decomposition: $\Gamma$ has rank $r\leq 2N$ and therefore admits a factorization $\Gamma=WW^\dagger$ with $W \in \mathbb{C}^{2N \times r}$. Writing $W=[A^\top\,B^\top]^\top$, we have $P=AA^\dagger$, $Q=BB^\dagger$, $G=AB^\dagger$, so the columns of $A$ and $B$ define a Haagerup decomposition with $r$ terms that gives the same $P,Q,G$ and the same norms in Eq.~\eqref{eq:s1s2}. Thus the infinite sum in Eq.~\eqref{eq:haagerup} can be replaced, for polynomial $f$, by a finite factorization, and the optimization variables reduce to the two $N\times N$ Hermitian matrices $P$ and $Q$.

This can be summarized in the lemma below.

\begin{lemma}[$L_\star$ as a semi-infinite optimization]\label{lem:semiinfinite_reformulation}
Let $f$ be a polynomial and $\mathcal{D}, N, \mathbf{\Phi}, G$ be as in Lemma \ref{lem:haagerup_poly}. 
Then $L_\star[f]$  is the optimal value of the following semi-infinite convex optimization problem 
\begin{align}\label{eq:Lstar-sip}
    \begin{aligned}
        \min_{L,P,Q} & \quad L\\
        \text{s.t.} &
        \begin{cases}
            \begin{bmatrix}
                P & G\\
                G^\dagger & Q
            \end{bmatrix}\succeq 0,\\
            \mathbf{\Phi}( x)^\top P\,\mathbf{\Phi}( x)^*\le L,\quad \forall \,  x \in \mathcal{D},\\
            \mathbf{\Phi}( y) ^\top Q\,\mathbf{\Phi}( y)^*\le L,\quad \forall \,  y \in \mathcal{D}\\
            P, Q \in \mathrm{Herm}(N).
        \end{cases}
    \end{aligned}
\end{align}
\end{lemma}
\noindent The proof is given in App.~\ref{app:proof-sip}, where we also show that the minimum is attained. This is a semi-infinite problem due to the infinitely many inequality constraints, parameterized by $ x, y\in\mathcal{D}$, mixed with a finite number of variables. 
The positive-semidefinite (PSD) constraint accounts for the $\inf_{\mathbf{p},\mathbf{q}}$ in Eq.~\eqref{eq:Lstar} over the set of all possible decompositions $\mathbf{p},\mathbf{q}$: every valid Haagerup decomposition can be compressed, without increasing either factor norm, to one whose Gram data satisfy Eq.~\eqref{eq:G-sip}, and every PSD block in Eq.~\eqref{eq:G-sip} defines a valid decomposition (App.~\ref{app:proof-sip}). The correspondence between decompositions and Gram matrices is many-to-one, not one-to-one. 

The infinitely many inequality constraints in Eq.~\eqref{eq:Lstar-sip} are amenable to exact reformulation via sum-of-squares (SOS) and Semidefinite Programming (SDP) (see App.~\ref{app:sip-to-sos} for the Chebyshev derivation). For a degree-$N$ Chebyshev
polynomial, this reformulation
contains $\mathcal{O}(N^2)$ scalar decision variables and positive-semidefinite
blocks of order $\mathcal{O}(N)$. A standard dense primal--dual interior-point
implementation therefore requires $\mathcal{O}(N^6)$ arithmetic operations per Newton
iteration and $\mathcal{O}(N^4)$ memory in the worst case
\cite{zhang2019conic}. Together with the standard
$\mathcal{O}(\sqrt{N}\log(1/\varepsilon))$ iteration bound for self-concordant
interior-point methods \cite{nesterov1994interior,vandenberghe1996semidefinite},
this gives a worst-case runtime of
$\mathcal{O}(N^{13/2}\log(1/\varepsilon))$. Even before runtime becomes an issue, the memory scaling is prohibitive. For real coefficients (the Chebyshev case, for which real symmetric variables suffice), the reformulation of App.~\ref{app:sip-to-sos} has $m=3N^2+N+1$ real decision variables ($P,Q,X_{P,1},X_{Q,1}\in\mathrm{Herm}(N)$, $X_{P,2},X_{Q,2}\in\mathrm{Herm}(N-1)$, and $L$). At $N=123$, the dense $m\times m$ Schur-complement matrix of a single interior-point step therefore occupies $m^2\times8\,\mathrm{B}\approx 16.6$~GB in double precision, before factorization and all other solver data. Complex Hermitian variables would double $m$ and quadruple this.

In the next section, we present an active-set method for Eq.~\eqref{eq:Lstar-sip} whose inner first-order iterations cost $\mathcal{O}(rN^2+N^3)$ arithmetic operations and $\mathcal{O}(rN+N^2)$ memory under standard dense-linear-algebra implementations, where $r$ is the current number of active points, which grows as points are inserted (Lemma~\ref{lem:sdp}). These per-iteration figures should be compared with the $\mathcal{O}(N^6)$ operations and $\mathcal{O}(N^4)$ memory of a single interior-point Newton step. Total runtimes are not directly comparable. The interior-point method needs $\mathcal{O}(\sqrt{N}\log(1/\varepsilon))$ Newton steps. In the active-set method, the finite dual is solved to relative accuracy $\varepsilon_{\rm rel}$ within $\mathcal{O}(r\log r/\varepsilon_{\rm rel}^3)$ iterations, typically far fewer in practice. A certified upper bound is available at every outer iteration without any recovery sweep, and the recovery iterations only tighten it. The number of outer iterations is not bounded a priori, and the accuracy of the final bounds is certified a posteriori. In terms of memory alone, the largest objects are two $2N\times2N$ matrices (a Gram matrix and its eigenvectors) and a constant number of $N\times N$ matrices, with real symmetric data and the evaluation vectors regenerated on demand (App.~\ref{sec:recov-haagerup}), a $2N\times2N$ matrix occupies $3.2$~GB at $N=10^4$, so degrees $N\sim10^4$ fit within the same $16$ GB.

\subsection{Active-set algorithm for approximating $L_\star$}
\label{subsec:active-set-alg}

This section turns the continuous problem in Eq.~\eqref{eq:Lstar-sip} into a finite active-set computation. This means that, at each step, we change the infinitely many inequality constraints on the continuous domain $\mathcal{D}$ to a carefully chosen finite set of points $\mathsf{X}=\{ x_i\}_{i=1}^{r_X},\mathsf{Y}=\{ y_j\}_{j=1}^{r_Y}$ that will asymptotically make the problem converge to $L_\star$.

Because it is a discretized formulation, given candidate solutions $P$ and $Q$, the maxima of $\mathbf{\Phi}( x)^\top P\,\mathbf{\Phi}( x)^*$ and $\mathbf{\Phi}( y) ^\top Q\,\mathbf{\Phi}( y)^*$ on $\mathsf{X},\mathsf{Y}$ is always smaller or equal to that of the semi-infinite problem. Therefore, the solution of this formulation yields a lower bound on $L_\star$. Moreover, given the optima $P, Q$ on the discrete set, the corresponding Haagerup factorization can be recovered, and, since it might not be the global optimal decomposition, it gives an upper bound on $L_\star$. 
These properties are formalized below, and their proof is presented in App.~\ref{sec:proof-lem-discr}.
\begin{lemma}[Bounding $L_\star$ as a finite optimization]\label{lem:sdp}
In the setup of Lemma \ref{lem:semiinfinite_reformulation}, let $\mathsf{X}=\{ x_i\}_{i=1}^{r_X},\mathsf{Y}=\{ y_j\}_{j=1}^{r_Y}$ be finite sets of points 
in the domain $\mathcal{D}$.
Then
\begin{enumerate}[leftmargin=*,label=(\roman*)]
\item substituting the infinitely many constraints in Eq. \eqref{eq:Lstar-sip} by their restriction to the finite number of points $ x\in\mathsf{X}$ and $ y\in\mathsf{Y}$ defines a SDP;
\item  the SDP in $(i)$ has a dual form without any semi-definite constraint and reads
\begin{align}\label{eq:dual-nuclear}
    L_{\mathrm{low}}=\max_{\substack{\mathbf{u}\in\mathbb{R}_+^{r_X},\,\mathbf{v}\in\mathbb{R}_+^{r_Y}\\
 \sum_i \! u_i+\sum_j \! v_j=2}}  \norm{\mathrm{diag}\big(\sqrt{\mathbf{u}}\big) \!\cdot\! E_\mathsf{X}^{\phantom{\dagger}} \!\cdot\! G \!\cdot\! E_\mathsf{Y}^\dagger \!\cdot\! \mathrm{diag}\big(\sqrt{\mathbf{v}}\big)}_*,
\end{align}
where $E_\mathsf{X}$ is the $r_X \times N$ matrix with rows given by $\mathbf{\Phi}( x_1)^\top,\ldots,\mathbf{\Phi}( x_{r_X})^\top$ (and similarly for $E_\mathsf{Y}$), and $\norm{\cdot}_*$ is the nuclear norm $\norm{A}_*\coloneq\trace\left(\sqrt{A^\dagger A} \right)$;
\item the SDP in $(i)$ has strong duality;
\item the solution of the dual problem \eqref{eq:dual-nuclear} gives the lower bound $L_{\mathrm{low}} \leq L_\star$;
\item if $E_\mathsf{X}$ and $E_\mathsf{Y}$ have full column rank, the finite primal optimum is attained, and, given dual variables $\mathbf{u},\mathbf{v}$ from $(ii)$, Haagerup polynomials $p_k,q_k$ in Eq.~\eqref{eq:haagerup} can be recovered by an iteration in which each sweep costs $\mathcal{O}(N^3+r N^2)$
arithmetic operations and $\mathcal{O}\!\left(rN+N^2\right)$ memory, with $r\coloneq r_X+r_Y$. The number of sweeps is not bounded a priori, and acceptance is certified a posteriori. A certified, possibly looser, upper bound in $(vi)$ is also available without any sweep;
\item with $p_k,\, q_k$ from $(v)$, a certified upper bound $L_{\mathrm{up}}\geq L_\star$ is obtained from the values of $\norm{\mathbf{p}(x)}_2^2$ and $\norm{\mathbf{q}(y)}_2^2$ on a grid of $M=\mathcal{O}(N)$ points, at a cost of $\mathcal{O}(N^2)$ arithmetic operations per point;
\item each first-order iteration for the problem in $(ii)$ costs $\mathcal{O}\!\left( I_{\rm{obj}}\big[(r_X+r_Y)N^2+N^3\big] \right)$
arithmetic operations and $\mathcal{O}\!\left((r_X+r_Y)N+N^2\right)$ memory, where $I_{\rm{obj}}$ is the number of objective evaluations in a line search. The number of iterations is bounded a priori by $\mathcal{O}(m\log m/\varepsilon_{\rm rel}^{3})$, with $m=r_X+r_Y$, for a relative accuracy $\varepsilon_{\rm rel}$ (App.~\ref{sec:dual-nuclear-complexity}), and is typically far smaller in practice.
\end{enumerate}
\end{lemma}

The gap between the upper and lower bounds given by Lem. \ref{lem:sdp} can be asymptotically closed to obtain the solution of the semi-infinite problem in Eq.~\eqref{eq:Lstar-sip}, and consequently $L_\star$. The idea is to solve the finite problem above and iteratively update the sets $\mathsf{X}$ and $\mathsf{Y}$.
At each step, we solve an instance of Lem.~\ref{lem:sdp}. 
We then use the solution to reconstruct the Haagerup decomposition polynomials, from which we derive the grid update rule. 
Algorithm \ref{alg:active-set} describes the procedure below. 
\begin{algo}\label{alg:active-set}
\quad\\
\indent Input: $\mathcal{D}$, $N$, $\mathbf{\Phi}$, $G$ as in Lem.~\ref{lem:semiinfinite_reformulation}, $t_{\max}\in \mathbb{Z}^+$, $\varepsilon_{\text{tol}}\ge 0$. \\
\indent Output: $L_{\mathrm{low}}$ and $L_{\mathrm{up}}^{\mathrm{best}}$ such that $L_{\mathrm{low}}\le L_\star\le L_{\mathrm{up}}^{\mathrm{best}}$ 
\begin{enumerate}[leftmargin=*]
    \item Choose two initial finite sets of distinct points
    $\mathsf{X}^{(0)}=\{x_i\}_{i=1}^{2N-1}$ and
    $\mathsf{Y}^{(0)}=\{y_j\}_{j=1}^{2N-1}$ in $\mathcal{D}$, \emph{e.g.}, Chebyshev--Lobatto nodes, and initialize $t=0$. The only property used below is that the evaluation matrices $E_{\mathsf{X}^{(0)}}$ and $E_{\mathsf{Y}^{(0)}}$ have full column rank, which holds for any $N$ or more distinct points (see the proof of Corollary~\ref{cor:finite-section}(b) in App.~\ref{app:proof-sip}). With $2N-1$ points, the squared factor norms, which are polynomials of degree at most $2N-2$ on $[-1,1]$ (or trigonometric polynomials of degree at most $N-1$ on $\mathbb{T}$), are determined by their values on the initial grid;

    \item Solve the dual problem
    in Eq.~\eqref{eq:dual-nuclear} for the current active sets
    $\mathsf{X}^{(t)}$ and $\mathsf{Y}^{(t)}$, obtaining
    $L_{\mathrm{low}}^{(t)}$ and optimal weights
    $\mathbf{u}^{\,(t)}_\star,\mathbf{v}^{\,(t)}_\star$.
    Recover from this solution the coefficient vectors of the
    Haagerup factors $p_k^{(t)}$ and $q_k^{(t)}$, as described in
    Sec.~\ref{sec:recov-haagerup};

    \item From the recovered factorization, compute certified upper bounds $\bar{a}_t\geq\max_{x\in\mathcal{D}}\norm{\mathbf{p}^{\,(t)}(x)}_2^2$ and $\bar{b}_t\geq\max_{y\in\mathcal{D}}\norm{\mathbf{q}^{\,(t)}(y)}_2^2$ with Lemma~\ref{lem:certified-max}, and set
    \begin{align*}
    L_{\mathrm{up}}^{(t)}
    &\coloneq
        \sqrt{\bar{a}_t\,\bar{b}_t}
        \;\geq\;
        \max_{x\in\mathcal{D}}
        \norm{\mathbf{p}^{\,(t)}(x)}_2 \cdot \max_{y\in\mathcal{D}}
        \norm{\mathbf{q}^{\,(t)}(y)}_2,
    \end{align*}
    and
    \begin{align}\label{eq:Lup-best}
        L_{\mathrm{up}}^{\mathrm{best},(t)} &\coloneq \min_{s\leq t} L_{\mathrm{up}}^{(s)}.
    \end{align}
    Both are upper bounds on $L_\star$ by Lemma~\ref{lem:sdp}(vi).
    To refine the active sets, locate the critical points of
    $\norm{\mathbf{p}^{\,(t)}(x)}_2^2$ and
    $\norm{\mathbf{q}^{\,(t)}(y)}_2^2$ in the interior of $\mathcal{D}$ (the zeros of the first derivative), together with the endpoints $\pm1$ when $\mathcal{D}=[-1,1]$ ($\mathbb{T}$ has no endpoints).
    Denote these sets by
    $\mathsf{X}^{(t)}_{\mathrm c}$ and
    $\mathsf{Y}^{(t)}_{\mathrm c}$, respectively, and
    extract from them the points at which the squared factor norms exceed the current lower bound $L_{\mathrm{low}}^{(t)}$, \emph{i.e.}, the violated constraints of Eq.~\eqref{eq:Lstar-sip} at level $L_{\mathrm{low}}^{(t)}$:
    \begin{align}\label{eq:active-violations}
    \begin{aligned}
        \widehat{\mathsf{X}}^{(t)}
        &\coloneq \left\{ x_{\mathrm c}\in\mathsf{X}^{(t)}_{\mathrm c} : \norm{\mathbf{p}^{\,(t)}(x_{\mathrm c})}_2^2 > L_{\mathrm{low}}^{(t)} \right\},
        \\
        \widehat{\mathsf{Y}}^{(t)}
        &\coloneq \left\{ y_{\mathrm c}\in\mathsf{Y}^{(t)}_{\mathrm c} : \norm{\mathbf{q}^{\,(t)}(y_{\mathrm c})}_2^2 > L_{\mathrm{low}}^{(t)} \right\}.
    \end{aligned}
    \end{align}
    The critical points are used only to select points to insert; the validity of $L_{\mathrm{up}}^{(t)}$ does not depend on them.

    \item If
    $\widehat{\mathsf{X}}^{(t)}=\emptyset$ and
    $\widehat{\mathsf{Y}}^{(t)}=\emptyset$, or $L_{\mathrm{up}}^{\mathrm{best},(t)} - L_{\mathrm{low}}^{(t)} \leq \varepsilon_{\mathrm{tol}}$, or
    $t=t_{\max}$,
    stop and return
    $L_{\mathrm{low}}^{(t)}$ and
    $L_{\mathrm{up}}^{\mathrm{best},(t)}$.
    Otherwise, augment the active sets according to
    \begin{align}\label{eq:update-active-sets}
    \begin{aligned}
        \mathsf{X}^{(t+1)}
        &= \mathsf{X}^{(t)} \cup \widehat{\mathsf{X}}^{(t)},
        \\
        \mathsf{Y}^{(t+1)}
        &= \mathsf{Y}^{(t)} \cup \widehat{\mathsf{Y}}^{(t)},
    \end{aligned}
    \end{align}
    set $t\leftarrow t+1$, and return to Step~2.
\end{enumerate}
\end{algo}
\noindent Notice that each iteration only augments the grid by new points, hence the optimal value $L_{\mathrm{low}}^{(t)}$ after each step can never decrease. The individual upper bounds $L_{\mathrm{up}}^{(t)}$ come from different factorizations and need not be monotone, but their running minimum $L_{\mathrm{up}}^{\mathrm{best},(t)}$ is nonincreasing.
In App.~\ref{sec:proof-alg-conv}, we prove that, in exact arithmetic, both $L_{\mathrm{low}}^{(t)}$ and $L_{\mathrm{up}}^{\mathrm{best},(t)}$ converge to the optimal value $L_\star$ of the semi-infinite problem in Eq.~\eqref{eq:Lstar-sip}.
In the numeric investigation in Sec.\ref{sec:h6-bench}, Algorithm~\ref{alg:active-set} with initial grids $\mathsf{X}^{(0)},\mathsf{Y}^{(0)}$ chosen as Chebyshev--Lobatto nodes terminated after a few iterations.

Algorithm~\ref{alg:active-set} therefore associates to each polynomial $f$
certified lower and upper bounds to its optimal OL constant $L_\star[f]$, 
\begin{align}\label{eq:lip-chain}
L_{\mathrm{low}}[f] \le L_\star[f] \le L_{\mathrm{up}}[f].
\end{align}
where $L_{\mathrm{up}}[f]$ denotes the returned value $L_{\mathrm{up}}^{\mathrm{best}}$. The lower bound is certified by weak duality (Corollary~\ref{cor:finite-section}(a) in App.~\ref{app:proof-sip}).This provide the connection between polynomial approximation and implementation cost, see Sec.~\ref{sec:end-to-end-budget}, where we illustrate the impact of these quantities taking a resource model for QSVT using LCU to construct the input oracle, and we use the resulting cost to compare and design polynomial approximations.

\section{Relation to Prior work}\label{sec:prior}
We now compare Eq.~\eqref{eq:robustness_bound} with the robustness bounds available for QSVT.
For a function $f$ that is Lipschitz with scalar constant $\ell_\star[f]$ on $[-1,1]$, the modulus-of-continuity estimate of Ref.~\cite[Cor.~7.5]{aleksandrov2010holder} gives, for Hermitian $A,B$ with spectra in $[-1,1]$, $\norm{f(A)-f(B)} \le c\,\ell_\star[f]\,\big(\log(2/\norm{A-B})+1\big)\norm{A-B}$ with an unspecified absolute constant $c>0$. Ref.~\cite[Thm.~20 and Cor.~21]{Gilyen_2018_arxiv} uses the analogous estimate of Ref.~\cite{FarforovskayaNikolskaya2009}, with a squared logarithm, for QSVT. This estimate requires only scalar Lipschitz continuity, but its logarithmic factor diverges as $\norm{A-B}\to0$. Hence, for a fixed OL function $f$, Eq.~\eqref{eq:op-lip} is tighter by an unbounded factor as $\norm{A-B}\to 0$, for a given $\norm{A-B}$, the comparison depends on $c$ and on the ratio $L_\star[f]/\ell_\star[f]$.

For degree-$N$ polynomials that satisfy the requirements of Ref.~\cite[Cor.~8]{Gilyen_2018_arxiv} (definite parity, $\vert f\vert\leq1$ on $[-1,1]$, and the conditions outside $[-1,1]$), Ref.~\cite[Lemma~22]{Gilyen_2018_arxiv} gives $\norm{f(A)-f(B)} \leq 4N\sqrt{\norm{A-B}}$ for all $\norm{A},\norm{B}\leq1$. Write $\varepsilon^\prime\coloneq\varepsilon-\varepsilon_{\rm{app}}-S\,\varepsilon_{\rm{synth}}$ for the part of a total error budget $\varepsilon$ that remains for the input error. This bound then admits $\varepsilon_{\rm{input}}\le\varepsilon^{\prime 2}/(16N^2)$, whereas Eq.~\eqref{eq:robustness_bound} admits $\varepsilon_{\rm{input}}\le\varepsilon^\prime/L_\star$. Pointwise, the linear bound is the tighter of the two exactly when $L_\star\varepsilon_{\rm{input}}<4N\sqrt{\varepsilon_{\rm{input}}}$, \emph{i.e.}, when $\varepsilon_{\rm{input}}<16N^2/L_\star^2$. For the admissible budgets, the linear route admits the larger input error exactly when $\varepsilon^\prime/L_\star>\varepsilon^{\prime 2}/(16N^2)$, i.e., when $\varepsilon^\prime<16N^2/L_\star$, and then by the factor $16N^2/(\varepsilon^\prime L_\star)$. For example, $L_\star[T_N]=N^2$, so for $f=T_N$ the factor is $16/\varepsilon^\prime$.
If, in addition, the spectra stay away from $\pm1$, Ref.~\cite[Lemma~23]{Gilyen_2018_arxiv} gives the linear bound $\norm{f(A)-f(B)}\leq N\sqrt{2/(1-\norm{(A+B)/2}^2)}\,\norm{A-B}$ whenever $\norm{A-B}+\norm{(A+B)/2}^2\leq1$. Its constant can be smaller than $L_\star$, whereas Eq.~\eqref{eq:robustness_bound} requires no gap condition and uses the optimal constant of the implemented polynomial. Finally, a dimension-dependent bound of the type of Eq.~\eqref{eq:op-lip-bound}, with the factor $\log(\min(\mathrm{rank}(A)\mathrm{rank}(B)))$ and for functions with $f(0)=0$, is stated for singular value transformations in the full version of Ref.~\cite{Chia_2020}.

\section{Conclusion}\label{sec:conclusion}

Here we have shown that the sensitivity of a QSP/QSVT implementation with respect to the input matrix is controlled by the operator-Lipschitz constant, $L_\star$, of the target function $f_{\mathrm{T}}$ or of its approximating polynomial $f$.
We then cataloged a list of operator-Lipschitz constants of known functions of practical interest, developed a test to discern the cases when it is not the same as the scalar constant, and developed an efficient classical algorithm that outputs certified upper and lower bounds on $L_\star$ for arbitrary polynomials on a given domain. The latter algorithm is specially relevant for functions without a known operator-Lipschitz constant.

From this characterization we devised a framework that allows one to think of an QSP/QSVT implementation with two layers.
The top layer is an algorithm that relates the target function $f_{\mathrm{T}}$ to its $\varepsilon_{\mathrm{app}}$-polynomial approximation $f$, which is in principle independent of the QSP protocol.  
The bottom layer deals with how this approximation is actually implemented physically on a fault-tolerant quantum computer, and it can itself be divided also in two levels.
The top level involves an algorithm for choosing the QSP sequence parameters, which are subject to finite precision errors $\varepsilon_{\mathrm{synth}}$;
the bottom level involves the algorithm used to build the approximate QSP oracle for a given precision $\varepsilon_{\mathrm{input}}$. The latter is in turn further specified by an access model to the input matrix itself.

Given a cost metric for physical implementation, such as for instance the $T$-gate count, our algorithm for bounding the operator-Lipschitz constant allows one to optimize the two layers jointly.
This is because, instead of using the standard approach of searching for a good polynomial approximation of $f_{\mathrm{T}}$ alone, one can now select the cheapest polynomial $f$ in terms of the chosen cost metric subject to the constraint that the total accumulated error is bounded by a user-defined threshold $\varepsilon$.
Moreover, we have applied the framework to the explicit example of ground-state optimization via eigenvalue thresholding for a molecular H$_6$ Hamiltonian. There, we found a polynomial that reduces the total $T$-gate counts by nearly a factor of $2$ compared to the standard polynomial described in the literature for that same task. 

The framework developed offers a toolbox to quantitatively introduce robustness considerations actively in the design of polynomials for QSP/QSVT applications.
Moreover, our findings may enable further optimization of QSP/QSVT algorithms across different access models, noise assumptions, and fault-tolerant cost metrics, helping to translate their asymptotic guarantees into practical advantages.

\section{Acknowledgements}

D.S.F.'s work was funded by the Novo Nordisk Foundation (Grant No. NNF20OC0059939 Quantum for Life). 
ChatGPT and Claude models have contributed to numerical implementations, refining mathematical proofs, and revising the text. The authors have manually checked all of them and take full responsibility for the correctness of the results.

\bibliography{arXiv-v1/references}

\appendix


\section{Input error breakdown using the LCU access model}\label{subsec:lcu-input-error}

Here we analyze in detail the input error $\varepsilon_{\mathrm{input}}$ using the LCU access model~\cite{Babbush_2018,Low_2024}. The ideal input operator is written as
\begin{align}\label{eq:lcu-normalized-input}
    A = \frac{1}{\alpha}\sum_{j\in[M]} a_j U_j,
    \quad \alpha \coloneq \sum_{j\in[M]}\vert a_j\vert,
\end{align}
with unitaries $U_j$ and coefficients $a_j\in\mathbb{C}$. 
A relevant example is that of molecular Hamiltonians, for which the $U_j=P_j$ are Pauli strings. A block encoding $U_A$ of $A$ is obtained from two subroutines~\cite{Babbush_2018,Gilyen_2019}: \textsc{prepare}, which loads the normalized coefficients into amplitudes of an auxiliary register, $\textsc{prepare}\ket{0}=\sum_j\sqrt{\vert a_j\vert/\alpha}\,\ket{j}$, and \textsc{select}$\,=\sum_j\ketbra{j}{j}\otimes U_j$, which applies the corresponding unitary $U_j$ controlled on that register. Then $U_A=\textsc{prepare}^\dagger\cdot\textsc{select}\cdot\textsc{prepare}$ satisfies $\bra{0}U_A\ket{0}=A$.

In practice, the implemented \textsc{prepare} only approximates the ideal one. Suppose that it outputs $\sum_j\widetilde{c}_j\ket{j}\ket{\gamma_j}$, where the normalized states $\ket{\gamma_j}$ account for possible garbage registers. If $\textsc{prepare}^\dagger$ is implemented as the exact inverse of the implemented \textsc{prepare}, and \textsc{select} is exact (as it is, \emph{e.g.}, for Pauli strings), then the operator actually block-encoded is
\begin{align}\label{eq:lcu-implemented-input}
    B = \sum_{j\in[M]} \beta_j\, U_j,
    \quad \beta_j\coloneq\vert\widetilde{c}_j\vert^2,
    \quad \sum_{j\in[M]} \beta_j=1.
\end{align}
Hence, $B$ has the same LCU structure as $A$, with the ideal probabilities $a_j/\alpha$ replaced by the implemented ones $\beta_j$. In particular, if the $U_j$ are Hermitian, as for Pauli strings, $B$ is Hermitian with $\norm{B}\leq1$, so that $A,B\in\mathcal{L}_{[-1,1]}(\mathcal{H})$. This structure yields the following simple and classically computable certificate for the input error.

\begin{corollary}[LCU input error bound]\label{coro:pauli-access-robustness}
    Let $A=\sum_{j\in[M]}\alpha_jU_j$ and $B=\sum_{j\in[M]}\beta_jU_j$ be linear combinations of the same unitaries $U_j$. Then\
    \begin{align}\label{eq:coefficient-l1-bound}
        \norm{A-B}\leq\norm{\mathbf{\alpha}-\mathbf{\beta}}_1\coloneq\sum_{j\in[M]}\vert\alpha_j-\beta_j\vert,\
    \end{align}
    so one may take $\varepsilon_{\rm{input}}=\norm{\mathbf{\alpha}-\mathbf{\beta}}_1$ in Eq.~\eqref{eq:end-to-end-budget}. 
\end{corollary}

\noindent The bound in Eq.~\eqref{eq:coefficient-l1-bound} does not require diagonalizing $A$ or $B$ and can be evaluated from the coefficients alone, before any quantum resource is spent. It can, however, be loose compared to the spectral norm, since it ignores cancellations between different terms.

Corollary~\ref{coro:pauli-access-robustness} suggests splitting $\varepsilon_{\rm{input}}$ according to the two mechanisms that make the coefficients of $B$ differ from those of $A$. The first is classical: to reduce the cost of \textsc{prepare} and \textsc{select}, which grows with the number of terms, one may discard the least important terms of Eq.~\eqref{eq:lcu-normalized-input} and retain only a subset $S\subseteq[M]$,
\begin{align}
    A_{\rm{trunc}}\coloneq\frac{1}{\alpha}\sum_{j\in S}a_jU_j.
\end{align}
We keep the normalization $\alpha$ fixed, so that the retained operator is not rescaled and $\norm{A_{\rm{trunc}}}\leq1$ still holds; in \textsc{prepare}, the discarded weight $1-\sum_{j\in S}\vert a_j\vert/\alpha$ can be assigned to a pair of terms $\pm\mathds{1}$ with equal weights, whose contributions cancel in the block. We also assume that $S$ is chosen so that $A_{\rm{trunc}}$ remains self-adjoint. The second mechanism is quantum: the implemented \textsc{prepare} loads the retained coefficients only to finite precision, giving the operator $B$ of Eq.~\eqref{eq:lcu-implemented-input}. By the triangle inequality,
\begin{align}\label{eq:encoded-input-error-budget}
    \norm{A-B}
    &\leq\varepsilon_{\rm{trunc}}+\varepsilon_{\rm{fp}}\coloneq\varepsilon_{\rm{input}},
\end{align}
where Corollary~\ref{coro:pauli-access-robustness} gives the truncation error as the discarded coefficient weight, $\norm{A-A_{\rm{trunc}}}\leq
\alpha^{-1}\sum_{j\notin S}\vert a_j\vert \coloneq \varepsilon_{\rm{trunc}}$, 
and $\norm{A_{\rm{trunc}}-B}\le \varepsilon_{\rm{fp}}$ is the block-encoding error. The latter can again be certified in coefficient space: denoting by $c_j$ the ideal amplitudes loaded by the truncated \textsc{prepare} (including those of the cancelling pair $\pm\mathds{1}$, so that $A_{\rm{trunc}}$ and $B$ are expanded in the same set of unitaries), Corollary~\ref{coro:pauli-access-robustness} and the Cauchy--Schwarz inequality give
\begin{align}
    \norm{A_{\rm{trunc}}-B}\leq\sum_j\Big\vert\vert c_j\vert^2-\vert\widetilde{c}_j\vert^2\Big\vert\leq2\,\norm{\mathbf{c}-\widetilde{\mathbf{c}}}_2,
\end{align}
so a coefficient-loading routine with amplitude error $\norm{\mathbf{c}-\widetilde{\mathbf{c}}}_2\leq\varepsilon_{\rm{fp}}/2$ suffices.

The two input contributions have different natures and costs. The truncation error $\varepsilon_{\rm{trunc}}$ is fixed classically, before the circuit is built, and reduces the number of LCU terms, on which the cost of \textsc{select} and \textsc{prepare} depends roughly linearly. The block-encoding error $\varepsilon_{\rm{fp}}$ is fixed by the precision of the coefficient loading, \emph{i.e.}, by the number of bits and auxiliary qubits used in \textsc{prepare}, and its associated precision $1/\varepsilon_{\rm{fp}}$ enters the cost only logarithmically (see App.~\ref{sec:lcu-res-model}). Both errors, however, are amplified by the same factor $L_{\mathrm{up}}$, so a smaller OL constant can be traded either for a more aggressive truncation or for a cheaper \textsc{prepare}. 

\section{Sufficient criterion for $L_\star \ne \ell_\star$}\label{app:criteria}
 
\begin{lemma}
    Let $g$ be a real differentiable function in $\mathcal{D}=[-1,1]$ with scalar Lipschitz constant $\ell_\star=\text{sup}_{x\in\mathcal{D}}|g'(x)|$, and let $x_\star$ be such that $|g'(x_\star)|=\ell_\star$. Define $\Delta g(y)=(g(y)-g(x_\star))/(y-x_\star)$. If there are $y\in\mathcal{D}$ and $\theta\in(0,2\pi)$ with $y\ne x_\star$ and $\theta\ne k\pi/2$ for $k\in\mathbb{Z}$ such that 
    \begin{align}\label{eq:cond_L_l}
        \frac{g'(y)}{\ell_\star} < \frac{1+\cos\theta}{\cos\theta}\left(\frac{\Delta g(y)}{\ell_\star}\right)^2-\frac{1}{\cos\theta},
    \end{align}    
   then the operator Lipschitz constant $L_\star$ of $g$ is strictly larger than $\ell_\star$.
\end{lemma}

\begin{proof}
     For any $\mathsf{X},\mathsf{Y}\in\mathcal{D}^2$, and $Z\in\mathbb{C}^{2\times 2}$ with $\norm{Z}=1$,  Eq. \eqref{eq:Schur} implies that
    \begin{align}\label{eq:lip_reduc}
       L_\star \geq \big\|K_{g,2}(\mathsf{X},\mathsf{Y})\circ Z\big\|\,.
    \end{align}
    In particular, Eq. \eqref{eq:lip_reduc} is satisfied for $\mathsf{X}=\mathsf{Y}=\{x_\star,\,y\}$ and $Z=\begin{pmatrix}
        \cos(\theta)&\sin(\theta)\\
        \sin(\theta)& -\cos(\theta)
    \end{pmatrix}$ for some $y\in \mathcal{D}$ and $\theta\in[0,2\pi)$. If there is a choice of $y,\,\theta$ such that the right-hand side of Eq. \eqref{eq:lip_reduc} is larger than $\ell_\star$, then $L_\star>\ell_\star$. 

   Denoting $c\coloneq \cos(\theta)$ and $s\coloneq \sin(\theta)$,  we can write explicitly $\big\|K_{g,2}(\mathsf{X},\mathsf{Y})\circ Z\big\|=\norm{\begin{matrix}
        c\,g'(x_\star)&s\,\Delta g(y)\\
        s\,\Delta g(y) & -c\,g'(y)
    \end{matrix}}$.
Without loss of generality, we can assume $g'(x_\star)>0$. Since by definition $g'(x_\star)=\ell_\star\geq |g'(y)|$, we get $2\big\|K_{g,2}(\mathsf{X},\mathsf{Y})\circ Z\big\|=c(\ell_\star -g'(y))+\sqrt{c^2[\ell_\star+g'(y)]^2+4s^2\Delta g(y)^2}$. Now, Eq. \eqref{eq:cond_L_l} follows from imposing $\big\|K_{g,2}(\mathsf{X},\mathsf{Y})\circ Z\big\|>\ell_\star$ and straightforward algebraic manipulations. 
\end{proof}

\begin{figure}[t]
    \centering
    \includegraphics[width=0.8\linewidth]{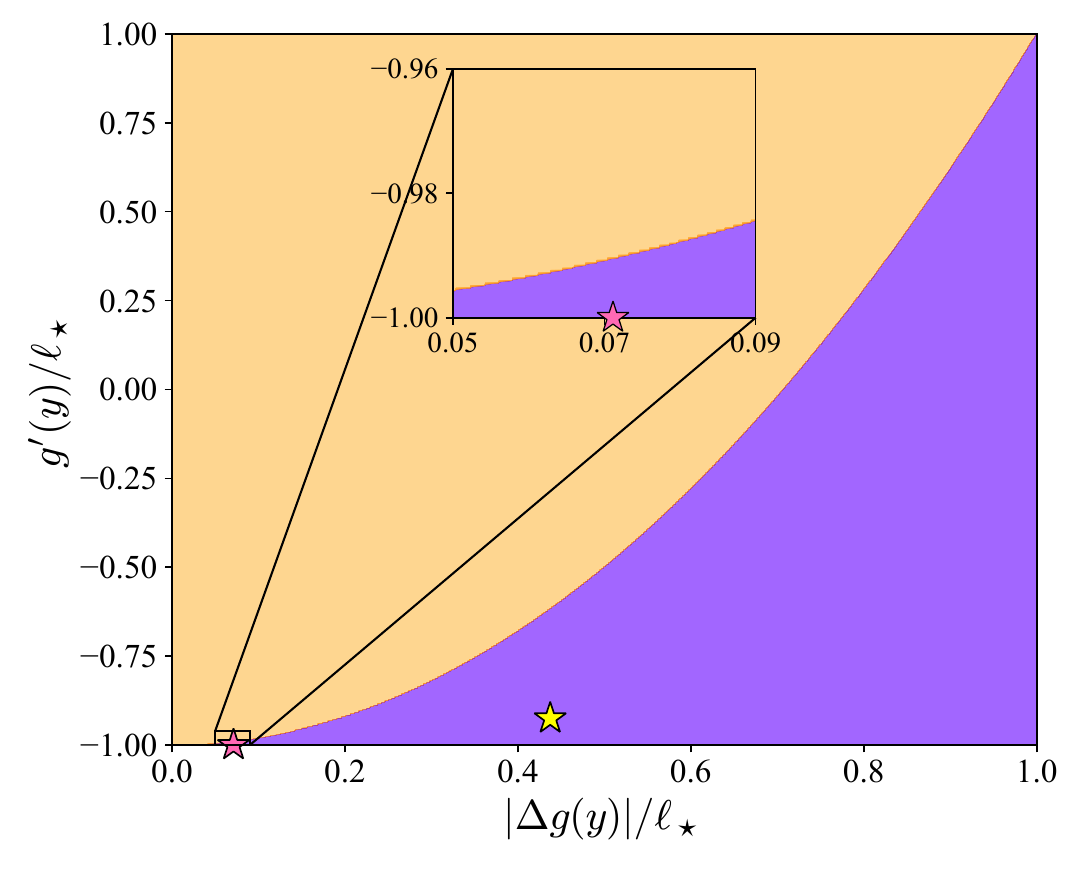}
    \caption{Region in which  Eq.~\eqref{eq:cond_L_l} holds for some admissible $\theta$. If a function $g(x)$ has at least one point $y$ falling inside the purple region, then it satisfies Eq.~\eqref{eq:cond_L_l} and therefore $L_\star>\ell_\star$. Outside of this region, we cannot make any claim. The stars correspond to the polynomial $g(x)=3x-2x^3$ (yellow, example from \cite{Aleksandrov2020}) and a polynomial generated by Sec.~\ref{sec:res-aware-poly-sel} (pink) to approximate the Heaviside function with $\nu=0.09$ and $\varepsilon=0.01$.}
    \label{fig:criteria}
\end{figure}

\section{Examples of operator-Lipschitz functions}
\label{app:trivial_bounds_scalar}

Here we list some known operator-Lipschitz (OL) functions and their optimal constants $L_\star$, based on Ref.~\cite{aleksandrov2016}. 
Moreover, we prove that polynomial functions on a bounded domain and continuous functions on the unit circle are OL, deriving upper bounds for their optimal constant $L_\star$. 
We also present trivial bounds for the corresponding scalar Lipschitz constants $\ell_\star$. 

Let $\mathcal{H}$ be a Hilbert space and consider functions on the set $\LD$ of linear operators acting on $\mathcal{H}$ whose spectra lie in $\mathcal{D}$. 
The following is a non-exhaustive list of  functions with known $L_\star$ (see \cite{aleksandrov2016}):
\begin{itemize}[leftmargin=*]
    \item $\forall\,\mu\in\mathbb{C}\setminus\mathbb{R}$, $f(x)=(\mu-x)^{-1}$ is OL on $\mathcal{L}_\mathbb{R}(\mathcal{H})$ with $L_\star=\vert\text{Im}\,\mu\vert^{-2}$;
    \item $\forall\,t\in\mathbb{R}$, $f(x)=e^{itx}$ is OL on $\mathcal{L}_\mathbb{R}(\mathcal{H})$ with $L_\star=\vert t\vert$;
    \item Any differentiable function $f$ on $\mathbb{R}$ with 
    $f^\prime$ positive-definite (meaning the matrix $\{f^\prime(x_i-x_j)\}_{i,j\in[r]}$ is positive-definite for any set of points $x_1,\ldots,x_r\in\mathbb{R}$) is OL on $\mathcal{L}_\mathbb{R}(\mathcal{H})$ with $L_\star=f^\prime(0)$; 
    \item $\forall\,\beta>0$ and $a,b\in\mathbb{R}$ with $a<b$, $f(x)=e^{-\beta x}$ is OL on $\mathcal{L}_{[a,b]}(\mathcal{H})$ with $L_\star=\beta e^{-\beta a}$; 
    \begin{proof}
    Using the Duhamel formula $e^{-\beta A}-e^{-\beta B}=-\beta\int_0^1e^{-\beta(1-s)A}(A-B)e^{-\beta s B}ds$ we get $\norm{e^{-\beta A}-e^{-\beta B}}\leq \beta e^{-\beta a}$ for $\mathcal{D}=[a,b]$. The conclusion comes from $L_\star\leq \ell_\star$ and $\ell_\star=\sup_{x\in\mathcal{D}}|f'(x)|=\beta e^{-\beta a}$. 
    \end{proof}
    \item $\forall\,\beta>0$,  $f(x)=\operatorname{erf}(\beta x)$ is OL on $\mathcal{L}_{[-1,1]}(\mathcal{H})$ with $L_\star=2 \beta /\sqrt{\pi}$;
    \item $\forall\,\mu\in\mathbb{C}\setminus\mathbb{T}$, $f(z)=(\mu-z)^{-1}$ is OL on $\mathcal{L}_\mathbb{T}(\mathcal{H})$ with $L_\star=\big(\vert\mu\vert-1\big)^{-2}$.    
\end{itemize}

The following Lemma proves that polynomials of bounded operators are operator-Lipschitz and provides an upper bound for their optimal constant $L_\star$. 

\begin{lemma}[Operator-Lipschitzness of polynomials]\label{lem:OLbound1}
Let $\mathcal{D}\subset\mathbb{R}$ be a bounded set, $R=\sup_{x\in\mathcal{D}}\vert x\vert$, $\mathcal{H}$ a finite-dimensional Hilbert space, and $f(x)=\sum_{k=0}^N c_kx^k$ a polynomial on $\mathcal{D}$. 
Then $f$ is operator-Lipschitz on $\LD$ with optimal constant $L_\star[f]\le\sum_{k=0}^N \vert c_k\vert\,k\,R^{k-1}$. 
\end{lemma}

\begin{proof}
We first prove it for the monomial $x^k$. For bounded operators $A,B\in\LD$, we have the telescopic identity
\begin{align*}
A^k - B^k = \sum_{j=0}^{k-1} A^j(A-B)B^{k-1-j}\,.
\end{align*}
Taking the spectral norm and using the triangle inequality we get
\begin{align*}
\norm{A^k - B^k} &\le \left(\sum_{j=0}^{k-1} \norm{A}^j\norm{B}^{k-1-j}\right)\norm{A-B}\\ 
&\le \left(k\,R^{k-1}\right)\norm{A-B}\,,
\end{align*}
where in the second line we have used $\norm{A}\le \sup_{x\in\mathcal{D}}\vert x\vert \coloneq R$ (and similarly for $\norm{B}$). 
It follows that
\begin{align*}
\norm{f(A) - f(B)} &\le \sum_{k=0}^{N} \vert c_k\vert\,\norm{A^k-B^k} \notag\\
&\le \left(\sum_{k=0}^{N} k\,R^{k-1}\vert c_k\vert\right)\norm{A-B}\,,
\end{align*}
which is the operator-Lipschitz condition with constant $L_{\rm up}=\sum_{k=0}^{N} k\,R^{k-1}\vert c_k\vert$. Since $L_\star$ is by the definition the minimum $L_{\rm up}$, the claimed inequality follows.
\end{proof}

\noindent Notice that if $f$ is provided in a different basis one can simply work out the bound by changing to the monomial basis. \emph{E.g.}, for a Chebyshev-basis polynomial $f(x)=\sum_{k=0}^N a_k\,T_k(x)$ the coefficients are related as
$$
a_j=\sum_{k=j}^N c_k\,\alpha_{j}^{(k)}\!,\quad \alpha_{j}^{(k)} \!=\! 
\begin{cases}
2^{1-k-\delta_{j,0}} \dbinom{k}{\tfrac{k-j}{2}} & k-j \text{ even}\\
0 & k-j \text{ odd}\\
\end{cases}\!.
$$

The following Lemma proves that continuous functions on the unit circle are operator-Lipschitz and provide an upper bound for their optimal constant $L_\star$. Notice that this includes as particular case the family of trigonometric polynomials $f(z)=\sum_{k=-N}^N c_kz^k$ of a given degree $N$. 

\begin{lemma}[Operator-Lipschitzness of continuous functions on the unit circle $\mathbb{T}$]\label{lem:OLbound2}
Let $\mathcal{H}$ be a finite-dimensional Hilbert space and $f(z)=\sum_{k\in\mathbb{Z}} c_kz^k$ a continuous function on the unit circle $\mathbb{T}\coloneq \{z\in\mathbb{C}\,:\,\vert z\vert=1\}$ such that $\sum_{k\in\mathbb{Z}}\vert k\vert\,\vert c_k\vert<\infty$. 
Then $f$ is operator-Lipschitz on $\mathcal{L}_\mathbb{T}(\mathcal{H})$ with optimal constant $L_\star\le\sum_{k\in\mathbb{Z}}\vert k\vert\,\vert c_k\vert$. 
\end{lemma}

\begin{proof}
First notice that
\begin{align*}
\|&f(A)-f(B)\| \le \sum_{k\in\mathbb{Z}} \vert c_k\vert\,\norm{A^k-B^k}\notag\\ 
&= \sum_{k>0} \vert c_k\vert\,\norm{A^k-B^k} + \sum_{k<0} \vert c_k\vert\,\norm{A^k-B^k}\,.
\end{align*}
For the $k>0$ terms one can proceed identically to the proof of Lemma \ref{lem:OLbound1}, noticing that for the unit circle we have $R=\sup_{z\in\mathbb{T}}\vert z\vert=1$. 
For the $k<0$ terms, one can also proceed similarly noticing that now the telescopic identity is to be interpreted as
\begin{align*}
A^k - B^k = (A^{-1})^{-k} - (B^{-1})^{-k}\,.
\end{align*} 
This leads to
\begin{align*}
\frac{\norm{f(A)-f(B)}}{\norm{A-B}} \le \sum_{k>0} k\,\vert c_k\vert + \sum_{k<0} (-k)\,\vert c_k\vert = \sum_{k\in\mathbb{Z}} \vert k\vert\,\vert c_k\vert.
\end{align*}
\end{proof}

\section{Proof of Prop.~\ref{prop:robustness}}\label{app:proof_main}

\begin{proof}
First notice that
\begin{align*}
\|{f_\text{T}(A)-\widetilde{f}(B)\|}\le \norm{f_\text{T}(A)-f(B)}+\|{f(B)-\widetilde{f}(B)\|}\,.
\end{align*}
For the second term, we use the fact the QSVT sequence has $S$ single-qubit rotations, each of which has a synthesis error $\varepsilon_{\mathrm{synth}}$, so that $\|{f(B)-\widetilde{f}(B)\|}\le S\,\varepsilon_{\mathrm{synth}}$ after a sequence of triangle inequalities. For the first term, we break the analysis in two cases depending on whether or not the target function is OL. 
If $f_\text{T}$ is OL, we have
\begin{align*}
    \norm{f_\text{T}(A)-f(B)}
    &\leq \norm{f_\text{T}(A)-f_\text{T}(B)}+\norm{f_\text{T}(B)-f(B)}\\
    &\leq L_\star[f_\text{T}]\cdot\norm{A-B} + \varepsilon_{\rm{app}}.
\end{align*}
Otherwise, we use the fact that $f$ is always OL to get
\begin{align*}
    \norm{f_\text{T}(A)-f(B)}
    &\leq \norm{f_\text{T}(A)-f(A)}+\norm{f(A)-f(B)}\\
    &\leq \varepsilon_{\rm{app}}+L_\star[f]\cdot\norm{A-B}.
\end{align*}
Putting all together, the claim follows.
\end{proof}



\section{Lipschitz continuity for singular value transformations}\label{app:ext-QSVT}

In this appendix, we state and prove how the robustness discussed in the main text for quantum eigenvalue transformation extends to quantum singular value transformations. This proof is adapted from \cite{Gilyen_2018_arxiv}[Lem. 22] and included here for completeness.

\begin{lemma}[Extension to QSVT] Let $A\in\mathds{C}^{\tilde{d}\times {d}}$ be a matrix with $\norm{A}\leq 1$ and singular value decomposition $A=\sum_{j=1}^{\text{min}\{d,\tilde{d}\}}\sigma_j\,\ketbra{\tilde{u}_j}{{u}_j}$, where $\sigma_j\in \Sigma(A)$ and $\ket{u_j}$ ($\ket{\tilde{u}_j}$) are the right (left) singular vectors of $A$. Let $f:\mathds{R}\rightarrow\mathds{C}$ be a function with defined parity and operator Lipschitz constant $L_\star$. Define the singular value transformation of $A$ by $f$ as $f^{(SV)}(A):=\sum_{j=1}^{\text{min}\{d,\tilde{d}\}}f(\sigma_j)\,\ketbra{\tilde{u}_j}{{u}_j}$ if $f$ is odd, and $f^{(SV)}(A):=\sum_{j=1}^{{d}}f(\sigma_j)\ketbra{u_j}{{u}_j}$ if $f$ is even, with $\sigma_j=0$ if $j>\text{min}\{d,\tilde{d}\}$. Then, for $B\in\mathds{C}^{\tilde{d}\times {d}}$ with $\norm{B}\leq 1$, it holds that 
\begin{align}
    \norm{f^{(SV)}(A)-f^{(SV)}(B)}\leq L_\star\norm{A-B}.
\end{align}
\end{lemma}

\begin{proof}
Let $\bar{A}=\begin{psmallmatrix}0&A\\A^\dagger&0\end{psmallmatrix}$, which is Hermitian with spectrum $\{\pm\sigma_j\}\cup\{0\}\subset[-1,1]$ and $\norm{\bar{A}-\bar{B}}=\norm{A-B}$. If $f$ is odd, $f(\bar{A})=\begin{psmallmatrix}0&f^{(SV)}(A)\\f^{(SV)}(A)^\dagger&0\end{psmallmatrix}$. If $f$ is even, $f(\bar{A})=\mathrm{diag}\big(f(\sqrt{AA^\dagger}),f(\sqrt{A^\dagger A})\big)$ and $f^{(SV)}(A)=f(\sqrt{A^\dagger A})$. In both cases $f^{(SV)}(A)-f^{(SV)}(B)$ is a block of $f(\bar{A})-f(\bar{B})$, so its norm is at most $L_\star[f]\norm{\bar{A}-\bar{B}}=L_\star[f]\norm{A-B}$.
\end{proof}

\section{Haagerup factorizations of $K_f$}\label{app:haagerup-fact}

Here we prove Lemma \ref{lem:haagerup_poly}, which states that when $f$ is a polynomial or Laurent polynomial, the divided difference kernel \eqref{eq:df-kernel} admits a finite-dimensional Haagerup factorization \eqref{eq:haagerup}. We also state and prove two lemmas concerning Eqs. \eqref{eq:Hcheb} and \eqref{eq:Htrig} with the particular cases of polynomials in the Chebyshev basis and Laurent polynomials in the basis of monomials.

 \begin{proof}[Proof of Lem. \ref{lem:haagerup_poly}]

Let $\widehat{\mathbf{\Phi}}(x)=[\Phi_0(x),\ldots,\Phi_N(x)]^\top$ and
$\widehat{\mathbf{x}}=[1,x,\ldots,x^N]^\top$. By definition,
$\widehat{\mathbf{\Phi}}(x)=\widehat D\,\widehat{\mathbf{x}}$, where
$\widehat D$ is lower triangular and invertible since
$\deg(\Phi_j)=j$. Let $D$ be its leading $N\times N$ block, so that
$\mathbf{\Phi}(x)=D\mathbf{x}$, with
$\mathbf{x}=[1,x,\ldots,x^{N-1}]^\top$.

Writing $f(x)=\sum_{s=0}^N c_sx^s$, its monomial coefficients are
$c_s=\sum_{j=0}^N a_j\widehat D_{j,s}$. Hence
\begin{gather}
\begin{aligned}
K_f(x,y)
&=\frac{f(x)-f(y)}{x-y}
 =\sum_{s=1}^N c_s\sum_{\ell=0}^{s-1}x^{s-1-\ell}y^\ell \\
&=\sum_{k,k'=0}^{N-1}G'_{k,k'}x^ky^{k'}
 =\mathbf{x}^\top G'\mathbf{y},
\end{aligned}
\end{gather}
where
\begin{align*}
    G'_{k,k'}=
\begin{cases}
c_{k+k'+1}
=\displaystyle\sum_{j=0}^N
\widehat D_{j,k+k'+1}a_j,
& k+k'+1\leq N,\\[2mm]
0,
& k+k'+1>N.
\end{cases}
\end{align*}

Since $\mathbf{x}=D^{-1}\mathbf{\Phi}(x)$ and, for
$y\in\mathcal D\subseteq\mathbb R$,
$\mathbf{y}=\overline D^{-1}\mathbf{\Phi}(y)^*$, it follows that
\begin{gather}
\begin{aligned}
K_f(x,y)
&=\mathbf{\Phi}(x)^\top
(D^{-1})^\top G'\overline D^{-1}\mathbf{\Phi}(y)^* \\
&=\mathbf{\Phi}(x)^\top G\mathbf{\Phi}(y)^*,
\end{aligned}
\end{gather}
with $G=(D^{-1})^\top G'\overline D^{-1}$.

Now take any rank-one decomposition
$G=\sum_k\mathbf h_{p_k}\mathbf h_{q_k}^\dagger$ and define
$p_k(x)=\mathbf h_{p_k}^\top\mathbf{\Phi}(x)$ and
$q_k(y)=\mathbf h_{q_k}^\top\mathbf{\Phi}(y)$. Then
\begin{gather}
\begin{aligned}
\sum_k p_k(x)q_k(y)^*
&=\mathbf{\Phi}(x)^\top
\left(\sum_k\mathbf h_{p_k}\mathbf h_{q_k}^\dagger\right)
\mathbf{\Phi}(y)^* \\
&=\mathbf{\Phi}(x)^\top G\mathbf{\Phi}(y)^*
=K_f(x,y).
\end{aligned}
\end{gather}
Such a decomposition always exists, for instance from the singular
value decomposition of $G$. Since it is finite and the basis functions
are bounded on $\mathcal D$, the corresponding square sums are bounded.
Therefore $K_f$ is Haagerup factorizable.
\end{proof}

\begin{lemma}[Haagerup factorization in the Chebyshev basis]\label{thm:haagerup-cheb}
    Let $f(x) = \sum^{N}_{j=0} a_j T_j(x)$ be a polynomial in the Chebyshev basis with domain $\mathcal{D}=[-1,1]$.
    Then its divided difference kernel $K_f( x, y)$ is Haagerup factorizable. Moreover, it admits the factorization
    given by $\mathbf{p}( x)=\mathbf{\Phi}( x)$, $\mathbf{q}( y) \coloneq G \cdot \mathbf{\Phi}( y)$,
    where $\mathbf{\Phi}(x)\coloneq\left[ T_k(x) \right]_{k=0}^{N-1} \in [-1,1]^{N}$ and $G \in \mathbb{C}^{N \times N}$ is a symmetric matrix with components ($\delta_{i,j}$ is the Kronecker delta) 
    \begin{align}\label{eq:Hcheb-repeat}
    G_{k,k'} = (2-\delta_{k,0})(2-\delta_{k',0})\sum_{j=0}^{\lfloor \frac{N-k-k'-1}{2}\rfloor} \hspace{-0.1cm}a_{k+k'+1+2j}, 
    \end{align}
    for $k+k'+1+2j\leq N$, and $G_{k,k'}=0$ otherwise.
\end{lemma}

\begin{proof}
For a fixed Chebyshev polynomial $T_s(x)$, the divided difference can be written as \cite[Thm. 5]{kim_difference_2016}
\begin{align*}
    \frac{T_{s}( x)-T_{s}( y)}{ x- y}= \sum^{s-1 }_{k=0}(2-\delta_{k,0})\,  T_{k}( x)U_{s-k-1}( y).
\end{align*}
Notice that this also includes $ x= y$ through the limit $ y\to x$, in which case by definition we have the derivative $T_s^\prime( x)$. This immediately gives rise to a Haagerup decomposition of the form \eqref{eq:haagerup}, namely
\begin{align}\label{eq:cheb_decomp1}
\nonumber \frac{f( x)-f( y)}{ x- y}&= \sum^{N}_{s=1} a_s \frac{T_{s}( x)-T_{s}( y)}{ x- y}\\
   \nonumber & =\sum^{N}_{s=1} a_s \sum^{s-1 }_{k=0}(2-\delta_{k,0})\,  T_{k}( x)U_{s-k-1}( y)\\
    &= \sum^{N-1 }_{k=0}\sum^{N-1 }_{k'=0} \tilde{G}_{k,k'}\, T_{k}( x)\,U_{k'}( y)
\end{align}
where $\tilde{G}_{k,k'}=(2-\delta_{k,0})\, a_{k+k'+1} \, \Theta(N-(k+k'+1))$ with $\Theta(x)=1$ for $x\geq0$ and $\Theta(x)=0$ otherwise. We can get another decomposition using only the Chebyshev of the first kind as the basis for the polynomial functions of both variables, \emph{i.e.}, we can use the identity
\begin{align*}
U_k( y)=\sum_{\substack
{k'=0 \\ k'\equiv k\,\text{mod}\,2}}^k (2-\delta_{k',0})\,T_{k'}( y),
\end{align*}
from which we obtain
\begin{align} \label{eq:cheb_decomp2}
\nonumber \frac{f( x)-f( y)}{ x- y}&=  \sum^{N-1 }_{k'=0}\sum^{N-1 }_{k=0} \tilde{G}_{k,k'}\, T_{k}( x)\,U_{k'}( y)\\
\nonumber &=  \sum^{N-1 }_{k=0}\sum^{N-1 }_{k'=0} \tilde{G}_{k,k'}\, T_{k}( x)\hspace{-13pt} \sum_{\substack
{k''=0 \\ k''\equiv k'\,\text{mod}\,2}}^{k'} \hspace{-13pt} (2-\delta_{k'',0})\,T_{k''}\!( y)\\
&= \sum^{N-1 }_{k=0}\sum^{N-1 }_{k'=0} G_{k,k'} T_{k}( x) T_{k'}( y)
\end{align}
with $G_{k,k'}=(2-\delta_{k',0}) \sum_{\substack
{k''=k' \\ k''\equiv k'\,\text{mod}\,2}}^{N-1}  \tilde{G}_{k,k''}=(2-\delta_{k,0})(2-\delta_{k',0})\sum_{s=0}^{\lfloor \frac{N-k-k'-1}{2}\rfloor} a_{k+k'+1+2s}$.
Finally, the claim follows after introducing $\mathbf{p}( x)=\left[ T_k( x) \right]_{k=0}^{N-1}$, $\mathbf{q}^*( y)=\left[\sum^{N-1 }_{k'=0} G_{k,k'} \, T_{k'}( y)\right]_{k=0}^{N-1}=G\cdot \left[ T_{k'}( y) \right]_{k'=0}^{N-1}$, and $G=\big[G_{k,k'}\big]_{k,k'=0}^{N-1}$.
\end{proof}

 The proof above highlights the non-uniqueness of Haagerup decompositions and the role of finding the optimal decomposition in determining the operator Lipschitz constant $L_\star$. Both Eqs. \eqref{eq:cheb_decomp1} and \eqref{eq:cheb_decomp2} give valid decompositions for $K_f$. While motivated by symmetry we chose to present the form \eqref{eq:cheb_decomp2} in the theorem, Eq. \eqref{eq:cheb_decomp1} is also a valid decomposition for $K_f$
 with the same $p_k( x)=T_k( x)$ and $q^*_{k}( y)=(2-\delta_{k,0})\sum_{k'=0}^{N-k-1}a_{k'+k+1} U_{k'}( y)$. Both decompositions give explicit, if crude, upper bounds on $L_\star$ in terms of $\norm{a}_1\coloneq\sum_{s=0}^N|a_s|$. In both, $\norm{\mathbf{p}(x)}_2^2=\sum_{k=0}^{N-1}T_k(x)^2\le N$. For Eq.~\eqref{eq:cheb_decomp1}, $\vert U_{k'}(y)\vert\le k'+1$ on $[-1,1]$ gives $\vert q_0(y)\vert\le\sum_{k'=0}^{N-1}(k'+1)\vert a_{k'+1}\vert\le N\norm{a}_1$ and $\vert q_k(y)\vert\le 2(N-k)\norm{a}_1$ for $k\ge1$, so that
\begin{align*}
N\sup_y\norm{\mathbf{q}(y)}_2^2
&\le\norm{a}_1^2\Big[N^3+4N\sum_{k=1}^{N-1}(N-k)^2\Big]\\
&=\norm{a}_1^2\Big[N^3+\tfrac{2}{3}N^2(N-1)(2N-1)\Big]\\
&\le\tfrac{8}{3}N^4\norm{a}_1^2,
\end{align*}
because the difference of the two sides equals $\big(\tfrac43N^4+N^3-\tfrac23N^2\big)\norm{a}_1^2\ge0$. Hence $L_\star\le\sqrt{8/3}\,N^2\norm{a}_1$. For Eq.~\eqref{eq:cheb_decomp2}, $\vert G_{k,k'}\vert\le4\norm{a}_1$ gives $\vert q_k(y)\vert\le\sum_{k'}\vert G_{k,k'}\vert\le4N\norm{a}_1$, so $\sup_y\norm{\mathbf{q}(y)}_2\le4N^{3/2}\norm{a}_1$ and $L_\star\le4N^2\norm{a}_1$. Both bounds are far from tight, they only show that $L_\star=\mathcal{O}(N^2\norm{a}_1)$.

\begin{lemma}[Haagerup factorization of trigonometric polynomials]\label{thm:haagerup-trig-poly}
    Let $f(z) = \sum_{j=-N/2}^{N/2} c_j z^j$, with $N$ even, be a finite trigonometric polynomial with domain $\mathcal{D}=\mathbb{T}\coloneq \{z\in\mathbb{C}\text{ s.t. } \vert z \vert = 1\}$. Then its divided difference kernel $K_f(\zeta,\eta)$ is Haagerup factorizable. Moreover, it admits the
     factorization given by $\mathbf{p}(\zeta)=\mathbf{\Phi}(\zeta)$, $\mathbf{q}^*(\eta)= G\cdot\mathbf{\Phi}(\eta)$, where $\mathbf{\Phi}(z)\coloneq \left[ z^{k-N/2} \right]_{k=0}^{N-1} \in \mathbb{T}^{N}$ and $G \in \mathbb{C}^{N \times N}$ is a block-diagonal matrix with two Hankel blocks, with components (here $ j=k+k'-N+1$)
    \begin{align}\label{eq:Htrig-repeat}
        G_{k,k'}&=\begin{cases}
        \phantom{-}c_j, & 0<j\leq N/2,\quad 0\leq k-N/2 \leq j-1 \\
        -c_j, & -N/2\leq j<0,\quad j\leq k-N/2 \leq -1 \\
        \phantom{-}0, & \text{otherwise}
    \end{cases}\!.
    \end{align}
\end{lemma}

\begin{proof}
From the telescopic sum of commuting variables, we have
\begin{align*}
    \frac{\zeta^s-\eta^s}{\zeta-\eta} &=\begin{cases}\sum_{r=0}^{s-1}\zeta^{s-1-r}\eta^r\quad &s>0, \\
    -\sum_{r=0}^{-s-1}\zeta^{s+r}\eta^{-r-1} \quad &s<0.\end{cases}
\end{align*}
Then, for $K_f$,
\begin{gather}\label{eq:fact_kernel_fourier}
    \begin{aligned}
        K_f(\zeta,\eta)&=\frac{f(\zeta)-f(\eta)}{\zeta-\eta} \\
        &= \sum_{s=1}^{N/2} c_s \sum_{r=0}^{s-1}\zeta^{s-1-r}\eta^r - \sum_{s=-1}^{-N/2} c_{s}\sum_{r=0}^{-s-1}\zeta^{s+r}\eta^{-1-r}.
    \end{aligned}
\end{gather}
Notice that, for every $s$ in the external sum, if we add the exponents $\alpha$ of $\zeta$ and $\beta$ of $\eta$ we get $\alpha+\beta=s-1$, an integer in $[-N/2-1,N/2-1]$ different from $-1$. Moreover, each exponent takes values in $[-N/2,N/2-1]$; for $s>0$ both exponents are nonnegative, and for $s<0$ both are negative. Hence, by defining $\tilde{c}_{\alpha,\beta}=c_{\alpha+\beta+1}$ if $0<\alpha+\beta+1\leq N/2$ and $\alpha,\beta\geq0$, $\tilde{c}_{\alpha,\beta}=-c_{\alpha+\beta+1}$ if $-N/2\leq \alpha+\beta+1<0$ and $\alpha,\beta\leq-1$, and $\tilde{c}_{\alpha,\beta}=0$ otherwise, we can write
\[
       K_f(\zeta,\eta) = \sum_{\alpha=-N/2}^{N/2-1}\sum_{\beta=-N/2}^{N/2-1} \tilde{c}_{\alpha,\beta}\,\zeta^{\alpha}\,\eta^{\beta}.
\]
With the re-indexing $\alpha=k-N/2$, $\beta=k'-N/2$, the double sum becomes
\[
       K_f(\zeta,\eta) = \sum_{k=0}^{N-1}\sum_{k'=0}^{N-1} G_{k,k'}\zeta^{k-N/2}\eta^{k'-N/2},
\]
with $G_{k,k'}$ as defined in Eq. \eqref{eq:Htrig-repeat}.
The factorization claimed then follows after introducing $p_k(\zeta)= \zeta^{k-N/2} $, $q_k^*(\eta)= \sum_{k'=0}^{N-1}G_{k,k'}\,\eta^{k'-N/2} $, for ${k=0,\ldots,N-1} $.
\end{proof}

\section{Proof of Lemma \ref{lem:lips-schur}}
\label{app:proof-lem-lips-schur}

Before proceeding with the proof, let us define the fixed-dimension OL constant
\begin{equation}
L^{(d)}[f]\coloneq\sup_{\substack{A,B\in\mathcal{L}_{\mathcal{D}}(\mathcal{H}_d)\\A\neq B}}\frac{\norm{f(A)-f(B)}}{\norm{A-B}}
\end{equation}
such that
\begin{equation}
L_\star[f]=\sup_{d\in\mathbb{N}}L^{(d)}[f].
\end{equation}
The proof of Lemma \ref{lem:lips-schur} is based on identifying an analogous fixed-dimension Schur multiplier norm of $K_f$ in Eq.~\eqref{eq:df-kernel} as both a lower and an upper bound for $L^{(d)}[f]$ for different values of $d$. This implies that when we take the supremum over the dimension, the bounds in opposite directions force the equality $L_\star[f]=\norm{K_f}_{\rm{Schur}}$. We state and prove the fixed-dimension result next, followed by the proof of Lemma \ref{lem:lips-schur}.

 We present below the notation, definitions, and previous results needed for the proof.
 
\paragraph{Notation:} 
Let $\mathcal{I}$ denote a general open interval in $\mathbb{R}$.
For every $n\in\mathbb{N}$, let  $\mathcal{U}_{n}(\mathbb{C})$ and $\mathcal{H}_{n}(\mathbb{C})$ be the sets of $n$-dimensional  unitary and Hermitian matrices over $\mathbb{C}^{n}$, respectively.
Moreover, for any set $\mathsf{Z}\subset\mathbb{C}$ let $\mathcal{N}_{n}(\mathsf{Z})$ denote the set of all normal operators over $\mathbb{C}^{n}$ whose spectrum is contained in $\mathsf{Z}$, notice that $\mathcal{N}_{n}(\mathbb{R})=\mathcal{H}_{n}(\mathbb{R})$ and $\mathcal{N}_{n}(\mathbb{T})=\mathcal{U}_{n}(\mathbb{C})$ by definition.
The symbol $\circ$ is used only for the Schur product (also called Hadamard, or point-wise product) of matrices.

\begin{definition}
Let $f\in C^{1}(\mathcal{I})$ and $A,H\in\mathcal{N}_{n}(\mathcal{I})$.
The \textbf{Frechet derivative of $f$ at $A$} is the linear map $\mathfrak{D}f(A):\mathcal{N}_{n}(\mathcal{I})\rightarrow\mathcal{N}_{n}(\mathcal{I})$ such that
\begin{equation}
\lim_{H\rightarrow0}\frac{\norm{f(A+H)-f(A)-\mathfrak{D}f(A)[H]}}{\norm{H}}=0.
\end{equation}
The value of $\mathfrak{D}f(A)$ at the matrix $H$ is called the \textit{Gateaux derivative of $f$ at $A$ in the direction $H$} and can also be written as the limit
\begin{equation}
\mathfrak{D}f(A)[H]=\lim_{t\rightarrow0}\frac{f(A+tH)-f(A)}{t}=\frac{d}{dt}f(A+tH)\Big|_{t=0}.
\end{equation}
A analogous definitions of both derivatives exist for functions $f\in C^{1}(\mathbb{T})$ in \cite{chattopadhyay_higher_2024}.
\end{definition}

Since all functions here are $C^{1}$ on the spectrum of the operators they are applied over, it turns out that they are always Frechet differentiable \cite{bhatia_positive_2007,chattopadhyay_higher_2024}.
Therefore, for such functions, all Gateaux derivatives exist where they are well-defined, and they can be computed using the formulas for Frechet derivatives.

\begin{theorem}\label{thm:Daleckii-Krein_open}
Let $f\in C^{1}(\mathcal{I})$ and $A,H\in\mathcal{N}_{n}(\mathcal{I})$.
Also, let $(\lambda^{A}_{i},\Pi^{A}_{i})_{i\in[n]}$ be the spectral decomposition for $A$.
\begin{equation}
\mathfrak{D}f(A)[H]=\sum_{i,j\in[n]}K_{f}(\lambda^{A}_{i},\lambda^{A}_{j})\,(\Pi^{A}_{i}H\,\Pi^{A}_{j}).
\end{equation}
Or yet, if we denote $H^{A}$ to be the matrix with $H$ in the eigenbasis of $A$ and $K^{\mathcal{A},\mathcal{A}}_{f}$ to be the matrix with entries $[K^{\mathcal{A},\mathcal{A}}_{f}]_{i,j}=K_{f}(\lambda^{A}_{i},\lambda^{A}_{j})$,
\begin{equation}
\mathfrak{D}f(A)[H]=K^{\mathcal{A},\mathcal{A}}_{f}\circ H^{A}.
\end{equation}
\end{theorem}

\begin{proof}
The proof of this statement is given bellow \cite{bhatia_positive_2007}[Theorem 5.3.1]
\end{proof}

\begin{corollary}\label{cor:Daleckii-Krein_closed}
Let $f\in C^{1}([a,b])$, $A\in\mathcal{N}_{n}([a,b])$ and $H\in\mathcal{N}_{n}(\mathbb{R})$.
Also, let $\tilde{f}:\mathbb{R}\rightarrow\mathbb{C}$ be the canonical linear extension of $f$.
Then $\tilde{f}$ is Frechet differentiable on every $A\in\mathcal{N}_{n}([a,b])$ and its derivative is also
\begin{equation}
\mathfrak{D}\tilde{f}(A)[H]=K^{\mathcal{A},\mathcal{A}}_{f}\circ H^{A},
\end{equation}
but for all $H\in\mathcal{N}_{n}(\mathbb{R})$.
\end{corollary}

\begin{proof}
Just apply the previous theorem to $\tilde{f}\in C^{1}(\mathbb{R})$.
\end{proof}

\begin{theorem}\label{thm:Daleckii-Krein_torus}
Let $f\in C^{1}(\mathbb{T})$, $U\in\mathcal{N}_{n}(\mathbb{T})$ and $H\in\mathcal{N}_{n}(\mathbb{R})$.
Also, let $U_{t}=e^{itH}U$ with $U^{\prime}_{0}=iHU$ and let $(\zeta_{i},\Pi^{U}_{i})_{i\in[n]}$ be the spectral decomposition for $U$.
Then
\begin{equation}
\mathfrak{D}f(U)[U^{\prime}_{0}]=\sum_{i,j\in[n]}K_{f}(\zeta_{i},\zeta_{j})\,(\Pi_{i}(iHU)\,\Pi_{j}).
\end{equation}
Or yet, if we denote $W$ to be the matrix $U^{\prime}_{0}$ in the eigenbasis of $U$ and $K_{f,n}(\mathsf{U},\mathsf{U})$ to be the matrix with entries $K_{f}(\zeta_{i},\zeta_{j})$,
\begin{equation}
\mathfrak{D}f(U)[U^{\prime}_{0}]=K_{f,n}(\mathsf{U},\mathsf{U})\circ W.
\end{equation}
\end{theorem}

\begin{lemma}[Fixed dimension OL and Schur multiplier norm]\label{lem:fund_lemma}
Let $f\in C^{1}(\mathcal{D})$, and let $\mathsf{X},\mathsf{Y}\in\mathcal{D}^d$ for any $d\in\mathbb{N}$.
Also, for all $d\in\mathbb{N}$, let
\begin{equation}\label{eq:discrete-schur}
S^{(d)}[K_f]\coloneq\sup_{\mathsf{X},\mathsf{Y}\in\mathcal{D}^d}\underset{\norm{Z}=1}{\sup_{Z\in \mathbb{C}^{d\times d}}}\norm{K_{f,d}(\mathsf{X},\mathsf{Y})\circ Z},
\end{equation}
where $K_{f,d}(\mathsf{X},\mathsf{Y})$ is the $d$-dimensional matrix with coefficients $K_{f}(x_{i},y_{j})$.
Then the following holds
\begin{enumerate}[label=(\arabic*)]
    \item $L^{(d)}[f]\leq S^{(d)}[K_f]$ for all $d\in\mathbb{N}$.
    \item $S^{(d)}[K_f]\leq L^{(2d)}[f]$ for all $d\in\mathbb{N}$.
\end{enumerate}
\end{lemma}

\begin{proof}
To prove the lemma, we prove points $(1)$ and $(2)$ separately.

\begin{paragraph}{\textbf{Proof of point $(1)$:}}
For a given $d\in\mathbb{N}$, let $A=UD_{\mathsf{X}}U^{\dagger}$ and $B=VD_{\mathsf{Y}}V^{\dagger}$ where $D_{\mathsf{S}}=diag(s_{1},\dots,s_{d})$ for any $\mathsf{S}\in\mathcal{D}^{d}$, and also let $W=U^{\dagger}V$.
For $A,B$ as above, we have that
$$U^{\dagger}(A-B)V=D_{\mathsf{X}}W-WD_{\mathsf{Y}}.$$
Also, by the spectral theorem
$$U^{\dagger}(f(A)-f(B))V=f(D_{\mathsf{X}})W-Wf(D_{\mathsf{Y}}).$$
Whenever $x_{i}\neq y_{j}$ we have that,
\begin{equation*}
\begin{split}
[f(D_{\mathsf{X}})W-Wf(D_{\mathsf{Y}})]_{i,j}&=(f(x_{i})-f(y_{j}))W_{i,j}\\
=&K_{f}(x_{i},y_{j})(x_{i}-y_{j})W_{i,j}\\
=&[K_{f,d}(\mathsf{X},\mathsf{Y})\circ(D_{\mathsf{X}}W-WD_{\mathsf{Y}})]_{i,j}.
\end{split}
\end{equation*}
Notice that when $x_{i}=y_{j}$, the first and last expressions in the previous equation are $0$ regardless of the value of $K_{f}(x_{i},y_{j})$, making it true also in this case, and allowing us to write
$$f(D_{\mathsf{X}})W-Wf(D_{\mathsf{Y}})=K_{f,d}(\mathsf{X},\mathsf{Y})\circ(D_{\mathsf{X}}W-WD_{\mathsf{Y}}).$$
Then, from the multiplicative property of the operator norm and the fact that unitary matrices have unit norm, we get
\begin{equation*}
\begin{split}
\frac{\norm{f(A)-f(B)}}{\norm{A-B}}=\norm{K_{f,d}(\mathsf{X},\mathsf{Y})\circ\frac{D_{\mathsf{X}}W-WD_{\mathsf{Y}}}{\norm{D_{\mathsf{X}}W-WD_{\mathsf{Y}}}}}.
\end{split}
\end{equation*}
Now, notice that taking the supremum of the right-hand side over the set of all $A,B\in\mathcal{L}_{\mathcal{D}}(\mathcal{H}_d)$ with $A\neq B$ is the same as taking the supremum of the left-hand side over all of $\mathcal{D}^{d}$ and all of the set $\mathcal{R}(\mathsf{X},\mathsf{Y})$, where for every $\mathsf{X},\mathsf{Y}\in\mathcal{D}$
\begin{small}
$$\mathcal{R}(\mathsf{X},\mathsf{Y})\coloneq\left\{\left.\frac{D_{\mathsf{X}}W-WD_{\mathsf{Y}}}{\norm{D_{\mathsf{X}}W-WD_{\mathsf{Y}}}}\right|W\in\mathcal{U}_{d}(\mathbb{C}),D_{\mathsf{X}}W\neq WD_{\mathsf{Y}}\right\}$$
\end{small}
is some subset of the set $\{Z\in \mathbb{C}^{d\times d}|\norm{Z}=1\}$.
Therefore, since the supremum defining $L^{(d)}[f]$ is over a restricted subset of the set of $\{Z\in \mathbb{C}^{d\times d}|\norm{Z}=1\}$, we have $L^{(d)}[f]\leq S^{(d)}[K_f]$ as we wanted to show.
\end{paragraph}

\begin{paragraph}{\textbf{Proof of point $(2)$:}}
For a given $d\in\mathbb{N}$ consider diagonal matrices $D_{\mathsf{X}}$ and $D_{\mathsf{Y}}$, with $\mathsf{X},\mathsf{Y}\in\mathcal{D}^d$.
To prove $(2)$ we will need to consider two cases:
\begin{enumerate}[label=(\roman*)]
    \item When $\mathcal{D}=\mathbb{T}$.
    \item When $\mathcal{D}=[a,b]$ for $a,b\in\mathbb{R}$ with $a<b$.
\end{enumerate}

\begin{subparagraph}{\textbf{Proof of sub-case $(i)$}}:
Let $\tilde{Z}\in \mathbb{C}^{d\times d}$ be arbitrary, and $G,H\in\mathcal{L}_{\mathcal{D}}(\mathcal{H}_{2d})$ be given by
$$G=
\begin{bmatrix}
D_{\mathsf{X}} & 0\\
0 & D_{\mathsf{Y}}
\end{bmatrix}
\text{ and }
H=
\begin{bmatrix}
0 & -i\tilde{Z}D_{\mathsf{Y}}^{\dagger}\\
iD_{\mathsf{Y}}\tilde{Z}^{\dagger} & 0
\end{bmatrix}
$$
Notice that $D_{\mathsf{X}}$, $D_{\mathsf{Y}}$ and $G$ are all unitary, while $H$ is Hermitian.
The spectrum of $G$ is $\mathsf{G}=\mathsf{X}\sqcup\mathsf{Y}$ and $\norm{H}=\norm{\tilde{Z}}$ as it is an isometric, Hermitian dilation of $-i\tilde{Z}D_{\mathsf{Y}}^{\dagger}$.
Now, let $U_{t}\coloneq e^{itH}G$.
Then, by the definition of $L^{(2d)}[f]$, we have in particular that for all $t\in\mathbb{R}$,
$$\norm{f(U_{t})-f(U_{0})}\leq L^{(2d)}[f]\norm{U_{t}-U_{0}}.$$
Therefore, by taking the limit $t\rightarrow0$ on both sides we get
$$\lim_{t\rightarrow0}\norm{\frac{f(U_{t})-f(U_{0})}{t}}\leq L^{(2d)}[f]\lim_{t\rightarrow0}\,\norm{\frac{U_{t}-U_{0}}{t}}.$$
The left-hand side is by definition $U^{\prime}_{0}$, and a simple computation give us
$$U^{\prime}_{0}=iHG=
\begin{bmatrix}
0 & \tilde{Z}\\
-D_{\mathsf{Y}}\tilde{Z}^{\dagger}D_{\mathsf{X}} & 0
\end{bmatrix}
,$$
from which we also get $\norm{U^{\prime}_{0}}=\norm{\tilde{Z}}$.
The limit on the right-hand side is by definition the {Gateaux derivative of $f$} at $U_{0}$ in the direction of $U^{\prime}_{0}$, and since $U_{0}=G$ is diagonal, Thm.\ref{thm:Daleckii-Krein_torus} gives us
\begin{align}
\mathfrak{D}f(U_{0})[U^{\prime}_{0}]&=K_{f,2d}(\mathsf{G},\mathsf{G})\circ U^{\prime}_{0}\notag\\
&=
\begin{bmatrix}
0 & K_{f,d}(\mathsf{X},\mathsf{Y})\circ \tilde{Z}\\
-K_{f,d}(\mathsf{X},\mathsf{Y})\circ(D_{\mathsf{Y}}\tilde{Z}^{\dagger}D_{\mathsf{X}}) & 0
\end{bmatrix}    
\end{align}
Using the dominance of the operator norm of a matrix over its sub-matrices we get
$$\norm{K_{f,d}(\mathsf{X},\mathsf{Y})\circ \tilde{Z}}\leq\norm{\mathfrak{D}f(U_{0})[U^{\prime}_{0}]}\leq L^{(2d)}[f]\norm{\tilde{Z}}.$$
for all $\tilde{Z}\in \mathbb{C}^{d\times d}$ and $\mathsf{X},\mathsf{Y}\in\mathcal{D}$.
Then, by redefining $Z=\tilde{Z}/\norm{\tilde{Z}}$ and taking the supremum of over all $\mathsf{X},\mathsf{Y}\in\mathcal{D}$ and all $Z\in \mathbb{C}^{d\times d}$ with $\norm{Z}=1$ we obtain $S^{(d)}[K_f]\leq L^{(2d)}[f]$. 
\end{subparagraph}

\begin{subparagraph}{\textbf{Proof of sub-case $(ii)$}}
Now, let $\tilde{Z}\in \mathbb{C}^{d\times d}$ be arbitrary and let $G,H\in\mathcal{L}_{\mathcal{D}}(\mathcal{H}_{2d})$ be given by
$$G=\begin{bmatrix}
D_{\mathsf{X}} & 0\\
0 & D_{\mathsf{Y}}
\end{bmatrix}\text{ and }
H=\begin{bmatrix}
0 & \tilde{Z}\\
\tilde{Z}^{\dagger} & 0
\end{bmatrix}$$ 
Again, let the spectrum of $G$ be denote by $\mathsf{G}=\mathsf{X}\sqcup\mathsf{Y}$ and notice that $\norm{H}=\norm{\tilde{Z}}$, as it is an isometric Hermitian dilation of $\tilde{Z}$.
For now, assume that every element of $\mathsf{X}$ and $\mathsf{Y}$ are in $(a,b)$, in other words, there is a $\varepsilon>0$ distance from the set union of all eigenvalues of $A$ and $B$ to the set $\{a,b\}$.
Then, consider the operator $G+tH$.
By Weyl's spectral perturbation theorem, the set-distance between the spectrum of $G$ and the spectrum of $G+tH$ is $\abs{t}\norm{H}=\abs{t}\norm{\tilde{Z}}$ for all $t\in\mathbb{R}$.
Therefore, if for every $\tilde{Z}\in \mathbb{C}^{d\times d}$ we impose $\abs{t}<\varepsilon/\norm{\tilde{Z}}$, we have $(G+tH)\in\mathcal{L}_{n}((a,b))$ and the definition of $L^{(2n)}_{f}((a,b))$ implies the particular inequality
$$\norm{f(G+tH)-f(G)}\leq\abs{t}L^{(2d)}_{f}((a,b))\norm{\tilde{Z}}.$$
By passing to the limit $t\rightarrow0$ we get 
$$\norm{\lim_{t\rightarrow0}\frac{f(G+tH)-f(G)}{t}}\leq L^{(2d)}_{f}((a,b))\norm{\tilde{Z}}.$$
Again, since $f$ is $C^{1}((a,b))$, Thm.\ref{thm:Daleckii-Krein_open} gives
\begin{align*}
    \mathfrak{D}f(G)[H]&=K_{f,2d}(\mathsf{G},\mathsf{G})\circ H\\
    &=
\begin{bmatrix}
0 & K_{f,d}(\mathsf{X},\mathsf{Y})\circ \tilde{Z}\\
K_{f,d}(\mathsf{X},\mathsf{Y})\circ \tilde{Z}^{\dagger} & 0
\end{bmatrix}
.
\end{align*}
As in the previous case, the dominance of the operator norm of a matrix over its sub-matrices implies that $$\norm{K_{f,d}(\mathsf{X},\mathsf{Y})\circ \tilde{Z}}\leq\norm{\mathfrak{D}f(G)[H]}\leq L^{(2d)}[f]\norm{\tilde{Z}},$$
for every $\tilde{Z}\in \mathbb{C}^{d\times d}$ and for all $\mathsf{X},\mathsf{Y}\in(a,b)^{d}$.
To cover the endpoints, notice that by definition $L^{(2n)}_{f}((a,b))\leq L^{(2n)}_{f}([a,b])$.
Also, since $f\in C^{1}([a,b])$, its linear canonical extension $\tilde{f}$ must coincide with $f$ in all points on $[a,b]$, and we have that $K_{f}=K_{\tilde{f}}|_{[a,b]^{2}}$, and in particular that $L^{(2n)}_{f}((a,b))=L^{(2n)}_{\tilde{f}}((a,b))$.
Therefore, we must also have
$$\norm{\tilde{f}(G+tH)-\tilde{f}(G)}\leq\abs{t}L^{(2n)}_{f}((a,b))\norm{\tilde{Z}}$$
Now, if the spectrum of either $A$ or $B$ contains one or both of the endpoints $a,b\in\mathbb{R}$, we can construct sequences of matrices; $\{A_{k}\}_{k\in\mathbb{N}}$ converging to $A$ and $\{B_{l}\}_{l\in\mathbb{N}}$ converging to $B$, both in operator norm, such that, for all $k,l\in\mathbb{N}$, their spectra lies entirely in $(a,b)$.
Using the previous conclusion, this implies that
$$\norm{K^{\mathsf{X}_{k},\mathsf{Y}_{l}}_{f}\circ \tilde{Z}}\leq L^{(2n)}_{f}((a,b))\norm{\tilde{Z}}\leq L^{(2n)}_{f}([a,b])\norm{\tilde{Z}}$$
for all $k,l\in\mathbb{N}$.
Since by Cor.\ref{cor:Daleckii-Krein_closed} the canonical linear extension $\tilde{f}$ is also Frechet differentiable, the limits of all such sequences must coincide with the Frechet derivative of $\tilde{f}$ 
$$\norm{K_{f,d}(\mathsf{X},\mathsf{Y})\circ \tilde{Z}}=\lim_{k,l\rightarrow\infty}\norm{K^{\mathsf{X}_{k},\mathsf{Y}_{l}}_{f}\circ \tilde{Z}}\leq L^{(2n)}_{f}([a,b])\norm{\tilde{Z}},$$
for every $\tilde{Z}\in \mathbb{C}^{d\times d}$ and for all $\mathsf{X},\mathsf{Y}\subset[a,b]$.
Then, by the change-of-variables $Z=\tilde{Z}/\norm{\tilde{Z}}$, and taking the supremum of over all $\mathsf{X},\mathsf{Y}\in(a,b)^{d}$ and all $Z\in \mathbb{C}^{d\times d}$ with $\norm{Z}=1$ we obtain $S^{(d)}[K_{f}]\leq L^{(2d)}[f]$. 
\end{subparagraph}
\end{paragraph}
\end{proof}

\begin{proof}[Proof of Lemma \ref{lem:lips-schur}]
Using Lemma \ref{lem:fund_lemma} we have for every $d\in\mathbb{N}$
\begin{equation*}\label{eq:Lipschitz_sandwich}
L^{(d)}[f]\leq S^{(d)}[K_f]\leq L^{(2d)}[f].
\end{equation*}
Taking the supremum over $\mathbb{N}$ gives us
\begin{equation*}\label{eq:supreme_Lipschitz_sandwich}
L_\star[f]=\sup_{d\in\mathbb{N}}L^{(d)}[f]\leq\sup_{d\in\mathbb{N}} S^{(d)}[K_f]\leq\sup_{d\in\mathbb{N}} L^{(2d)}[f].
\end{equation*}
Now, since $\{L^{(2d)}[f]\}_{d\in\mathbb{N}}\subseteq\{L^{(d)}[f]\}_{d\in\mathbb{N}}$ we have that
\begin{equation*}
\sup_{d\in\mathbb{N}}L^{(2d)}[f]\leq\sup_{d\in\mathbb{N}} L^{(d)}[f]=L_\star[f],
\end{equation*}
forcing equality of all four quantities, and concluding the proof.


\end{proof}

\section{Proof of the Theorem \ref{thm:OL-haagerup}}
\label{app:proof-main}

To prove Thm.~\ref{thm:OL-haagerup}, which connects $L_\star[f]$ with the Haagerup norm of $K_f$, we begin by stating the equality between the Schur multiplier norm of $K_f$ and its Haagerup norm on finite sets of points in the domain (Lemma~\ref{lem:discrete-schur-haagerup} below,  and \cite{Pisier2001}[Thm. 5.1]). Then we use the compactness arguments and the equivalence between $L_\star$ and the Schur multiplier norm from Lemma~\ref{lem:lips-schur} to conclude the proof. Because for the cases of interest Lemma~\ref{lem:haagerup_poly} holds, we assume that the divided-difference kernel admits a Haagerup
factorization. More generally, every scalar kernel in two variables
with finite Schur multiplier norm admits such a factorization (see  \cite[Theorem~3.1]{aleksandrov2024}). 

Before proceeding with the proof, we build intuition for the relation between the OL constant of a function and the Haagerup factorization of its divided difference kernel. In the next section, we prove that the right-hand side of Eq.~\eqref{eq:Lstar} is an upper bound on $L_\star[f]$ using only familiar linear algebra tools.

\subsection{Upper bound on $L_\star$ using Haagerup factorization}

Using the definition of the OL constant in Eq. \eqref{eq:op-lip} directly, here we state and prove the following particular case of Thm. \ref{thm:OL-haagerup}:

\begin{proposition}
    Let 
$\mathcal{D}\subset \mathbb{C}$ be a compact domain and $f:\mathcal{D}\to\mathbb{C}$ be a bounded function whose divided difference 
kernel $K_f$ is Haagerup-factorizable as in Eq.~\eqref{eq:haagerup}. 
Then $f$ is operator-Lipschitz on $\LD$ 
with optimal constant 
\begin{align}\label{eq:Lstar}
L_\star[f] \leq 
\inf_{\mathbf{p},\mathbf{q}} \left\{\sup_{x\in\mathcal{D}} \big\|\mathbf{p}(x)\big\|_2 \cdot \sup_{y\in\mathcal{D}} \big\|\mathbf{q}(y)\big\|_2 \right\},
\end{align}
where the infimum is over all $\mathbf{p},\mathbf{q}$ such that Eq.~\eqref{eq:haagerup} holds. 
\end{proposition}

\begin{proof}
Let $A = \sum_{i\in[d]} x_i\ketbra{x_i}{x_i}$ and $B = \sum_{j\in[d]} y_j\ketbra{y_j}{y_j}$ be the spectral decompositions of $A$ and $B$, with $x_i,y_j\in\mathcal{D}$ for all $i,j\in[d]$. 
Using resolutions of the identity in terms of the eigenvectors of $A$ and $B$ one can write 
\begin{align}\label{eq:fAfBkernel}
f(A)-&f(B) =
\sum_{i\in[d]} f(x_i)\ketbra{x_i}{x_i}\cdot\sum_{j\in[d]} \ketbra{y_j}{y_j} \notag\\
&\qquad- \sum_{i\in[d]} \ketbra{x_i}{x_i}\cdot\sum_{j\in[d]} f(y_j)\ketbra{y_j}{y_j} \notag\\
&= \sum_{i,j\in[d]} \big[f(x_i)-f(y_j)\big]\ketbra{x_i}{x_i}\!y_j\rangle\!\!\bra{y_j} \notag\\
&= \sum_{i,j\in[d]} K_f(x_i,y_j)(x_i-y_j)\ketbra{x_i}{x_i}\!y_j\rangle\!\!\bra{y_j},
\end{align}
with $K_f$ given by Eq.~\eqref{eq:df-kernel}.
Using the Haagerup factorization condition in Eq.~\eqref{eq:haagerup}, this becomes
\begin{align}
f(A)-f(B) 
&= \sum_{k \geq 0} \sum_{i\in[d]} p_k(x_i)x_i\ketbra{x_i}{x_i}\sum_{j\in[d]} q^*_k(y_j)\ketbra{y_j}{y_j}  \notag\\
&- \sum_{k \geq 0} \sum_{i\in[d]} p_k(x_i)\ketbra{x_i}{x_i}\sum_{j\in[d]} y_j q^*_k(y_j)\ketbra{y_j}{y_j} \notag\\
&= \sum_{k \geq 0} P_k(A)(A-B)Q_k(B)\,,
\end{align}
where we have introduced the operators
\begin{gather}\label{eq:PkQk}
   \begin{aligned}
    P_k(A) &= \sum_{i\in[d]} p_k(x_i)\ketbra{x_i}{x_i},\\
    Q_k(B) &= \sum_{j\in[d]} q^*_k(y_j)\ketbra{y_j}{y_j}\,.
    \end{aligned} 
\end{gather}
Taking the spectral norm, we obtain
\begin{align}\label{eq:linear-algebra-OL}
\|f(A)&-f(B)\| 
 \coloneq \sup_{\ket{\phi},\ket{\psi}} \big\vert\!\bra{\phi}f(A)-f(B)\ket{\psi}\!\big\vert \notag\\
 & = \sup_{\ket{\phi},\ket{\psi}} \left\vert\sum_{k \geq 0}\bra{\phi} P_k(A)(A-B)Q_k(B)\ket{\psi}\right\vert \,,
\end{align}
where the supremum is taken over $d$-dimensional $\ell_2$-normalized states $\ket{\phi}$ and $\ket{\psi}$, with $\norm{\ket{\phi}}_2=\norm{\ket{\psi}}_2=1$.

Using the Cauchy-Schwarz inequality, the absolute value can be bounded as 
\begin{align}
\Bigg\vert&\sum_{k\ge0}\bra{\phi} P_k(A)(A-B)Q_k(B)\ket{\psi}\!\Bigg\vert \notag\\
&\le\norm{A-B}\sqrt{\sum_{k\ge0}\big\|P_k(A)\ket{\phi}\!\big\|_2^2}\sqrt{\sum_{k\ge0}\big\|Q_k(B)\ket{\psi}\!\big\|_2^2}. 
\end{align}
Now, it follows from \eqref{eq:PkQk} that
\begin{align}
\sup_{\ket{\phi}}\sqrt{\sum_{k\ge0}\big\|P_k(A)\ket{\phi}\!\big\|_2^2} &= \sup_{\ket{\phi}} \sqrt{\sum_{k\ge0}\sum_{i\in[d]} \big\vert p_k(x_i)\big\vert^2 \big\vert\!\braket{\phi} {x_i}\!\big\vert^2} \notag\\
&=\sup_{\ket{\phi}} \sqrt{\sum_{i\in[d]}\big\|\mathbf{p}(x_i)\big\|_2^2\,\big\vert\!\braket{\phi}{x_i}\!\big\vert^2} \notag\\
&=\sup_{i\in[d]} \big\|\mathbf{p}(x_i)\big\|_2\,,
\end{align}
where the last step uses the orthonormality of the eigenstates $\ket{x_i}$ to select the maximal value of $\norm{\mathbf{p}(x_i)}_2$ over all the $x_i$. 
Similarly, $\sqrt{\sum_{k \geq 0}\big\|Q_k(B)\big\|^2}=\sup_{j\in[d]} \big\|\mathbf{q}(y_j)\big\|_2$. 

Putting all together, we arrive at the operator-Lipschitz condition
\begin{align}\label{eq:beforeinf}
\|f(A)&-f(B)\| \le \norm{A-B}\cdot\sup_{i\in[d]} \big\|\mathbf{p}(x_i)\big\|_2 \cdot \sup_{j\in[d]} \big\|\mathbf{q}(y_j)\big\|_2 .
\end{align}
Since the Haagerup decomposition \eqref{eq:haagerup} is not unique, this gives actually a family of operator-Lipschitz conditions labeled by $\mathbf{p}$ and $\mathbf{q}$.  
The tightest one corresponds to the choice of $\mathbf{p},\mathbf{q}$ that minimizes the constant 
in the right hand side of Eq. \eqref{eq:beforeinf}. 
For this choice, Eq. \eqref{eq:beforeinf} implies
\begin{align}
L_\star \le \inf_{\substack{\mathbf{p},\mathbf{q}\,\,\text{ s.t.}\\K_f(x,y)=\mathbf{p}(x)\cdot\mathbf{q}(y)}} \left\{\sup_{x\in\mathcal{D}} \big\|\mathbf{p}(x)\big\|_2 \cdot \sup_{y\in\mathcal{D}} \big\|\mathbf{q}(y)\big\|_2 \right\}\,.
\end{align}
\end{proof}
\subsection{Schur and Haagerup norms on finite sets of points}

\begin{lemma}[Equality of the discrete Schur and Haagerup norms]
\label{lem:discrete-schur-haagerup}
Let $\mathcal{D}$ be either $[-1,1]$ or the unit circle $\mathbb{T}$.
Let $K:\mathcal D\times\mathcal D\to\mathbb C$ be a 
Haagerup factorizable function such that $K(x,y)=\boldsymbol{p}(x)^\top\boldsymbol{q}(y)^*,$
for  factors $\boldsymbol{p}(x),\boldsymbol{q}(y)$ with uniformly bounded vector norms in $\mathcal{D}$.
Fix finite sets of points
$\mathbf x=(x_1,\ldots,x_d)$ and
$\mathbf y=(y_1,\ldots,y_d)$ in $\mathcal D$, and let
$K_{d}(\mathbf x,\mathbf y)$ be the $d\times d$ matrix with entries
$K(x_i,y_j)$.
Define  the Schur multiplier and Haagerup norms of $K_{d}(\mathbf x,\mathbf y)$ as
\begin{align}
\|K_d(\mathbf x, \mathbf y)\|_{s}
&:=
\sup_{\substack{Z\in\mathbb C^{d\times d}\\ \|Z\|\le1}}
\|K_d(\mathsf{X},\mathsf{Y})\circ Z\|,\\
\|K_{d}(\mathbf x, \mathbf y)\|_h
&:=
\inf_{\boldsymbol{p},\boldsymbol{q}}
\left(\max_{1\le i\le d}\|\boldsymbol{p}(x_i)\|_2\right)
\left(\max_{1\le j\le d}\|\boldsymbol{q}(y_j)\|_2\right),
\end{align}
respectively.
Then
\begin{align}\label{eq:disc-schur-hag}
\norm{K_{d}(\mathsf{X},\mathsf{Y})}_s = \norm{K_{d}(\mathsf{X},\mathsf{Y})}_h.
\end{align}
\end{lemma}

\begin{proof}
For simplicity, we denote $K_{d}(\mathsf{X},\mathsf{Y})$ simply as $K_d$ throughout the proof.
If $K_d=0$, both norms vanish. Assume henceforth that $K_d\ne0$.

For any factorization vectors $\boldsymbol{p}(x_i),\boldsymbol{q}(y_j)$ of the sampled entries, define the operators $P_k=\operatorname{diag}(p_k(x_1),\cdots,p_k(x_d))$ and $Q_k=\operatorname{diag}(q_k(y_1),\cdots,q_k(y_d))$. Then for any $Z\in\mathbb{C}^{d\times d}$ we can write
\begin{gather}
    \begin{aligned}
        K_d\circ Z&= \sum_{i,j=1} K(x_i,y_j)\,Z_{i,j}\ketbra{i}{j}\\
        &=\sum_{i,j=1}^d \sum_k p_k(x_i)\,q_k(y_j)^* \,Z_{i,j}\ketbra{i}{j}\\
        &=\sum_k P_k\,Z\,Q_k.
    \end{aligned}
\end{gather} 
Consequently, similar to Eq. \eqref{eq:linear-algebra-OL} we get
\begin{align}
\|K_{d}\circ Z\|
\le
\left(\max_i\|\boldsymbol{p}(x_i)\|_2\right)
\left(\max_j\|\boldsymbol{q}(y_j)\|_2\right)\|Z\|.
\end{align}
Taking the supremum over $Z$ and the infimum over
factorizations proves
\begin{align}
\|K_{d}\|_s\le\|K_{d}\|_h.
\end{align}

For the reverse inequality, Gram factorization gives the
semidefinite representation
\begin{align}
\|K_{d}\|_h
 =
\min_{\substack{t>0\\P,Q}}&
 \quad t \\
 \text{s.t.} &
        \begin{cases}
            \begin{bmatrix}
                P&K_{d}\\
                K_{d}^\dagger&Q
            \end{bmatrix}\succeq0\\
            P_{ii}\le t,\quad Q_{jj}\le t,
        \end{cases}
 \label{eq:finite-haagerup-primal}
\end{align}
where $P,Q$ are $d\times d$ Hermitian matrices.
Indeed, any Haagerup factorization defines a positive semidefinite matrix as in the constraint with 
\begin{align}
P_{ii'}
&=\boldsymbol{p}(x_i)^\top\boldsymbol{p}(x_{i'})^*,
&
Q_{jj'}
&=\boldsymbol{q}(y_j)^\top\boldsymbol{q}(y_{j'})^*,
\end{align}
and the rescaling $\boldsymbol{p}\mapsto c\,\boldsymbol{p},\quad\boldsymbol{q}\mapsto \boldsymbol{q}/c$
preserves $K_{d}$, such that the symmetric constraint $P_{ii}\le t,\quad Q_{jj}\le t$ is not a loss of generality.
Conversely, factoring any feasible positive block matrix as
$W\,W^\dagger$ supplies vectors representing its entries whose
squared norms are bounded by $t$.

The minimum in \eqref{eq:finite-haagerup-primal} is attained because
bounded $t$ bounds every entry of the positive block matrix,
so bounded sublevel sets are compact.
The program is also strictly feasible, since one may take
$P=Q=\lambda I$ and $t>\lambda>\|K_{d}\|$.
Strong semidefinite duality therefore applies.

For a vector $\mathbf{r}$, write $D_{\mathbf{r}}=\operatorname{diag}(\mathbf{r})$.
Introducing nonnegative multipliers $\mathbf{r},\mathbf{s}$ for the diagonal
constraints yields the dual
\begin{gather}
\begin{aligned}
    \|K_{d}\|_h=
\max_{\substack{\mathbf{r},\mathbf{s}\in\mathbb R_+^d,\;
B\in\mathbb C^{d\times d}\\
\sum_i r_i+\sum_j s_j=1}}
2\operatorname{Re}\operatorname{Tr}(B^\dagger K_{d})\\
\text{s.t.}\quad\begin{pmatrix}
D_r&-B\\
-B^\dagger&D_s
\end{pmatrix}\succeq0.
\label{eq:finite-haagerup-dual}
\end{aligned}
\end{gather}
The block positivity condition is equivalent to
\begin{align}
B=D_{\sqrt r}\,U\,D_{\sqrt s},
\qquad \|U\|\le1.
\end{align}
This follows from the Schur complement on the supports of
$r$ and $s$; zero diagonal entries force the corresponding
rows or columns of $B$ to vanish.

Maximizing over $U$ and using operator/trace-norm duality gives
\begin{align}
\|K_{d}\|_h
=
\max_{\substack{r,s\ge0\\
\sum_i r_i+\sum_j s_j=1}}
2\|D_{\sqrt r}K_{d} D_{\sqrt s}\|_1,
\end{align}
where $\|\cdot\|_1$ denotes the trace norm.
Let
\begin{align}
\Delta_n
:=
\left\{\alpha\in\mathbb R_+^n:\sum_i\alpha_i=1\right\}.
\end{align}
Writing $r=a\alpha$ and $s=(1-a)\beta$, with
$\alpha,\beta\in\Delta_n$, yields
\begin{align}
2\|D_{\sqrt r}K_{d} D_{\sqrt s}\|_1
=
2\sqrt{a(1-a)}
\|D_{\sqrt\alpha}K_{d} D_{\sqrt\beta}\|_1.
\end{align}
The factor $2\sqrt{a(1-a)}$ is at most one and equals one
when $a=1/2$. Hence
\begin{align}
\|K_{d}\|_h
=
\max_{\alpha,\beta\in\Delta_n}
\|D_{\sqrt\alpha}K_{d} D_{\sqrt\beta}\|_1.
\label{eq:finite-haagerup-trace}
\end{align}

Finally, the Schur multiplier norm admits the same expression:
\begin{align}
\|K_d\|_s
&=
\max_{\|\mathbf u\|_2=\|\mathbf v\|_2=1}
\sup_{\|Z\|\le1}
\left|
\mathbf u^\dagger(K_{d}\circ Z)\mathbf v
\right|
\notag\\
&=
\max_{\|\mathbf u\|_2=\|\mathbf v\|_2=1}
\sup_{\|Z\|\le1}
\left|
\operatorname{Tr}\!\left[
(D_{\mathbf u^*}K_{d} D_{\mathbf v})^\top Z
\right]
\right|
\notag\\
&=
\max_{\|\mathbf u\|_2=\|\mathbf v\|_2=1}
\|D_{\mathbf u^*}K_{d} D_{\mathbf v}\|_1
\notag\\
&=
\max_{\alpha,\beta\in\Delta_n}
\|D_{\sqrt\alpha}K_{d} D_{\sqrt\beta}\|_1.
\end{align}
The third equality uses operator/trace-norm duality and
invariance of the trace norm under transposition.
For the last equality, absorb the phases of $\mathbf u$
and $\mathbf v$ into diagonal unitaries and set
$\alpha_i=|u_i|^2$ and $\beta_j=|v_j|^2$.
Comparison with \eqref{eq:finite-haagerup-trace} proves
$\|K_{d}\|_s=\|K_{d}\|_h$.
\end{proof}

\subsection{$L_\star$-Haagerup norm equivalence}

We now extend the result of Lemma \ref{lem:discrete-schur-haagerup} to any dimension and to continuous sets of points in the domain.  Intuitively, one can see this extension as simply taking the supremum over the set of points $\mathsf{X},\mathsf{Y}$ and over the dimension $d$ on both sides of Eq. \eqref{eq:disc-schur-hag}. For the left-hand side, we obtain precisely the definition of the Schur multiplier norm in Eq. \eqref{eq:Schur}. For the right-hand side, due to the continuity of $K$, if the supremum over $d$ and $\mathsf{X},\mathsf{Y}$ can commute with the infimum over decompositions, one can expect that the expression converges to a global decomposition of the kernel and the supremum of its vector norms over the whole domain. The formal proof, however, is not so straightforward and relies on the compactness of the set of factorizations. The equivalence between the Schur multiplier norm and the Haagerup norm of $K_f$, together with Lemma \ref{lem:lips-schur}, proves Thm.\ref{thm:OL-haagerup}.

\section{Proof of Lemma~\ref{lem:semiinfinite_reformulation}: compression to the coefficient span}\label{app:proof-sip}

Here we prove Lemma~\ref{lem:semiinfinite_reformulation}.  First, notice that every feasible point of the semi-infinite problem in Eq.~\eqref{eq:Lstar-sip} defines a Haagerup factorization of $K_f$ through the factorization $\Gamma=WW^\dagger$ discussed before the lemma, so the optimal value of Eq.~\eqref{eq:Lstar-sip} is at least $L_\star[f]$. The converse is not immediate, the infimum in Eq.~\eqref{eq:Lstar} runs over all factorizations of $K_f$, including the ones with non-polynomial factors and factorizations with polynomial factors of degree larger than that of $K_f$, whose excess terms cancel in $\sum_k p_k(x)q_k(y)^*$. We must show in Lemma~\ref{lem:compression} that none of them achieves a smaller value than the factorizations parametrized by $(P,Q)$. 

Our Lemma~\ref{lem:compression} is an elementary instance of the injectivity of the Haagerup tensor product~\cite{paulsen_smith_1987,pisier2003}, and Ref.~\cite{aleksandrov2024} for the relation between Haagerup tensor products and Schur multipliers. We include a direct proof because the infimum in Eq.~\eqref{eq:Lstar} is defined pointwise, over factorizations as in Eq.~\eqref{eq:haagerup}, rather than as an operator-space tensor norm.

Throughout this appendix, valid factorizations are written as $K(x,y)=\sum_{k\in I}p_k(x)q_k(y)^*$, with an arbitrary index set $I$ that may be finite, countable or uncountable. Here $\mathbf{p}(x)=(p_k(x))_{k\in I}$ and $\mathbf{q}(y)=(q_k(y))_{k\in I}$ belong to $\ell^2(I)$. 
In addition, let $\mathbf{\Phi}=(\Phi_0,\ldots,\Phi_{n-1})^\top$ be the basis of Lemma~\ref{lem:haagerup_poly} on $\mathcal{D}\in\{[-1,1],\mathbb{T}\}$, and let $\mathcal{D}_1,\mathcal{D}_2\subseteq\mathcal{D}$. 
In the main text $\mathcal{D}_1=\mathcal{D}_2=\mathcal{D}$, we distinguish them because we  analyze the case of two finite sample sets $\mathcal{D}_1=\mathsf{X}$, $\mathcal{D}_2=\mathsf{Y}$ later in Corollary~\ref{cor:finite-section}. Assume that $\mathbf{\Phi}$ is linearly independent as a family of functions on $\mathcal{D}_1$, and also on $\mathcal{D}_2$. If $\mathbf{\Phi}$ is linearly independent on a set $\mathcal{S}$, then the vectors $\{\mathbf{\Phi}(x)\}_{x\in\mathcal{S}}$ and $\{\mathbf{\Phi}(x)^*\}_{x\in\mathcal{S}}$ span $\mathbb{C}^{n}$, otherwise some $\mathbf{c}\neq0$ satisfies $\mathbf{c}^\top\mathbf{\Phi}(x)=0$ for all $x\in\mathcal{S}$. In particular there are $x_1,\ldots,x_n\in\mathcal{S}$ whose evaluation matrix, with rows $\mathbf{\Phi}(x_i)^\top$, is invertible. We will recall these facts recurrently in what follows below.

Also, assume the divided difference kernel is of the form
\begin{align}\label{eq:comp-kernel}
K(x,y)=\mathbf{\Phi}(x)^\top G\,\mathbf{\Phi}(y)^*,
\qquad G\in\mathbb C^{n\times n}.
\end{align}
For a kernel $K$ on $\mathcal{D}_1\times\mathcal{D}_2$ we write
\begin{align}\label{eq:comp-haagerup-norm}
\norm{K}_h\coloneq\inf_{I,\,\mathbf{p},\,\mathbf{q}}\;
\sup_{x\in\mathcal{D}_1}\norm{\mathbf{p}(x)}_2\,
\sup_{y\in\mathcal{D}_2}\norm{\mathbf{q}(y)}_2.
\end{align}
The infimum runs over all index sets $I$ and all maps $\mathbf{p}:\mathcal{D}_1\to\ell^2(I)$, $\mathbf{q}:\mathcal{D}_2\to\ell^2(I)$ with $K(x,y)=\sum_{k\in I}p_k(x)q_k(y)^*$ for all $x,y$. For $\mathcal{D}_1=\mathcal{D}_2=\mathcal{D}$ and $K=K_f$, this is the infimum in Eq.~\eqref{eq:Lstar}. Suprema may equal $+\infty$, and we set $0\cdot\infty\coloneq 0$.

In the lemma below, we prove that every factorization $K(x,y)=\langle\mathbf{p}(x),\mathbf{q}(y)\rangle_{\mathcal{H}}$ through an abstract Hilbert space $\mathcal{H}$ is of the form of Eq.~\eqref{eq:haagerup}, with the same pointwise norms. Indeed, we can expand $\mathbf{p}(x)$ and $\mathbf{q}(y)$ in an orthonormal basis $\{e_k\}_{k\in I}$ of $\mathcal{H}$. In particular, this is why Eq.~\eqref{eq:haagerup} allows an arbitrary index set $I$. For kernels of the form of Eq.~\eqref{eq:comp-kernel}, Lemma~\ref{lem:compression} shows that restricting to finite index sets, or to countable ones, does not change the infimum in Eq.~\eqref{eq:comp-haagerup-norm}.

\begin{lemma}[Reduction to the coefficient span]\label{lem:compression}
Let $\mathbf{\Phi}$, $G$ and $K$ be as in Eq.~\eqref{eq:comp-kernel}, with $\mathbf{\Phi}$ linearly independent on $\mathcal{D}_1$ and on $\mathcal{D}_2$. Let $I$ be any index set and $\mathbf{p}:\mathcal{D}_1\to\ell^2(I)$, $\mathbf{q}:\mathcal{D}_2\to\ell^2(I)$ any maps with $K(x,y)=\sum_{k\in I}p_k(x)q_k(y)^*$ for all $(x,y)\in\mathcal{D}_1\times\mathcal{D}_2$. Then there exist  vectors $\mathbf{h}_{p_k},\mathbf{h}_{q_k}\in\mathbb{C}^{n}$, for $k=1,\ldots,m$ and $m\le\mathrm{rank}(G)$, such that 
\begin{align}\label{eq:comp-factors}
\hat{p}_k\coloneq\mathbf{h}_{p_k}^\top\mathbf{\Phi},\qquad
\hat{q}_k\coloneq\mathbf{h}_{q_k}^\top\mathbf{\Phi}
\end{align}
satisfy
\begin{align}
&\sum_{k=1}^{m}\mathbf{h}_{p_k}\mathbf{h}_{q_k}^\dagger=G,
\label{eq:comp-gram}\\
&\sum_{k=1}^{m}\vert\hat{p}_k(x)\vert^2\le\norm{\mathbf{p}(x)}_2^2
\quad\forall \, x\in\mathcal{D}_1,
\label{eq:comp-pnorm}\\
&\sum_{k=1}^{m}\vert\hat{q}_k(y)\vert^2\le\norm{\mathbf{q}(y)}_2^2
\quad\forall \, y\in\mathcal{D}_2.
\label{eq:comp-qnorm}
\end{align}
In particular, $K(x,y)=\sum_{k=1}^{m}\hat{p}_k(x)\hat{q}_k(y)^*$ is a factorization of the form of Eq.~\eqref{eq:valid-fact}. We have $m=0$ if and only if $G=0$.
\end{lemma}
\begin{proof}
\emph{Step 1 (project $\mathbf{p}$ onto the span of $\mathbf{q}$).}
Let $\mathcal{H}_q$ be the smallest closed subspace of $\ell^2(I)$ that contains every vector $\mathbf{q}(y)$, $y\in\mathcal{D}_2$, where the closure is taken with respect to the norm of $\ell^2(I)$. This definition gives two properties that will be used below. First, since $\mathcal{H}_q$ is a closed subspace of the Hilbert space $\ell^2(I)$, there exists an orthogonal projection $ \Pi_q:\ell^2(I)\rightarrow\mathcal{H}_q$.
Second, the span of the vectors $\mathbf{q}(y)$ is dense in $\mathcal{H}_q$. Thus, for every $\mathbf{v}\in\mathcal{H}_q$, there exists a sequence
\begin{align}\label{eq:comp-Hq-density}
\mathbf{v}_j \in \mathrm{span}\{ \mathbf{q}(y):y\in\mathcal{D}_2 \}, \quad \text{s.t.}, \quad \norm{\mathbf{v}_j-\mathbf{v}}_{2}\rightarrow 0.
\end{align}
The approximation property in Eq.~\eqref{eq:comp-Hq-density} will be used in Step~2 to extend an identity proved first for finite linear combinations of the vectors $\mathbf{q}(y)$ to all of $\mathcal{H}_q$.

We now remove from $\mathbf{p}(x)$ the component orthogonal to all of the vectors $\mathbf{q}(y)$. Define
\begin{align}
\tilde{\mathbf{p}}(x)
\coloneq
\Pi_q\mathbf{p}(x),
\qquad x\in\mathcal{D}_1.
\label{eq:comp-p-projection}
\end{align}
Since every $\mathbf{q}(y)$ belongs to $\mathcal{H}_q$, the projection acts trivially on these vectors, $\Pi_q\mathbf{q}(y)=\mathbf{q}(y)$.
Moreover, $\Pi_q$ is self-adjoint. Hence, for every
$x\in\mathcal{D}_1$ and $y\in\mathcal{D}_2$,
\begin{align}\label{eq:comp-step1}
\langle \tilde{\mathbf{p}}(x), \mathbf{q}(y) \rangle
&= \langle \Pi_q\mathbf{p}(x), \mathbf{q}(y) \rangle \nonumber\\
&= \langle \mathbf{p}(x), \Pi_q\mathbf{q}(y) \rangle \nonumber\\
&= \langle \mathbf{p}(x), \mathbf{q}(y) \rangle
= K(x,y),
\end{align}
where we used the shorthand notation $\langle\mathbf{u},\mathbf{v}\rangle\coloneq\sum_{k\in I}u_kv_k^*$.
Thus projecting $\mathbf{p}(x)$ onto $\mathcal{H}_q$ leaves the kernel unchanged.

Finally, orthogonal projections are contractions, so
\begin{align}\label{eq:comp-step1-norm}
\norm{\tilde{\mathbf{p}}(x)}_2 = \norm{\Pi_q\mathbf{p}(x)}_2
\le \norm{\mathbf{p}(x)}_2,
\quad x\in\mathcal{D}_1.
\end{align}
Therefore, after this projection we may restrict attention to
$\tilde{\mathbf{p}}(x)\in\mathcal{H}_q$ without changing $K$ and without increasing the pointwise norm of the first factor.

\emph{Step 2 (the image of $\mathcal{H}_q$ is finite-dimensional).}
We first identify the finite-dimensional space of functions generated by the kernel sections $K(\cdot,y)$. Define $\mathcal{V}_1 \coloneq \mathrm{span} \left\{ K(\cdot,y):y\in\mathcal{D}_2 \right\}$.
Using Eq.~\eqref{eq:comp-kernel}, every kernel section can be written as
\begin{align}
K(\cdot,y) = \mathbf{\Phi}(\cdot)^\top G\, \mathbf{\Phi}(y)^*.
\end{align}
Since the functions forming $\mathbf{\Phi}$ are linearly independent on $\mathcal{D}_2$,  the vectors $ \left\{ \mathbf{\Phi}(y)^*:y\in\mathcal{D}_2 \right\}$
span $\mathbb{C}^{n_2}$. Therefore,
\begin{align}\label{eq:comp-V1-coeff}
\mathcal{V}_1
= \left\{ \mathbf{\Phi}(\cdot)^\top G\,\mathbf{w}: \mathbf{w}\in\mathbb{C}^{n_2}
\right\}.
\end{align}
Moreover, linear independence of $\mathbf{\Phi}$ implies that $\mathbf{\Phi}(\cdot)^\top\mathbf{z}=0 \Longrightarrow
\mathbf{z}=0$.
Hence the map $ \mathbf{z} \longmapsto \mathbf{\Phi}(\cdot)^\top\mathbf{z} $
is injective on $\mathrm{span} \left\{ \text{columns of }G \right\}$. It follows that
\begin{align}\label{eq:comp-V1-dim}
\dim(\mathcal{V}_1) = \mathrm{rank}(G).
\end{align}

We now associate to each vector in $\mathcal{H}_q$ the function obtained by taking its inner product with $\tilde{\mathbf{p}}(x)$. Define
\begin{align}\label{eq:comp-J}
J:\mathcal{H}_q \rightarrow \mathbb{C}^{\mathcal{D}_1}, \quad 
J(\mathbf{v})(x) \coloneq \left\langle \tilde{\mathbf{p}}(x), \mathbf{v} \right\rangle ,
\end{align}
where $\mathbb{C}^{\mathcal{D}_1}$ denotes the space of all complex-valued functions on $\mathcal{D}_1$. Under our convention, the map $J$ is conjugate-linear in its argument, $J(\lambda\mathbf{v}) = \lambda^*\,J(\mathbf{v})$. Thus, for each $y\in\mathcal{D}_2$, Eq.~\eqref{eq:comp-step1} gives
\begin{align}
J\!\left(\mathbf{q}(y)\right)(x)
&= \left\langle \tilde{\mathbf{p}}(x), \mathbf{q}(y) \right\rangle \nonumber\\
&= K(x,y),
\end{align}
and $J\!\left(\mathbf{q}(y)\right) = K(\cdot,y) \in \mathcal{V}_1$.
By conjugate-linearity, if $\mathbf{v} = \sum_{i=1}^{s} \lambda_i\mathbf{q}(y_i)$ is any finite linear combination of the vectors $\mathbf{q}(y)$, then
\begin{align}\label{eq:comp-J-span}
J(\mathbf{v}) = \sum_{i=1}^{s} \lambda_i^* K(\cdot,y_i) \in \mathcal{V}_1.
\end{align}
Therefore,
\begin{align}\label{eq:comp-J-span-inclusion}
J\!\left( \mathrm{span} \left\{ \mathbf{q}(y):y\in\mathcal{D}_2 \right\} \right) \subseteq \mathcal{V}_1.
\end{align}

We next extend this inclusion from finite linear combinations to all of $\mathcal{H}_q$. Let $\mathbf{v}\in\mathcal{H}_q$ be arbitrary. By Eq.~\eqref{eq:comp-Hq-density}, there exists a sequence $ \mathbf{v}_j \in \mathrm{span} \left\{ \mathbf{q}(y):y\in\mathcal{D}_2 \right\} $ such that $ \norm{\mathbf{v}_j-\mathbf{v}}_2 \rightarrow 0$.

For every $x\in\mathcal{D}_1$, the Cauchy--Schwarz inequality gives
\begin{align}\label{eq:comp-J-pointwise}
\left| J(\mathbf{v}_j)(x) - J(\mathbf{v})(x) \right|
&= \left| \left\langle \tilde{\mathbf{p}}(x), \mathbf{v}_j-\mathbf{v} \right\rangle \right| \nonumber\\
&\leq \norm{\tilde{\mathbf{p}}(x)}_2 \norm{\mathbf{v}_j-\mathbf{v}}_2 \rightarrow 0.
\end{align}
Hence $J(\mathbf{v}_j)(x) \rightarrow J(\mathbf{v})(x)$ for every $x\in\mathcal{D}_1$. In other words, $J(\mathbf{v}_j)$ converges pointwise to $J(\mathbf{v})$.

It remains to verify that this pointwise limit still belongs to $\mathcal{V}_1$. Let $\phi_1,\ldots,\phi_r$ be a basis of $\mathcal{V}_1$, where $ r = \dim(\mathcal{V}_1) = \mathrm{rank}(G)$.
Since the functions $\phi_1,\ldots,\phi_r$ are linearly independent,there exist points $z_1,\ldots,z_r \in \mathcal{D}_1$
such that the evaluation matrix 
\[
F \coloneq \left[ \phi_{\ell}(z_i) \right]_{i,\ell=1}^{r}
\]
is invertible.

Since $J(\mathbf{v}_j)\in\mathcal{V}_1$ by Eq.~\eqref{eq:comp-J-span}, there are coefficients $c_{j1},\ldots,c_{jr}$ such that $J(\mathbf{v}_j) = \sum_{\ell=1}^{r} c_{j\ell}\phi_{\ell}$.
Writing
\[
\mathbf{c}_j \coloneq
\begin{bmatrix}
c_{j1}&\cdots&c_{jr}
\end{bmatrix}^{\top},
\]
evaluation at $z_1,\ldots,z_r$ gives
\begin{align}
\mathbf{c}_j = F^{-1}
\begin{bmatrix}
J(\mathbf{v}_j)(z_1)\\
\vdots\\
J(\mathbf{v}_j)(z_r)
\end{bmatrix}
\rightarrow \mathbf{c} \coloneq F^{-1}
\begin{bmatrix}
J(\mathbf{v})(z_1)\\
\vdots\\
J(\mathbf{v})(z_r)
\end{bmatrix}
\end{align}
by the pointwise convergence in Eq.~\eqref{eq:comp-J-pointwise}.
Therefore, for every $x\in\mathcal{D}_1$,
\begin{align}
J(\mathbf{v})(x)
&= \lim_{j\rightarrow\infty}J(\mathbf{v}_j)(x)
= \lim_{j\rightarrow\infty} \sum_{\ell=1}^{r} c_{j\ell}\phi_{\ell}(x) \nonumber\\
&= \sum_{\ell=1}^{r} c_{\ell}\phi_{\ell}(x).
\end{align}
Thus $J(\mathbf{v}) \in \mathcal{V}_1$.
Since $\mathbf{v}\in\mathcal{H}_q$ was arbitrary, we conclude that $J(\mathcal{H}_q) \subseteq \mathcal{V}_1$. Combining this with Eq.~\eqref{eq:comp-V1-dim}, the image of $J$ is finite-dimensional and satisfies
\begin{align}\label{eq:comp-J-rank}
\dim (J(\mathcal{H}_q)) \leq \dim(\mathcal{V}_1) = \mathrm{rank}(G).
\end{align}

\emph{Step 3 (the vectors $\tilde{\mathbf{p}}(x)$ span a finite-dimensional subspace).}
Define $E \coloneq \mathrm{span} \left\{ \tilde{\mathbf{p}}(x):x\in\mathcal{D}_1 \right\}$.
By construction in Step~1, $\tilde{\mathbf{p}}(x)\in\mathcal{H}_q$ for every $x\in\mathcal{D}_1$. Hence
\[
E\subseteq\mathcal{H}_q.
\]
Since $\mathcal{H}_q$ is closed in $\ell^{2}(I)$, the closure $\overline{E}$ of $E$ also satisfies $\overline{E}\subseteq\mathcal{H}_q$. We next characterize the kernel of the map $J$ defined in Eq.~\eqref{eq:comp-J}. For $\mathbf{v}\in\mathcal{H}_q$,
\begin{align}
\mathbf{v}\in\mathrm{ker}(J)
&\Longleftrightarrow J(\mathbf{v})(x)=0 \quad \text{for every }x\in\mathcal{D}_1 \nonumber\\
&\Longleftrightarrow \left\langle \tilde{\mathbf{p}}(x),\mathbf{v} \right\rangle =0 \quad \text{for every }x\in\mathcal{D}_1 \nonumber\\
&\Longleftrightarrow \mathbf{v}\in E^{\perp}.
\end{align}
Therefore,
\begin{align}\label{eq:comp-ker-J}
\mathrm{ker}(J) = \mathcal{H}_q\cap E^{\perp}.
\end{align}
Since taking an orthogonal complement is insensitive to closure, $E^{\perp} = \overline{E}^{\,\perp}$. Thus Eq.~\eqref{eq:comp-ker-J} can equivalently be written as
\begin{align}\label{eq:comp-ker-J-closure}
\mathrm{ker}(J) = \mathcal{H}_q\cap\overline{E}^{\,\perp} = \overline{E} \oplus \mathrm{ker}(J),
\end{align}
where the last equality stands since $\overline{E}$ is a closed subspace of the Hilbert space $\mathcal{H}_q$, every vector in $\mathcal{H}_q$ admits a unique orthogonal decomposition into a component in $\overline{E}$ and a component orthogonal to $\overline{E}$. 

Equation~\eqref{eq:comp-ker-J-closure} implies that the restriction of $J$ to $\overline{E}$ is injective. Indeed, if $\mathbf{v}\in\overline{E}$ and $J(\mathbf{v})=0$, then
\[
\mathbf{v}
\in
\overline{E}\cap\mathrm{ker}(J).
\]
The sum in Eq.~\eqref{eq:comp-ker-J-closure} is direct, so
\[
\overline{E}\cap\mathrm{ker}(J)
=
\left\{\mathbf{0}\right\},
\]
and therefore $\mathbf{v}=\mathbf{0}$.

Although $J$ is conjugate-linear rather than linear, its injectivity still preserves linear independence, so the dimension of $\overline{E}$ cannot exceed the dimension of its image under $J$, \emph{i.e.}, $\dim\left(\overline{E}\right) = \dim\left(J(\overline{E})\right)$.
Moreover,
\[
J(\overline{E})
\subseteq
J(\mathcal{H}_q),
\]
and Step~2 established in Eq.~\eqref{eq:comp-J-rank} that
\[
\dim\left(J(\mathcal{H}_q)\right)
\leq
\mathrm{rank}(G).
\]
Therefore,
\begin{align}\label{eq:comp-E-rank}
\dim\left(\overline{E}\right)
\leq
\dim\left(J(\mathcal{H}_q)\right)
\leq
\mathrm{rank}(G).
\end{align}

Thus $\overline{E}$ is finite-dimensional. Since $E\subseteq\overline{E}$, the space $E$ is also finite-dimensional. Every finite-dimensional subspace of a Hilbert space is closed, so $E=\overline{E}$.
Consequently,
\begin{align}\label{eq:comp-E-dim}
m
\coloneq
\dim(E)
\leq
\mathrm{rank}(G).
\end{align}
Therefore the entire family
$\{\tilde{\mathbf{p}}(x):x\in\mathcal{D}_1\}$ lies in an $m$-dimensional closed subspace of $\mathcal{H}_q$, with $m\leq\mathrm{rank}(G)$.

\emph{Step 4 (project $\mathbf{q}$ onto $E$).}
As in Step~1, let $\Pi_E$ be the orthogonal projection onto $E$ and $\tilde{\mathbf{q}}(y)\coloneq\Pi_E\mathbf{q}(y)$. Since $\tilde{\mathbf{p}}(x)\in E$, Eq.~\eqref{eq:comp-step1} gives
\begin{align}\label{eq:comp-step4}
\langle\tilde{\mathbf{p}}(x),\tilde{\mathbf{q}}(y)\rangle
=\langle\Pi_E\tilde{\mathbf{p}}(x),\mathbf{q}(y)\rangle=K(x,y),
\end{align}
and $\norm{\tilde{\mathbf{q}}(y)}_2\le\norm{\mathbf{q}(y)}_2$.

\emph{Step 5 (coordinates lie in the polynomial span).}
Fix an orthonormal basis $e_1,\ldots,e_m$ of $E$ and define
\begin{align*}
\hat{p}_k(x)\coloneq\langle\tilde{\mathbf{p}}(x),e_k\rangle,\qquad
\hat{q}_k(y)\coloneq\langle\tilde{\mathbf{q}}(y),e_k\rangle.
\end{align*}
First, $\hat{p}_k=J(e_k)\in\mathcal{V}_1\subseteq\mathrm{span}\{\Phi_t\}$, since $e_k\in E\subseteq\mathcal{H}_q$.

Second, $E$ is finite-dimensional and spanned by the vectors $\tilde{\mathbf{p}}(x)$, so there are finitely many points $x_1,\ldots,x_s$ and coefficients $c_{ki}$ with $e_k=\sum_i c_{ki}\tilde{\mathbf{p}}(x_i)$. Using $\Pi_E e_k=e_k$ and Eq.~\eqref{eq:comp-step1},
\begin{align*}
\hat{q}_k(y)
&=\langle\mathbf{q}(y),e_k\rangle
=\sum_i c^*_{ki}\,\langle\mathbf{q}(y),\tilde{\mathbf{p}}(x_i)\rangle\\
&=\Big(\sum_ic_{ki}K(x_i,y)\Big)^{*}
=\big(\mathbf{w}_k^\top\mathbf{\Phi}(y)^*\big)^{*}
=\mathbf{w}_k^{\dagger}\,\mathbf{\Phi}(y),
\end{align*}
with $\mathbf{w}_k^\top\coloneq\sum_ic_{ki}\mathbf{\Phi}(x_i)^\top G$. Hence $\hat{q}_k\in\mathrm{span}\{\Phi_t\}$. Writing $\hat{p}_k=\mathbf{h}_{p_k}^\top\mathbf{\Phi}$ and $\hat{q}_k=\mathbf{h}_{q_k}^\top\mathbf{\Phi}$ defines the vectors in Eq.~\eqref{eq:comp-factors}; they are unique by linear independence.

\emph{Step 6 (Parseval and uniqueness of $G$).}
Since $\tilde{\mathbf{p}}(x),\tilde{\mathbf{q}}(y)\in E$ and $\{e_k\}$ is an orthonormal basis of $E$, Eq.~\eqref{eq:comp-step4} and Parseval's identity give
\begin{align}\label{eq:comp-parseval}
K(x,y)=\sum_{k=1}^m\hat{p}_k(x)\hat{q}_k(y)^*,\qquad
\sum_{k=1}^m\vert\hat{p}_k(x)\vert^2=\norm{\tilde{\mathbf{p}}(x)}_2^2,
\end{align}
and similarly $\sum_k\vert\hat{q}_k(y)\vert^2=\norm{\tilde{\mathbf{q}}(y)}_2^2$. Together with Steps 1 and 4 this proves Eqs.~\eqref{eq:comp-pnorm} and \eqref{eq:comp-qnorm}.

Next, $\hat{p}_k(x)\hat{q}_k(y)^*=\mathbf{\Phi}(x)^\top\mathbf{h}_{p_k}\mathbf{h}_{q_k}^\dagger\mathbf{\Phi}(y)^*$. Put $M\coloneq\sum_k\mathbf{h}_{p_k}\mathbf{h}_{q_k}^\dagger$. By Eqs.~\eqref{eq:comp-kernel} and \eqref{eq:comp-parseval}, $\mathbf{\Phi}(x)^\top(M-G)\mathbf{\Phi}(y)^*=0$ for all $(x,y)\in\mathcal{D}_1\times\mathcal{D}_2$. For fixed $y$, linear independence on $\mathcal{D}_1$ gives $(M-G)\mathbf{\Phi}(y)^*=0$. By the linear independence of $\mathbf{\Phi}$ on $\mathcal{D}_2$, these vectors span $\mathbb{C}^n$, so $M=G$, which is Eq.~\eqref{eq:comp-gram}.

Finally, if $G=0$ then $\mathcal{V}_1=\{0\}$ and $m\le\mathrm{rank}(G)=0$. Conversely, if $m=0$ then Eq.~\eqref{eq:comp-gram} gives $G=0$.

\end{proof}

We now consider the semi-infinite program
\begin{align}\label{eq:comp-sip}
\begin{aligned}
L_{\rm SIP}\coloneq\inf_{L,P,Q}\ & L\\
\text{s.t.}\ &
\Gamma\coloneq\begin{bmatrix}P&G\\G^\dagger&Q\end{bmatrix}\succeq0,\\
&\mathbf{\Phi}(x)^\top P\,\mathbf{\Phi}(x)^*\le L\quad\forall x\in\mathcal{D}_1,\\
&\mathbf{\Phi}(y)^\top Q\,\mathbf{\Phi}(y)^*\le L\quad\forall y\in\mathcal{D}_2,\\
&P,Q\in\mathrm{Herm}(n).
\end{aligned}
\end{align}
For $\mathcal{D}_1=\mathcal{D}_2=\mathcal{D}$ and $n=N$, this is Eq.~\eqref{eq:Lstar-sip}. Feasibility forces $L\ge0$, because $P\succeq0$.

\begin{proposition}[Exactness and attainment of the semi-infinite program]\label{prop:sip-proof}
Let $\mathcal{D}_1,\mathcal{D}_2\subseteq\mathcal{D}$ be compact, and let $\mathbf{\Phi}$ be continuous and linearly independent on $\mathcal{D}_1$ and on $\mathcal{D}_2$. Let $K$ be given by Eq.~\eqref{eq:comp-kernel}. Then:
\begin{enumerate}[label=(\roman*)]
\item $L_{\rm{SIP}}=\norm{K}_h$;
\item the infimum in Eq.~\eqref{eq:comp-sip} is attained;
\item the infimum in Eq.~\eqref{eq:comp-haagerup-norm} is attained by a factorization with $m\le\mathrm{rank}(G)$ terms of the form of Eq.~\eqref{eq:comp-factors}.
\end{enumerate}
\end{proposition}

\begin{proof}
\emph{Equality (i).} First, let us prove that $\norm{K}_h\le L_{\rm{SIP}}$.
Let $(L,P,Q)$ be feasible and let $r=\mathrm{rank}(\Gamma)$. Factor $\Gamma=WW^\dagger$ with $W\in\mathbb{C}^{2n\times r}$, and write $W=[A^\top\,B^\top]^\top$ with columns $\mathbf a_k$ of $A$ and $\mathbf b_k$ of $B$. Then $P=AA^\dagger$, $Q=BB^\dagger$ and $G=AB^\dagger$. Define $p_k\coloneq\mathbf a_k^\top\mathbf{\Phi}$ and $q_k\coloneq\mathbf b_k^\top\mathbf{\Phi}$. Then
\begin{align*}
\sum_{k=1}^rp_k(x)q_k(y)^*&=\mathbf{\Phi}(x)^\top AB^\dagger\,\mathbf{\Phi}(y)^*=K(x,y),\\
\sum_{k=1}^r\vert p_k(x)\vert^2&=\mathbf{\Phi}(x)^\top P\,\mathbf{\Phi}(x)^*\le L,
\end{align*}
and likewise $\sum_k\vert q_k(y)\vert^2\le L$. This is a factorization with $I=\{1,\ldots,r\}$ and $\sup\norm{\mathbf{p}}_2\cdot\sup\norm{\mathbf{q}}_2\le L$. (If $r=0$ then $G=0$ and $K\equiv0$.) Taking the infimum over feasible points gives the claim.

Now, let us prove the converse, $L_{\rm{SIP}}\le\norm{K}_h$.
If $K=0$, then $G=0$ by the linear independence of $\mathbf{\Phi}$, and $(L,P,Q)=(0,0,0)$ is feasible, so $L_{\rm{SIP}}=0$. Assume now $K\neq 0$. Take any factorization $(I,\mathbf{p},\mathbf{q})$ and set $a\coloneq\sup_x\norm{\mathbf{p}(x)}_2$ and $b\coloneq\sup_y\norm{\mathbf{q}(y)}_2$. Both are positive; if either is infinite there is nothing to prove.

Lemma~\ref{lem:compression} yields $\mathbf{h}_{p_k},\mathbf{h}_{q_k}$ with $\hat{a}\coloneq\sup_x(\sum_k\vert\hat{p}_k(x)\vert^2)^{1/2}\le a$ and $\hat{b}\le \coloneq\sup_x(\sum_k\vert\hat{q}_k(x)\vert^2)^{1/2} \le b$. Both are positive, since $\hat{a}=0$ would give $K=0$ by Eq.~\eqref{eq:comp-parseval}. With $t\coloneq(\hat{b}/\hat{a})^{1/2}$ set
\begin{align*}
P\coloneq t^2\sum_k\mathbf{h}_{p_k}\mathbf{h}_{p_k}^\dagger,\qquad
Q\coloneq t^{-2}\sum_k\mathbf{h}_{q_k}\mathbf{h}_{q_k}^\dagger.
\end{align*}
Then
\begin{align*}
\Gamma=\sum_k\begin{bmatrix}t\,\mathbf{h}_{p_k}\\ t^{-1}\mathbf{h}_{q_k}\end{bmatrix}
\begin{bmatrix}t\,\mathbf{h}_{p_k}\\ t^{-1}\mathbf{h}_{q_k}\end{bmatrix}^\dagger\succeq0,
\end{align*}
and its off-diagonal block is $G$ by Eq.~\eqref{eq:comp-gram}. Moreover $\mathbf{\Phi}(x)^\top P\,\mathbf{\Phi}(x)^*=t^2\sum_k\vert\hat{p}_k(x)\vert^2\le t^2\hat{a}^2=\hat{a}\hat{b}$, and similarly $\mathbf{\Phi}(y)^\top Q\,\mathbf{\Phi}(y)^*\le\hat{a}\hat{b}$. Hence $(\hat{a}\hat{b},P,Q)$ is feasible, and
\begin{align*}
L_{\rm{SIP}}\le\hat{a}\hat{b}\le ab.
\end{align*}
Taking the infimum over factorizations proves (i).

\emph{Attainment (ii).}
The feasible set is nonempty. Take $P=\norm{G}\mathds 1$, $Q=\norm{G}\mathds 1$ and
\begin{align*}
L=\norm{G}\sup_{x\in\mathcal{D}_1\cup\mathcal{D}_2}\norm{\mathbf{\Phi}(x)}_2^2,
\end{align*}
which is finite by continuity on compact sets. This choice satisfies $\Gamma\succeq 0$, with $c\coloneq\norm{G}$, for all $\mathbf{u},\mathbf{v}\in\mathbb{C}^{n}$,
\begin{align*}
\begin{bmatrix}\mathbf{u}\\\mathbf{v}\end{bmatrix}^{\dagger}\!\Gamma\begin{bmatrix}\mathbf{u}\\\mathbf{v}\end{bmatrix}
&=c\norm{\mathbf{u}}^2+c\norm{\mathbf{v}}^2+2\mathrm{Re}\left(\mathbf{u}^\dagger G\mathbf{v}\right)\\
&\ge c\big(\norm{\mathbf{u}}-\norm{\mathbf{v}}\big)^2\ge 0,
\end{align*}
since $\vert\mathbf{u}^\dagger G\mathbf{v}\vert\le c\norm{\mathbf{u}}\norm{\mathbf{v}}$. This includes the case $G=0$.

Next we show that the sets $\mathcal{D}_1,\mathcal{D}_2$ are bounded. By the linear independence of $\mathbf{\Phi}$, choose $x_1,\ldots,x_{n}\in\mathcal{D}_1$ whose evaluation matrix $E_1$ is invertible. For feasible $(L,P,Q)$ with $L\le\lambda$, we have $P\succeq 0$ and
\begin{align*}
\mathrm{tr}(E_1PE_1^\dagger)=\sum_{i}\mathbf{\Phi}(x_i)^\top P\,\mathbf{\Phi}(x_i)^*\le n\lambda,
\end{align*}
so $\norm{P}\le\norm{E_1^{-1}}^2\mathrm{tr}(E_1PE_1^\dagger)\le n\lambda\norm{E_1^{-1}}^2$. The same argument on $\mathcal{D}_2$ bounds $Q$. The sets $\{(L,P,Q)\ \text{feasible}:L\le\lambda\}$ are closed, being intersections of the closed PSD cone with closed half-spaces, one for each $x$ and $y$. They are therefore compact. Since the objective is continuous, the infimum is attained.

\emph{Attainment (iii).}
Let $(L_{\rm{SIP}},P,Q)$ be optimal. The construction in the first part produces a factorization with $\sup\norm{\mathbf{p}}_2\cdot\sup\norm{\mathbf{q}}_2\le L_{\rm{SIP}}=\norm{K}_h$, so this factorization is optimal. Lemma~\ref{lem:compression} does not increase either supremum, so compressing it yields an optimal factorization with at most $\mathrm{rank}(G)$ terms of the form of Eq.~\eqref{eq:comp-factors}.

\end{proof}

Now, we are finally in place to prove Lemma~\ref{lem:semiinfinite_reformulation}.
\begin{proof}[Proof of Lemma~\ref{lem:semiinfinite_reformulation}]
Let $\mathcal{D}\in\{[-1,1],\mathbb{T}\}$ and let $f$ be a polynomial as in Lemma~\ref{lem:haagerup_poly}. Its basis $\mathbf{\Phi}$ is continuous and linearly independent on the compact set $\mathcal{D}$. Its divided-difference kernel has the form of Eq.~\eqref{eq:comp-kernel}, directly for $\mathcal{D}=[-1,1]$ and, for $\mathcal{D}=\mathbb{T}$, in one of the two forms described in the conventions above.

Proposition~\ref{prop:sip-proof} with $\mathcal{D}_1=\mathcal{D}_2=\mathcal{D}$ shows that the program of Eq.~\eqref{eq:Lstar-sip} attains its minimum and that the minimum equals the infimum in Eq.~\eqref{eq:Lstar}.

A polynomial is $C^1$ on $\mathcal{D}$, and its kernel is factorizable by Lemma~\ref{lem:haagerup_poly}. Theorem~\ref{thm:OL-haagerup} identifies this infimum with $L_\star[f]$. Hence the optimal value of Eq.~\eqref{eq:Lstar-sip} is $L_\star[f]$, and it is attained. By Proposition~\ref{prop:sip-proof}(iii), $L_\star[f]$ is moreover attained by a Haagerup factorization with at most $\mathrm{rank}(G)$ terms whose factors lie in the coefficient span. In the convention of the main text $G\in\mathbb{C}^{N\times N}$, so at most $N$ terms are needed, and for a Chebyshev polynomial of exact degree $N$ one has $\mathrm{rank}(G)=N$ (see the discussion after Eq.~\eqref{eq:Hcheb}).
\end{proof}

\begin{corollary}[Sampled lower bound]\label{cor:finite-section}
In the setting of Lemma~\ref{lem:sdp}, let $K$ be the kernel of Eq.~\eqref{eq:comp-kernel} on $\mathcal{D}\times\mathcal{D}$, and $\mathsf{X},\mathsf{Y}\subset\mathcal{D}$ be finite sets.
\begin{enumerate}[label=(\alph*)]
\item For every $\mathbf{u}\in\mathbb R_+^{r_X}$ and $\mathbf{v}\in\mathbb R_+^{r_Y}$ with $\sum_iu_i+\sum_jv_j=2$,
\begin{align}\label{eq:comp-dual-bound}
\Big\Vert\mathrm{diag}\big(\sqrt{\mathbf{u}}\big)\,E_{\mathsf{X}}GE_{\mathsf{Y}}^\dagger\,\mathrm{diag}\big(\sqrt{\mathbf{v}}\big)\Big\Vert_*\le\norm{K}_h.
\end{align}
Consequently $L_{\mathrm{low}}\le L_\star[f]$ for arbitrary finite $\mathsf{X},\mathsf{Y}$.
\item If, in addition, $E_{\mathsf{X}}$ and $E_{\mathsf{Y}}$ have full column rank, then the optimal value of Eq.~\eqref{eq:sdp} is $L_{\mathsf{X},\mathsf{Y}}=\norm{K|_{\mathsf{X}\times\mathsf{Y}}}_h$.
\end{enumerate}
\end{corollary}

\begin{proof}
(a) Part~(a) bounds every feasible value of the nuclear-norm dual, not only its maximum. Hence the lower bound $L_{\mathrm{low}}\le L_\star[f]$ remains valid for inexact dual solves, which is the situation in Algorithm~\ref{alg:dual-nuclear-inner}. Let $K(x,y)=\langle\mathbf{p}(x),\mathbf{q}(y)\rangle$ be any factorization through $\ell^2(I)$, with finite suprema $a\coloneq\sup_x\norm{\mathbf{p}(x)}_2$ and $b\coloneq\sup_y\norm{\mathbf{q}(y)}_2$. The case of infinite suprema is trivial. Define operators $A:\ell^2(I)\to\mathbb{C}^{r_X}$ and $B:\ell^2(I)\to\mathbb{C}^{r_Y}$ by
\begin{align*}
(A\mathbf{h})_i\coloneq\sqrt{u_i}\,\langle\mathbf{h},\mathbf{p}(x_i)\rangle,\qquad
(B\mathbf{h})_j\coloneq\sqrt{v_j}\,\langle\mathbf{h},\mathbf{q}(y_j)\rangle.
\end{align*}
Then $A^\dagger\mathbf e_i=\sqrt{u_i}\,\mathbf{p}(x_i)$ and $B^\dagger\mathbf e_j=\sqrt{v_j}\,\mathbf{q}(y_j)$. Since $K(x_i,y_j)=[E_{\mathsf{X}}GE_{\mathsf{Y}}^\dagger]_{ij}$,
\begin{align*}
(BA^\dagger)_{ji}=\sqrt{u_iv_j}\,\langle\mathbf{p}(x_i),\mathbf{q}(y_j)\rangle
=\big[D_u^{1/2}E_{\mathsf{X}}GE_{\mathsf{Y}}^\dagger D_v^{1/2}\big]_{ij},
\end{align*}
with $D_u\coloneq\mathrm{diag}(\mathbf{u})$ and $D_v\coloneq\mathrm{diag}(\mathbf{v})$. Thus the matrix in Eq.~\eqref{eq:comp-dual-bound} is the transpose of $BA^\dagger$, and it has the same nuclear norm.

Hölder's inequality for Schatten norms gives $\norm{BA^\dagger}_*\le\norm{B}_{\rm{F}}\norm{A}_{\rm{F}}$, where $\norm{\cdot}_{\rm{F}}$ is the Frobenius norm. Moreover $\norm{A}_{\rm{F}}^2=\sum_i\norm{A^\dagger\mathbf e_i}_2^2=\sum_iu_i\norm{\mathbf{p}(x_i)}_2^2\le Ua^2$ and $\norm{B}_{\rm{F}}^2\le Vb^2$, with $U\coloneq\sum_iu_i$ and $V\coloneq\sum_jv_j$. Since $\sqrt{UV}\le(U+V)/2=1$,
\begin{align*}
\Big\Vert D_u^{1/2}E_{\mathsf{X}}GE_{\mathsf{Y}}^\dagger D_v^{1/2}\Big\Vert_*\le\sqrt{UV}\,ab\le ab.
\end{align*}
Taking the infimum over factorizations gives Eq.~\eqref{eq:comp-dual-bound}. Maximizing over $(\mathbf{u},\mathbf{v})$, Eq.~\eqref{eq:dual-nuclear} gives $L_{\mathrm{low}}\le\norm{K}_h=L_\star[f]$.

(b) Part~(b) is not needed for the lower bound. It identifies the finite program with the Haagerup norm of the sampled kernel. Its hypothesis holds throughout Algorithm~\ref{alg:active-set}, \emph{i.e.}, the initial sets contain $2N+1\ge N$ distinct points, and a nonzero element of the $N$-dimensional span of $\mathbf{\Phi}$ has at most $N-1$ zeros in $\mathcal{D}$ (for $\mathbb{T}$ after multiplication by $z^{N/2}$). Since the active sets are nested, the hypothesis persists at every iteration. Full column rank of $E_{\mathsf{X}}$ (resp.\ $E_{\mathsf{Y}}$) is exactly linear independence of $\mathbf{\Phi}$ on $\mathsf{X}$ (resp.\ $\mathsf{Y}$). Finite sets are compact, and every function on them is continuous. Proposition~\ref{prop:sip-proof} with $\mathcal{D}_1=\mathsf{X}$ and $\mathcal{D}_2=\mathsf{Y}$ therefore applies; the constraint $L\ge0$ in Eq.~\eqref{eq:sdp} is implied by $P\succeq 0$.
\end{proof}

 
\section{Exact sum-of-squares reformulation of the semi-infinite problem}\label{app:sip-to-sos}

The continuous optimization problem in Eq.~\eqref{eq:Lstar-sip} is a convex semi-infinite program with finite-dimensional variables $(P,Q,L)$, and inequalities
\[
\mathbf{\Phi}( x)^\top P\,\mathbf{\Phi}( x)^*\leq L,
\quad
\mathbf{\Phi}( y)^\dagger Q\,\mathbf{\Phi}( y)\leq L
\]
that must hold for all points in the domain $\mathcal{D}$. Let
\[
g_P(x)\coloneq L-\mathbf{\Phi}(x)^\dagger P\,\mathbf{\Phi}(x),
\,\,
g_Q(x)\coloneq L-\mathbf{\Phi}(x)^\dagger Q\,\mathbf{\Phi}(x).
\]
Then, the problem can be written equivalently as
\begin{gather}\label{eq:Lstar-sip-gen}
    \begin{aligned}
        \min_{L,P,Q} &\quad L\\
        \text{s.t.} & \quad
        \begin{cases}
            \begin{bmatrix}
                P & G\\
                G^\dagger & Q
            \end{bmatrix}\succeq 0,\\
            g_P(x)\ge 0,\quad \forall \, x \in \mathcal{D},\\
            g_Q(x)\ge 0,\quad \forall \, x \in \mathcal{D},
        \end{cases}
    \end{aligned}
\end{gather}
for Hermitian matrices $P,Q$ of appropriate sizes.
Thus, the only infinite-dimensional component of the problem is the requirement that the scalar functions $g_P$ and $g_Q$ be nonnegative on $\mathcal{D}$. For Chebyshev polynomials of fixed degree, $g_P$ and $g_Q$ belong to a finite-dimensional space of univariate polynomials. Their global nonnegativity can therefore be characterized exactly using sum-of-squares (SOS) representations. In both cases, global nonnegativity is equivalent to membership in a convex cone that admits an exact finite-dimensional representation \cite{roh2006}. This allows us to replace the semi-infinite constraints in \eqref{eq:Lstar-sip-gen} by conic constraints.

\subsection{Chebyshev case}
For algebraic polynomials on $\mathcal{D}=[-1,1]$, the relevant result is the Markov-Lukács theorem \cite[Sec.~1.21]{szeg1939}.  In this case, for Chebyshev polynomials,
\[
\mathbf{\Phi}(x)=(T_0(x),\dots,T_{N-1}(x))^\top,
\]
and $g_P,g_Q$ are polynomials with degree at most $2N-2$. 

Using $T_\ell(\cos\theta)=\cos(\ell\theta)$ for $\theta\in[0,\pi]$, both $g_P$ and $g_Q$ can be written as cosine polynomials
\[
g(\theta)=\sum_{\ell=0}^{M} c_\ell \cos(\ell\theta),
\qquad M=2N-2.
\]
Therefore,
\[
g(x)\ge 0 \, \forall \,\, x\in[-1,1]
\quad \Longleftrightarrow \quad
g(\theta)\ge 0 \, \forall \,\, \theta\in[0,\pi].
\]
Let
\[
\theta_k=\frac{k\pi}{M},\qquad k=0,\dots,M,
\]
be the Chebyshev--Lobatto nodes, and define the evaluation matrix
\[
U\in\mathbb{R}^{(M+1)\times N},
\quad
U_{k,\ell}=T_{\ell-1}(\cos\theta_k).
\]
Then the sampled values of $g_P$ and $g_Q$ in these nodes are
\begin{align*}
y^{(P)} &\coloneq \big(g_P(\cos\theta_0),\dots,g_P(\cos\theta_M)\big)^\top
\\ &= L\mathds{1}-\mathrm{diag}(UPU^\top),
\end{align*}
and, similarly, $y^{(Q)} = L\mathds{1}-\mathrm{diag}(UQU^\top)$. Since $g_P,g_Q$ are cosine polynomials of degree at most $M$, their coefficients are uniquely determined by these samples.

Now we are in place to state the exact nonnegativity via discrete cosine transforms. Define the discrete cosine transform matrix
\[
[V_{\mathrm{DCT}}]_{k,\ell}=\cos(\ell\theta_k),
\qquad k,\ell=0,\dots,M,
\]
and
\[
W_{\mathrm{DCT}} \coloneq \frac{2}{M} D V_{\mathrm{DCT}} D,
\quad
D=\mathrm{diag}\!\left(\tfrac{1}{2},1,\dots,1,\tfrac{1}{2}\right),
\]
so that $W_{\mathrm{DCT}}V_{\mathrm{DCT}}=\mathds{1}$. Let $V_1$ be the first $N$ columns of $V_{\mathrm{DCT}}$, and let $V_2$ be the matrix
\[
[V_2]_{k,\ell}=\sin(\ell\theta_k),
\quad
k=0,\dots,M,\ \ell=1,\dots,N-1.
\]

By \cite[Thm.~4]{roh2006}, a cosine polynomial of degree $M=2N-2$ is nonnegative on $[0,\pi]$ if and only if there exist
\[
X_1\in\mathrm{Herm}(N),\,\,
X_2\in\mathrm{Herm}(N-1),
\quad
X_1,X_2\succeq 0,
\]
such that its sampled values satisfy
\[
y = \mathrm{diag}(V_1 X_1 V_1^\top) + \mathrm{diag}(V_2 X_2 V_2^\top).
\]

Finally, let $G \in \mathbb{R}^{N \times N}$ be fixed, then the semi-infinite problem in \eqref{eq:Lstar-sip} is equivalent to the following finite-dimensional SDP:
\begin{gather}\label{eq:cheb-rv-sdp}
    \begin{aligned}
        \min_{\substack{L,P,Q,\\ X_{P,1},X_{P,2},\\X_{Q,1},X_{Q,2}}} & \quad L\\
        \text{s.t.} & 
        \begin{cases}
            \begin{bmatrix}
                P & G\\
                G^\dagger & Q
            \end{bmatrix}\succeq 0,\\
            L\mathds{1}-\mathrm{diag}(UPU^\top)
            = \mathrm{diag}(V_1X_{P,1}V_1^\top) +\\
            \qquad\qquad +\mathrm{diag}(V_2X_{P,2}V_2^\top),\\
            L\mathds{1}-\mathrm{diag}(UQU^\top)
            = \mathrm{diag}(V_1X_{Q,1}V_1^\top) + \\
            \qquad\qquad +\mathrm{diag}(V_2X_{Q,2}V_2^\top),\\
            P,Q,X_{P,1},X_{Q,1}\in\mathrm{Herm}(N),\\
            X_{P,2},X_{Q,2}\in\mathrm{Herm}(N-1),\\
            X_{P,1},X_{P,2},X_{Q,1},X_{Q,2}\succeq 0.
        \end{cases}
    \end{aligned}
\end{gather}


\section{Proof of Lemma~\ref{lem:sdp}}\label{sec:proof-lem-discr}
Here we prove Lemma~\ref{lem:sdp}, which we restate for the convenience of the reader.

\begin{lemma}[Bounding $L_\star$ as a finite optimization]
In the setup of Lemma \ref{lem:semiinfinite_reformulation}, let $\mathsf{X}=\{ x_i\}_{i=1}^{r_X},\mathsf{Y}=\{ y_j\}_{j=1}^{r_Y}$ be finite sets of points 
in the domain $\mathcal{D}$.
Then
\begin{enumerate}[leftmargin=*,label=(\roman*)]
\item substituting the infinitely many constraints in Eq. \eqref{eq:Lstar-sip} by their restriction to the finite number of points $ x\in\mathsf{X}$ and $ y\in\mathsf{Y}$ defines a SDP;
\item  the SDP in $(i)$ has a dual form without any semi-definite constraint and reads
\begin{align}\label{eq:dual-nuclear}
    L_{\mathrm{low}}=\max_{\substack{\mathbf{u}\in\mathbb{R}_+^{r_X},\,\mathbf{v}\in\mathbb{R}_+^{r_Y}\\
 \sum_i \! u_i+\sum_j \! v_j=2}}  \norm{\mathrm{diag}\big(\sqrt{\mathbf{u}}\big) \!\cdot\! E_\mathsf{X}^{\phantom{\dagger}} \!\cdot\! G \!\cdot\! E_\mathsf{Y}^\dagger \!\cdot\! \mathrm{diag}\big(\sqrt{\mathbf{v}}\big)}_*,
\end{align}
where $E_\mathsf{X}$ is the $r_X \times N$ matrix with rows given by $\mathbf{\Phi}( x_1)^\top,\ldots,\mathbf{\Phi}( x_{r_X})^\top$ (and similarly for $E_\mathsf{Y}$), and $\norm{\cdot}_*$ is the nuclear norm $\norm{A}_*\coloneq\trace\left(\sqrt{A^\dagger A} \right)$;
\item the SDP in $(i)$ has strong duality;
\item the solution of the dual problem \eqref{eq:dual-nuclear} gives the lower bound $L_{\mathrm{low}} \leq L_\star$;
\item if $E_\mathsf{X}$ and $E_\mathsf{Y}$ have full column rank, the finite primal optimum is attained, and, given dual variables $\mathbf{u},\mathbf{v}$ from $(ii)$, Haagerup polynomials $p_k,q_k$ in Eq.~\eqref{eq:haagerup} can be recovered by an iteration in which each sweep costs $\mathcal{O}(N^3+r N^2)$
arithmetic operations and $\mathcal{O}\!\left(rN+N^2\right)$ memory, with $r\coloneq r_X+r_Y$. The number of sweeps is not bounded a priori, and acceptance is certified a posteriori. A certified, possibly looser, upper bound in $(vi)$ is also available without any sweep;
\item with $p_k,\, q_k$ from $(v)$, a certified upper bound $L_{\mathrm{up}}\geq L_\star$ is obtained from the values of $\norm{\mathbf{p}(x)}_2^2$ and $\norm{\mathbf{q}(y)}_2^2$ on a grid of $M=\mathcal{O}(N)$ points, at a cost of $\mathcal{O}(N^2)$ arithmetic operations per point;
\item each first-order iteration for the problem in $(ii)$ costs $\mathcal{O}\!\left( I_{\rm{obj}}\big[(r_X+r_Y)N^2+N^3\big] \right)$
arithmetic operations and $\mathcal{O}\!\left((r_X+r_Y)N+N^2\right)$ memory, where $I_{\rm{obj}}$ is the number of objective evaluations in a line search. The number of iterations is bounded a priori by $\mathcal{O}(m\log m/\varepsilon_{\rm rel}^{3})$, with $m=r_X+r_Y$, for a relative accuracy $\varepsilon_{\rm rel}$ (App.~\ref{sec:dual-nuclear-complexity}), and is typically far smaller in practice.
\end{enumerate}
\end{lemma}
\subsection{Proof of $(i)$}

Let $\mathsf{X}=(x_i)_{i=1}^{r_X}\in\mathcal{D}^{r_X}$, and $\mathsf{Y}=(y_j)_{j=1}^{r_Y}\in\mathcal{D}^{r_Y}$ be the current active sets of sampling points at iteration $t$ with cardinalities $\vert \mathsf{X} \vert = r_X$, and $\vert \mathsf{Y} \vert = r_Y$, respectively. Let us define the evaluation matrices
\[
E_\mathsf{X} \coloneq 
\begin{bmatrix}
\mathbf{\Phi}( x_1)^\top \\
\vdots \\
\mathbf{\Phi}( x_{r})^\top
\end{bmatrix}
\in \mathbb{C}^{r_{\mathsf{X}} \times N},
\]
and similarly for $E_\mathsf{Y}$.
Therefore, from $K_f(x,y)=\mathbf{\Phi}(x)^\top G\,\mathbf{\Phi}(y)^*$ (Lemma~\ref{lem:haagerup_poly}), the kernel evaluated on this set gives
\begin{align}\label{eq-kernel-coeff-finite}
K_{\mathsf{X},\mathsf{Y}} = E_\mathsf{X} \cdot G \cdot E_\mathsf{Y}^\dagger \quad \in\mathbb{C}^{r_X\times r_Y}.     
\end{align}
Now, the Gram block for the kernel evaluated on this set is the image of the coefficient-space Gram block in Eq.~\eqref{eq:G-sip} under the congruence
\begin{align}\label{eq:G-discrete}
\begin{bmatrix}
    E_\mathsf{X} & 0 \\
    0 & E_\mathsf{Y}
\end{bmatrix}
\begin{bmatrix}
    P & G \\
    G^\dagger & Q
\end{bmatrix}
\begin{bmatrix}
    E_\mathsf{X} & 0 \\
    0 & E_\mathsf{Y}
\end{bmatrix}^\dagger.
\end{align}
Since congruence preserves positive semidefiniteness, we only need to ensure that
\[
\begin{bmatrix}
    P & G \\
    G^\dagger & Q
\end{bmatrix} \succeq 0.
\]
Moreover, restricting the infinitely many inequalities in Eq.~\eqref{eq:Lstar-sip} to the active lists gives pointwise bounds and the following SDP:
\begin{gather}\label{eq:sdp}
    \begin{aligned}
        L_{\mathsf{X},\mathsf{Y}} \coloneq \min_{L,P,Q} & \quad \quad L \\
        \text{s.t}. & \begin{cases}
            \begin{bmatrix}
                P & G \\
                G^\dagger & Q
            \end{bmatrix} \succeq 0 \\
            \mathbf{\Phi}( x_i)^\top P\,\mathbf{\Phi}( x_i)^* \leq L, \quad \,\, i \in [r_X] \\
            \mathbf{\Phi}( y_j)^\top Q\,\mathbf{\Phi}( y_j)^* \leq L, \quad \,j \in [r_Y] \\
            L \geq 0,\\
            P, Q \in \mathrm{Herm}(N).
        \end{cases} 
    \end{aligned}
\end{gather}

\subsection{Proof of $(ii)$}\label{sec:dual-formulation}
\begin{theorem}\label{thm:dual}
    Let $N, G, E_{\mathsf{X}}, E_{\mathsf{Y}}$ be as in \eqref{eq:sdp}. The dual problem of the SDP in Eq. \eqref{eq:sdp} reads 
    \begin{gather}\label{eq:dual-form}
    \begin{aligned}
        \max_{\mathbf{u},\mathbf{v},X} & \quad   \mathrm{Re}\left[ \trace\left( X^\dagger G \right) \right] \\
        \text{s.t}. & \quad \begin{cases}
            \frac{1}{2}\begin{bmatrix}
                E_{\mathsf{X}}^\dagger \!\cdot\! \mathrm{diag}(u) \!\cdot\!E_{\mathsf{X}} & -X \\
                -X^\dagger & E_{\mathsf{Y}}^\dagger \!\cdot\! \mathrm{diag}(v) \!\cdot\! E_{\mathsf{Y}}
            \end{bmatrix} \succeq 0 \\
            \sum_i (u_i + v_i )=2 \\
            \mathbf{u}\in\mathbb{R}_+^{r_X},\,\mathbf{v}\in\mathbb{R}_+^{r_Y} \\
            X \in \mathbb{C}^{N \times N}
        \end{cases} .
    \end{aligned}
\end{gather}
\end{theorem}
\begin{proof}
Let us write the Lagrangian of the primal SDP 
\begin{gather*}
    \begin{aligned}
        L_{\mathsf{X},\mathsf{Y}} \coloneq \min_{L,P,Q} & \quad \quad L \\
        \text{s.t}. & \begin{cases}
            \Gamma = \begin{bmatrix}
                P & G \\
                G^\dagger & Q
            \end{bmatrix} \succeq 0 \\
            \mathbf{\Phi}( x_i)^\top P\,\mathbf{\Phi}( x_i)^* \leq L, \quad \,\, i \in [r_X] \\
            \mathbf{\Phi}( y_j)^\top Q\,\mathbf{\Phi}( y_j)^* \leq L, \quad \,j \in [r_Y] \\
            L \geq 0,\\
            P, Q \in \mathrm{Herm}(N).
        \end{cases} 
    \end{aligned}
\end{gather*}
First, let us write the pointwise bounds as
\begin{gather*}
\begin{aligned}
\mathbf{\Phi}( x_i)^\top P\,\mathbf{\Phi}( x_i)^*=\left[ E_\mathsf{X} \!\cdot\! P \!\cdot\! E_\mathsf{X}^\dagger\right]_{ii} \leq L,
\\
\mathbf{\Phi}( y_j)^\top Q\,\mathbf{\Phi}( y_j)^*=\left[ E_\mathsf{Y} \!\cdot\! Q \!\cdot\! E_\mathsf{Y}^\dagger\right]_{jj} \leq L.
\end{aligned}
\end{gather*}
Then, the Lagrangian reads
\begin{gather}\label{eq:lagrangian}
   \begin{aligned}
\mathcal{L}(P,Q,&L,Y,\mathbf{u},\mathbf{v}) = L - \trace(Y\, G) + \\
&+ \sum_{i=1}^{r_X} \frac{u_i}{2}(\left[ E_\mathsf{X} \!\cdot\! P \!\cdot\! E_\mathsf{X}^\dagger\right]_{ii}-L)+\\ 
&+\sum_{j=1}^{r_Y} \frac{v_j}{2}(\left[ E_\mathsf{Y} \!\cdot\! Q \!\cdot\! E_\mathsf{Y}^\dagger\right]_{jj}-L),
\end{aligned} 
\end{gather}
with the dual multipliers $Y \succeq 0$, $\mathbf{u}\in\mathbb{R}_+^{r_X},\,\mathbf{v}\in\mathbb{R}_+^{r_Y}$, associated with the constraints in Eq.~\eqref{eq:sdp} and signs according to the orientation of the inequalities. Since $\Gamma$ is a block matrix, $Y$ can also be seen as a block matrix
\[
Y =-\frac{1}{2}\begin{bmatrix}
        X_{11} & X_{12} \\
        X_{12}^\dagger & X_{22}
    \end{bmatrix}, 
\]
with $X_{12} \in \mathbb{C}^{N \times N}$, and $X_{11}, X_{22} \in \mathrm{Herm}(N)$. Thus, we can write $\trace(Y\, G)$ as
\begin{gather}
    \begin{aligned}
       \trace(Y\Gamma)=-\frac{1}{2}[ &\trace (X_{11}P)+\trace (X_{12} G^\dagger)\\
    &+\trace (X_{12}^\dagger G)+\trace (X_{22}Q)]. 
    \end{aligned}
\end{gather}

Collecting the primal variables in $\mathcal{L}$, we have
\begin{align*}
    P: & \frac{1}{2}\trace (X_{11}P) + \sum_{i=1}^{r_X} \frac{u_i}{2} \left[ E_\mathsf{X} \!\cdot\! P \!\cdot\! E_\mathsf{X}^\dagger\right]_{ii}  =  \\
    & \qquad\qquad \trace \left[\frac{1}{2} \left( X_{11}+ E_\mathsf{X}^\dagger \!\cdot\! \mathrm{diag}(\mathbf{u}) \!\cdot\! E_\mathsf{X} \right) P \right]  \\
    Q: & \frac{1}{2}\trace (X_{22}Q) + \sum_{j=1}^{r_Y} \frac{v_j}{2} \left[ E_\mathsf{Y} \!\cdot\! Q \!\cdot\! E_\mathsf{Y}^\dagger\right]_{jj}  =  \\
    & \qquad\qquad \trace \left[\frac{1}{2} \left( X_{22}+ E_\mathsf{Y}^\dagger \!\cdot\! \mathrm{diag}(\mathbf{v}) \!\cdot\! E_\mathsf{Y} \right) Q \right]  \\
    s: & \, s - \sum_{i=1}^{r_X} \frac{u_i}{2} s - \sum_{j=1}^{r_Y} \frac{v_j}{2} s  = s\left[1-\frac{1}{2}\sum_{i=1}^{r_X} u_i - \frac{1}{2}\sum_{j=1}^{r_Y}v_j \right],
\end{align*}
and constant terms $\frac{1}{2}\trace (X_{12}G^\dagger)+\frac{1}{2}\trace (X_{12}^\dagger G)$. A dual function $g(Y,\mathbf{u},\mathbf{v})$ is the infimum of $\mathcal{L}$ over $P,Q,L$. For the infimum to be achievable, we need to satisfy stationary conditions \cite{ben2013}, \emph{i.e.},  differentiate the above equations with respect to the primal variables and set them to zero,
\begin{gather*}
\begin{aligned}
-X_{11}= E_\mathsf{X}^\dagger \!\cdot\! \mathrm{diag}(\mathbf{u}) \!\cdot\! E_\mathsf{X} \\
-X_{22}= E_\mathsf{Y}^\dagger \!\cdot\! \mathrm{diag}(\mathbf{v}) \!\cdot\! E_\mathsf{Y}\\
  \sum_i (u_i + v_i)=2,
\end{aligned}
\end{gather*}
for feasible $Y\succeq 0$, $\mathbf{u}\in\mathbb{R}_+^{r_X},\,\mathbf{v}\in\mathbb{R}_+^{r_Y}$. Thus, under stationary conditions, the Lagrangian reduces to
\[
g(Y,\mathbf{u},\mathbf{v}) = \frac{1}{2}\trace (X_{12}G^\dagger)+\frac{1}{2}\trace (X_{12}^\dagger G) = \, \mathrm{Re}\left[ \trace\left( X^\dagger G \right) \right],
\]
and the dual problem reads
\begin{gather*}
    \begin{aligned}
        \max_{\mathbf{u},\mathbf{v},X} & \quad   \mathrm{Re}\left[ \trace\left( X^\dagger G \right) \right] \\
        \text{s.t}. & \quad \begin{cases}
            \frac{1}{2}\begin{bmatrix}
                E_{\mathsf{X}}^\dagger \!\cdot\! \mathrm{diag}(u) \!\cdot\!E_{\mathsf{X}} & -X \\
                -X^\dagger & E_{\mathsf{Y}}^\dagger \!\cdot\! \mathrm{diag}(v) \!\cdot\! E_{\mathsf{Y}}
            \end{bmatrix} \succeq 0 \\
            \sum_i  u_i+\sum_j v_j=2 \\
            \mathbf{u}\in\mathbb{R}_+^{r_X},\,\mathbf{v}\in\mathbb{R}_+^{r_Y} \\
            X \in \mathbb{C}^{N \times N}
        \end{cases} .
    \end{aligned}
\end{gather*}
\end{proof}

As the primal formulation in Eq.~\eqref{eq:sdp}, it is possible to obtain a solution to this semidefinite problem using off-the-shelf solvers, but the memory cost of solving it using standard interior-point methods might still be prohibitive. 
Next, we show that the dual problem can be rewritten without the semidefinite constraints. 

\begin{theorem}[Nuclear-norm form of the dual]\label{thm:dual-nuclear}
Let $G$, $E_{\mathsf{X}}$, and $E_{\mathsf{Y}}$ be as in
Theorem~\ref{thm:dual}. For $u\in\mathbb{R}_+^{r_X}$ and $v\in\mathbb{R}_+^{r_Y}$, define
\begin{align}\label{eq:weighted-coefficient-gram}
D_u &\coloneq \mathrm{diag}(\mathbf{u}),
&
D_v &\coloneq \mathrm{diag}(\mathbf{v}),
\\
A_u &\coloneq E_{\mathsf{X}}^\dagger D_u E_{\mathsf{X}},
&
B_v &\coloneq E_{\mathsf{Y}}^\dagger D_v E_{\mathsf{Y}}.
\end{align}
Then the dual problem in Eq.~\eqref{eq:dual-form} is equivalent to
\begin{align}\label{eq:dual-nuclear-sample-space}
L_{\mathrm{low}} =
\max_{\substack{
\mathbf{u}\in\mathbb{R}_+^{r_X},\, \mathbf{v}\in\mathbb{R}_+^{r_Y}\\
\sum_i \! u_i+\sum_j \! v_j=2}}
\norm{A_u^{1/2} G B_v^{1/2}}_* ,
\end{align}
where $\norm{\cdot}_*$ denotes the nuclear norm and
$A_u^{1/2}$ and $B_v^{1/2}$ are the positive-semidefinite square roots
of $A_u$ and $B_v$, respectively.
\end{theorem}

\begin{proof}
For a feasible $\mathbf{u}$ and $\mathbf{v}$, the semidefinite constraint in
Eq.~\eqref{eq:dual-form} becomes
\begin{align}\label{eq:fixed-weight-block}
\begin{bmatrix}
A_u & -X\\
-X^\dagger & B_v
\end{bmatrix}\succeq0.
\end{align}
Suppose first that $A_u,B_v \succ 0$. Applying a congruence
transformation with
$\mathrm{diag}(A_u^{-1/2},B_v^{-1/2})$, we have
\begin{align}\label{eq:normalized-block-constraint}
\begin{bmatrix}
A_u & -X\\
-X^\dagger & B_v
\end{bmatrix}\succeq0
\quad\Longleftrightarrow\quad
\begin{bmatrix}
\mathds{1} & -R\\
-R^\dagger & \mathds{1}
\end{bmatrix}\succeq 0,
\end{align}
where $R\coloneq A_u^{-1/2}XB_v^{-1/2}$.
Taking the Schur complement of either identity block in
Eq.~\eqref{eq:normalized-block-constraint}, we obtain
\begin{align*}
\mathds{1}-R^\dagger R\succeq 0 \quad \Longleftrightarrow \quad \norm{R}\leq 1.
\end{align*}
Hence the feasible matrices $X$ are those of the form
\begin{align}\label{eq:proof-X-factorization}
X=A_u^{1/2}RB_v^{1/2}, \quad \norm{R} \leq 1.
\end{align}

The same remains valid when $A_u$ or $B_v$ is
singular. Indeed, positivity of the block matrix in
Eq.~\eqref{eq:fixed-weight-block} implies that $X$ has support only
between the supports of $A_u$ and $B_v$. Restricting the block matrix to
these supports makes $A_u$ and $B_v$ invertible, so the argument above
applies there. Extending the contraction by zero on the
corresponding kernels also yields an $X$ as in Eq.~\eqref{eq:proof-X-factorization}.

Substituting Eq.~\eqref{eq:proof-X-factorization} into the dual objective,
\begin{align*}
\mathrm{Re}\left[ \trace(X^\dagger G) \right]
&=
\mathrm{Re}\left[ \trace\left( R^\dagger A_u^{1/2} G B_v^{1/2} \right) \right].
\end{align*}
Therefore, for fixed $\mathbf{u},\mathbf{v}$,
\begin{align}\label{eq:proof-coefficient-nuclear}
\max_X \mathrm{Re}\left[\trace(X^\dagger G) \right]
&= \max_{\norm{R} \leq 1} \mathrm{Re}\left[ \trace\left( R^\dagger A_u^{1/2}G B_v^{1/2} \right) \right] \notag\\
&=\norm{A_u^{1/2}G B_v^{1/2}}_* ,
\end{align}
where the second equality follows from the duality between the operator
and nuclear norms. Eliminating $X$ and maximizing over the nonnegative weights satisfying $\sum_i u_i+\sum_j v_j=2$ yields Eq.~\eqref{eq:dual-nuclear}.

\end{proof}

\subsection{Proof of $(iii)$}
We now establish strong duality between 
Eq.~\eqref{eq:sdp} and its dual formulation. By Slater's theorem, it
suffices to exhibit a strictly feasible primal point \cite{ben2013}.
\begin{theorem}[Strong duality]\label{thm:strong-duality}
Let $\varepsilon,\delta>0$, and take
\begin{align}
\gamma &= \norm{G}+\varepsilon, \\
P &= Q = \gamma \mathds{1}, \\
L &= \gamma N+\delta .
\end{align}
Then $(P,Q,L)$ is a strictly feasible point of the primal SDP in
Eq.~\eqref{eq:sdp}. Consequently, Slater's condition holds and the
primal and dual optimal values coincide,
\begin{align}
L_{\mathrm{low}} = L_{\mathsf{X},\mathsf{Y}}.
\end{align}
Moreover, the dual optimum is attained.
\end{theorem}

\begin{proof}
We first verify the scalar constraints for the evaluation of the kernel on the sets $\mathsf{X}$ and $\mathsf{Y}$. For every $x_i \in \mathsf{X}$,
\begin{align*}
\mathbf{\Phi}( x_i)^\top P\,\mathbf{\Phi}( x_i)^*
&= \gamma \norm{\mathbf{\Phi}(x_i)}_2^2 \\
&= \gamma \sum_{k=0}^{N-1}\vert T_k(x_i) \vert^2 \\
&\leq \gamma N < L,
\end{align*}
where we used $\vert T_k(x_i) \vert \leq 1$ for $x_i\in[-1,1]$ (on $\mathbb{T}$, $\vert x_i^{k-N/2}\vert=1$ gives the same bound), and likewise for every $y_j\in\mathsf{Y}$.
Thus all scalar inequality constraints are satisfied strictly.

It remains to prove that the block constraint is strictly positive
definite. With $P=Q=\gamma\mathds{1}$,
\begin{align*}
\Gamma= \begin{bmatrix}
\gamma\mathds{1} & G\\
G^\dagger & \gamma\mathds{1}
\end{bmatrix}.
\end{align*}
For arbitrary $\mathbf{a},\mathbf{b} \in \mathbb{C}^N$, we have
\begin{align}\label{eq:slater-quadratic-form}
\begin{bmatrix}
\mathbf{a}\\ \mathbf{b}
\end{bmatrix}^{\!\dagger}
\Gamma
\begin{bmatrix}
\mathbf{a} \\ \mathbf{b}
\end{bmatrix}
&= \gamma\norm{\mathbf{a}}_2^2
+ \gamma\norm{\mathbf{b}}_2^2 + 2\,\mathrm{Re}(\mathbf{a}^\dagger G\mathbf{b}).
\end{align}
Using
\begin{align*}
\mathrm{Re}(\mathbf{a}^\dagger G \mathbf{b}) \geq -|\mathbf{a}^\dagger G \mathbf{b}| \geq
-\norm{G}\,\norm{a}_2\norm{b}_2
\end{align*}
and
\begin{align*}
2\norm{a}_2\norm{b}_2 \leq \norm{a}_2^2+\norm{b}_2^2,
\end{align*}
Eq.~\eqref{eq:slater-quadratic-form} gives
\begin{align}
\begin{bmatrix}
\mathbf{a}\\ \mathbf{b}
\end{bmatrix}^{\!\dagger}
\Gamma
\begin{bmatrix}
\mathbf{a} \\ \mathbf{b}
\end{bmatrix}
&\geq \gamma \left( \norm{\mathbf{a}}_2^2+\norm{\mathbf{b}}_2^2 \right)
- 2\norm{G}\,\norm{\mathbf{a}}_2\norm{\mathbf{b}}_2 \\
&\geq \left(\gamma-\norm{G}\right) \left( \norm{\mathbf{a}}_2^2+\norm{\mathbf{b}}_2^2 \right) \\
&= \varepsilon \left( \norm{\mathbf{a}}_2^2+\norm{\mathbf{b}}_2^2 \right).
\end{align}
Since $\varepsilon>0$, the right-hand side is strictly positive for
every $(\mathbf{a},\mathbf{b})\neq(\mathbf{0},\mathbf{0})$, and therefore, $\Gamma\succ0$.

Hence $(P,Q,L)$ satisfies all primal inequality constraints strictly,
so the primal SDP satisfies Slater's condition. Strong duality follows:
the dual optimum is attained and equals the primal optimum,
\begin{align}
L_{\mathrm{low}} = L_{\mathsf{X},\mathsf{Y}}.
\end{align}
\end{proof}

\subsection{Proof of $(iv)$}\label{sec:dual-lb}
$(iv)$\label{sec:dual-lb}
Since Eq.~\eqref{eq:sdp} is obtained by restricting the constraints of Eq.~\eqref{eq:Lstar-sip}, every feasible point of the semi-infinite problem is feasible for the finite problem. Hence, since the optimal value of Eq.~\eqref{eq:Lstar-sip} equals $L_\star$ by Lemma~\ref{lem:semiinfinite_reformulation} (proved in App.~\ref{app:proof-sip}),
\[
    L_{\mathsf{X},\mathsf{Y}}\le L_\star .
\]
Independently of Lemma~\ref{lem:semiinfinite_reformulation} and of strong duality, Corollary~\ref{cor:finite-section}(a) shows directly that every feasible value of the dual problem in Eq.~\eqref{eq:dual-nuclear}, and therefore its optimum $L_{\mathrm{low}}$, is bounded above by $L_\star$ for arbitrary finite $\mathsf{X},\mathsf{Y}$. Strong duality from $(iii)$ is consistent with these two bounds and gives
\begin{align}\label{eq:sampled-lower-bound}
    L_{\mathrm{low}} = L_{\mathsf{X},\mathsf{Y}} \leq L_\star .
\end{align}

\subsection{Proof of $(v)$}\label{sec:recov-haagerup}

Let $L_{\mathrm{low}}$ be the optimal value of the finite dual in Eq.~\eqref{eq:dual-nuclear}.  By strong duality and, since $E_{\mathsf{X}}$ and $E_{\mathsf{Y}}$ have full column rank, attainment of the finite primal (Corollary~\ref{cor:finite-section}(b) and Proposition~\ref{prop:sip-proof}(ii)),
\begin{align}
L_{\mathrm{low}} = L_{\mathsf{X},\mathsf{Y}},
\end{align}
and there exist Hermitian matrices $P_\star,Q_\star$ such that
\begin{align}\label{eq:recovery-optimal-gram}
\Gamma_\star
\coloneq
\begin{bmatrix}
P_\star & G\\
G^\dagger & Q_\star
\end{bmatrix}
\succeq0
\end{align}
and
\begin{align}\label{eq:recovery-optimal-x-y}
\begin{aligned}
\mathbf{\Phi}(x_i)^\top
P_\star\mathbf{\Phi}(x_i)^*
&\leq L_{\mathrm{low}},
&&i\in[r_X],\\
\mathbf{\Phi}(y_j)^\top
Q_\star\mathbf{\Phi}(y_j)^*
&\leq L_{\mathrm{low}},
&&j\in[r_Y].
\end{aligned}
\end{align}
Consequently, once the finite dual value has been determined, recovery of the factor functions reduces to the fixed-level feasibility problem
\begin{align}\label{eq:fixed-level-gram-feasibility}
\mathcal{F}_{\mathsf{X},\mathsf{Y}}(s) \coloneq
\left\{
\begin{bmatrix}
P&G\\
G^\dagger & Q
\end{bmatrix}\succeq0:
\begin{array}{l}
\mathbf{\Phi}(x_i)^\top P\mathbf{\Phi}(x_i)^*\leq s,\\
\mathbf{\Phi}(y_j)^\top Q\mathbf{\Phi}(y_j)^*\leq s
\end{array}
\right\},
\end{align}
Equations~\eqref{eq:recovery-optimal-gram} and
\eqref{eq:recovery-optimal-x-y} imply
\begin{align}
\mathcal{F}_{\mathsf{X},\mathsf{Y}}
\!\left(L_{\mathrm{low}}\right)
\neq\emptyset.
\end{align}
Conversely, every element of
$\mathcal{F}_{\mathsf{X},\mathsf{Y}}
(L_{\mathrm{low}})$
is feasible for the finite primal at its optimal objective and is therefore primal optimal. Thus the recovery stage does not require a second objective optimization.

Importantly, this construction imposes no rank condition on the weighted evaluation matrices associated with the dual solution.  In particular, the dual weights may have sparse support: all constraints in Eq.~\eqref{eq:fixed-level-gram-feasibility} are imposed independently of whether the corresponding dual weight vanishes.

\paragraph{Normalized numerical recovery.}
Let $d_L$ be the objective value of a feasible point of the finite dual, so that $d_L\leq L_{\mathsf{X},\mathsf{Y}}$. The recovery inputs are the current active sets $\mathsf{X},\mathsf{Y}$, the coefficient matrix $G$, the corresponding dual weights $(\mathbf{u},\mathbf{v})$, a recovery tolerance $\varepsilon_{\rm{rec}}>0$, a maximum number of projection sweeps $I_{\max}$, a conditioning threshold $\kappa_{\max}$, and, when available, the previously certified Gram matrix $\widehat{\Gamma}^{(t-1)}$.

The case $G=0$ is trivial, with $P=Q=0$, so assume $G \neq 0$ and define
\begin{align}\label{eq:recovery-normalization}
\gamma \coloneq \max\!\left\{\norm{G}_2,|d_L|\right\}, \quad
\widetilde{G} \coloneq G/\gamma, \quad
\widetilde{d} \coloneq d_L/\gamma.
\end{align}
In floating point we round $\gamma$ up to the nearest power of two; barring underflow, the scalings by $1/\gamma$ here and by $\gamma$ in Eq.~\eqref{eq:recovery-rescaling} are then exact, so the off-diagonal block of the final Gram matrix equals the stored $G$ exactly.
We further set
\begin{align}\label{eq:recovery-target-level}
\Delta_{\rm{rec}} \coloneq \varepsilon_{\rm{rec}} \max\!\left\{1,|d_L|\right\}, \quad
\widetilde{s}_{\rm{targ}} \coloneq \widetilde{d}+\frac{\Delta_{\rm{rec}}}{2\gamma}.
\end{align}
The second half of $\Delta_{\rm{rec}}$ leaves room for incomplete projection convergence, positive-semidefinite repair, and floating-point error.

\paragraph{Recovery candidates.}
We first construct inexpensive candidates for the normalized Gram matrix. When a certified Gram matrix from the preceding active-set iteration is available, define
\begin{align}\label{eq:recovery-warm-candidate}
\Gamma_{\rm{warm}} \coloneq \frac{\widehat{\Gamma}^{(t-1)}}{\gamma}.
\end{align}
Its diagonal blocks are Hermitized and its off-diagonal block is reset exactly to $\widetilde{G}$. The preceding certificate is not inherited: $\Gamma_{\rm{warm}}$ is tested against every constraint in the current active sets.

A second candidate is available when the current weighted coefficient Gramians are nonsingular.  Define
\begin{align}
A_{\mathbf{u}}\coloneq E_{\mathsf{X}}^\dagger
\mathrm{diag}(\mathbf{u}) E_{\mathsf{X}}, \quad
B_{\mathbf{v}} \coloneq E_{\mathsf{Y}}^\dagger \mathrm{diag}(\mathbf{v}) E_{\mathsf{Y}}.
\end{align}
Since $\mathbf{u},\mathbf{v}\geq 0$, both matrices are positive semidefinite. Take spectral decompositions restricted to their positive eigenspaces,
\begin{align}
A_{\mathbf{u}} = U_A\Lambda_AU_A^\dagger,\quad
B_{\mathbf{v}} = U_B\Lambda_BU_B^\dagger,
\end{align}
where $\Lambda_A$ and $\Lambda_B$ contain only the strictly positive eigenvalues, and define the possibly rectangular factors
\begin{align}
L_A \coloneq U_A\Lambda_A^{1/2},\quad
L_B \coloneq U_B\Lambda_B^{1/2}.
\end{align}
Then $A_{\mathbf{u}}=L_A L_A^\dagger$, and $B_{\mathbf{v}}=L_B L_B^\dagger$.
For feasible optimal dual weights, complementary slackness gives
\begin{align}\label{eq:recovery-kkt-linear-systems}
A_{\mathbf{u}}\widetilde{P}_{\rm{fr}} =
X\widetilde{G}^\dagger,\quad
B_{\mathbf{v}}\widetilde{Q}_{\rm{fr}} = X^\dagger\widetilde{G}.
\end{align}
These equations determine the  $\widetilde{P}_\star,\widetilde{Q}_\star$ uniquely only when $A_{\mathbf{u}}$ and $B_{\mathbf{v}}$ are nonsingular.  We therefore use them as a direct recovery only when
\begin{align}\label{eq:recovery-conditioning-gate}
\frac{\lambda_{\min}(A_{\mathbf{u}})}
{\lambda_{\max}(A_{\mathbf{u}})}
&\geq\tau_{\rm rank},
&
\frac{\lambda_{\min}(B_{\mathbf{v}})}
{\lambda_{\max}(B_{\mathbf{v}})}
&\geq\tau_{\rm rank},
\end{align}
for a prescribed numerical threshold $\tau_{\rm rank}>0$. When the threshold holds, the system of equations in~\eqref{eq:recovery-kkt-linear-systems} are solved without forming matrix inverses, and we construct
\begin{align}\label{eq:recovery-full-rank-candidate}
\Gamma_{\rm{fr}}
\coloneq
\begin{bmatrix}
\mathrm{Herm}(\widetilde{P}_{\rm{fr}})&
\widetilde{G}\\
\widetilde{G}^\dagger&
\mathrm{Herm}(\widetilde{Q}_{\rm{fr}})
\end{bmatrix}.
\end{align}
For numerically inexact dual weights, $\Gamma_{\rm{fr}}$ is regarded only as an initialization candidate.

If either weighted Gramian is rank deficient or fails Eq.~\eqref{eq:recovery-conditioning-gate}, the direct recovery is skipped.  In the rank-deficient case Eq.~\eqref{eq:recovery-kkt-linear-systems} fixes only the components seen on the supports of $A_{\mathbf{u}}$ and $B_{\mathbf{v}}$, while the generic fixed-level Gram recovery described below supplies the missing null-space completion.  In all cases, $\Gamma_{\rm{fr}}$, when constructed, is accepted only after passing the same a posteriori positive-semidefinite and sampled-constraint certification as the other recovery candidates.

Finally, a canonical candidate is obtained from a singular-value decomposition
\begin{align}
\widetilde{G}=U_G\Sigma_GV_G^\dagger.
\end{align}
Define
\begin{align}\label{eq:recovery-spectral-candidate}
\Gamma_{\rm{spec}}
\coloneq
\begin{bmatrix}
U_G\Sigma_GU_G^\dagger&\widetilde{G}\\
\widetilde{G}^\dagger&V_G\Sigma_GV_G^\dagger
\end{bmatrix}.
\end{align}
Indeed,
\begin{align}
\Gamma_{\rm{spec}}
=
\begin{bmatrix}
U_G\Sigma_G^{1/2}\\
V_G\Sigma_G^{1/2}
\end{bmatrix}
\begin{bmatrix}
U_G\Sigma_G^{1/2}\\
V_G\Sigma_G^{1/2}
\end{bmatrix}^{\!\dagger}
\succeq 0.
\end{align}
The spectral candidate depends neither on the dual weights nor on the active sets. After the certification step below, which restores the cross block exactly and certifies positive semidefiniteness, it defines a valid Haagerup factorization of $K_f$ on all of $\mathcal{D}$. Lemma~\ref{lem:certified-max} therefore turns it into a certified upper bound $L_{\mathrm{up}}\geq L_\star$ without any projection sweep. The projection iterations below serve only to tighten this bound toward the finite optimum, not to make it valid.

\paragraph{Certification.}
Every candidate needs to be certified.  Given diagonal blocks $P,Q$, first restore the exact normalized cross block and Hermitize them,
\begin{align}
Y
\coloneq
\begin{bmatrix}
\mathrm{Herm}(P)&\widetilde{G}\\
\widetilde{G}^\dagger&\mathrm{Herm}(Q)
\end{bmatrix}.
\end{align}
Choose a trial shift $\eta\geq\max\{0,-\widehat{\lambda}_{\min}(Y)\}$, where $\widehat{\lambda}_{\min}(Y)$ is the smallest eigenvalue returned by any symmetric eigensolver. Let $S$ be the stored matrix $Y+\eta\mathds{1}_{2N}$, so the positive semidefiniteness of $S$ is certified a posteriori by a Cholesky test~\cite{rump2006verification}. Then set
\begin{align}\label{eq:recovery-repaired-gram-normalized}
\widehat{\Gamma}_{\rm{n}} \coloneq S = Y+\eta \mathds{1}_{2N} =
\begin{bmatrix}
\widehat{P}_{\rm{n}}&\widetilde{G}\\
\widetilde{G}^\dagger&\widehat{Q}_{\rm{n}}
\end{bmatrix}
\succeq0.
\end{align}
Its normalized finite-primal objective is
\begin{gather}\label{eq:recovery-certified-upper-normalized}
  \begin{aligned}
\widehat{s}_{\rm{n}}& \coloneq \\
& \max\!\left\{
\max_i \mathbf{\Phi}(x_i)^\top \widehat{P}_{\rm{n}}\mathbf{\Phi}(x_i)^*, \max_j \mathbf{\Phi}(y_j)^\top \widehat{Q}_{\rm{n}}\mathbf{\Phi}(y_j)^*
\right\}.
\end{aligned}
\end{gather}
Consequently, $\gamma\widehat{\Gamma}_{\rm{n}}$ is a certified feasible point of
the current finite primal at objective $\gamma\widehat{s}_{\rm{n}}$. Positive semidefiniteness holds exactly for the stored matrix; the sampled objective $\widehat{s}_{\rm{n}}$ is evaluated in floating point and is used only for the finite-primal bracket in Eq.~\eqref{eq:recovery-finite-bracket}, not for the bound $L_\star\leq L_{\mathrm{up}}$, which is certified separately in $(vi)$.

The candidates are evaluated sequentially, beginning with
$\Gamma_{\rm{warm}}$ when available, followed by $\Gamma_{\rm{fr}}$ when
Eq.~\eqref{eq:recovery-conditioning-gate} holds.  Any candidate satisfying
\begin{align}\label{eq:recovery-initial-gate}
\gamma \left( \widehat{s}_{\rm{n}}-\widetilde{d} \right) \leq \Delta_{\rm{rec}}
\end{align}
may be accepted immediately.  Otherwise $\Gamma_{\rm{spec}}$ is also certified, and the candidate with the smallest $\widehat{s}_{\rm{n}}$ is retained as the initial point for the projection iteration.  They can be generated and discarded sequentially, so this selection requires only $\mathcal{O}(N^2)$ matrix storage in addition to the stored sample evaluations.

\paragraph{Projection fallback.}
If Eq.~\eqref{eq:recovery-initial-gate} is not satisfied, cyclic projections are applied to the cross-block affine set, the sampled halfspaces, and the positive-semidefinite cone.  Write the current Hermitian iterate as
\begin{align}
\Gamma=
\begin{bmatrix}
P&\widetilde{G}\\
\widetilde{G}^\dagger & Q
\end{bmatrix}.
\end{align}
At the beginning of each iteration, Hemitize,
\begin{align}\label{eq:recovery-cross-projection}
\Gamma \leftarrow
\begin{bmatrix}
\mathrm{Herm}(P)&\widetilde{G}\\
\widetilde{G}^\dagger&\mathrm{Herm}(Q)
\end{bmatrix}.
\end{align}
For $x_i\in\mathsf{X}$, let $ \mathbf\phi_i \coloneq \mathbf{\Phi}(x_i)^*$. Whenever
\begin{align}
\mathbf\phi_i^\dagger
P\mathbf\phi_i
>
\widetilde{s}_{\rm{targ}},
\end{align}
the orthogonal projection onto the halfspace is
\begin{align}\label{eq:recovery-halfspace-projection-x}
P
\leftarrow
P-
\frac{
\mathbf\phi_i^\dagger P\mathbf\phi_i
-
\widetilde{s}_{\rm{targ}}
}{
\norm{\mathbf\phi_i}_2^4
}
\mathbf\phi_i\mathbf\phi_i^\dagger.
\end{align}
The analogous update is applied to $Q$ for every $y_j\in\mathsf{Y}$.
When the iteration is warm started, newly added active points are processed
first, followed by the remaining points in their stored order. After the halfspace projections, take
\begin{align}
\Gamma=
U_\Gamma\mathrm{diag}(\lambda_k)U_\Gamma^\dagger
\end{align}
and project onto the positive-semidefinite cone,
\begin{align}\label{eq:recovery-psd-projection}
\Gamma \leftarrow U_\Gamma \mathrm{diag}\!\left((\lambda_k)_+\right) U_\Gamma^\dagger,
\qquad
(\lambda)_+\coloneq\max\{\lambda,0\}.
\end{align}
Because this projection need not preserve the cross block exactly, the raw iterate is not itself used as a certificate.  Instead, after every iteration we reapply the a-posteriori certification in Eqs.~\eqref{eq:recovery-repaired-gram-normalized} and \eqref{eq:recovery-certified-upper-normalized}, retain the smallest certified $\widehat{s}_{\rm{n}}$ encountered, and stop when Eq.~\eqref{eq:recovery-initial-gate} holds.

For an exact dual optimum, $d_L=L_{\mathrm{low}}$, the projection problem at $\widetilde{s}_{\rm{targ}}$ has a nonempty intersection.  Indeed, $\Gamma_\star/\gamma$ is feasible at level $L_{\mathrm{low}}/\gamma$, and for every
\begin{align}
0<\delta < \widetilde{s}_{\rm{targ}}-\frac{L_{\mathrm{low}}}{\gamma},
\end{align}
adding a small multiple of the identity to $\Gamma_\star/\gamma$ preserves the cross block and produces a strictly positive-definite Gram matrix satisfying all sampled inequalities strictly below $\widetilde{s}_{\rm{targ}}$.  Hence the cyclic projection is a genuine
fixed-level feasibility procedure.

For a feasible dual value $d_L$, no assumption that $\widetilde{s}_{\rm{targ}}$ already lies above the finite optimum is required.  Correctness is instead a posteriori: recovery is accepted only when Eq.~\eqref{eq:recovery-initial-gate} is satisfied by a certified primal matrix.  If this does not occur within $I_{\max}$ sweeps, the active-set iteration is not advanced; the finite dual is solved to tighter feasibility/optimality tolerances and/or the recovery iteration budget is increased.

\paragraph{Recovered coefficient factors.}
For an accepted normalized matrix, define
\begin{align}\label{eq:recovery-rescaling}
\widehat{\Gamma} \coloneq \gamma\widehat{\Gamma}_{\rm{n}}, \quad
\widehat{s}_{\rm{rec}} \coloneq \gamma\widehat{s}_{\rm{n}}.
\end{align}
Then
\begin{align}
\widehat{\Gamma} =
\begin{bmatrix}
\widehat{P}&G\\
G^\dagger&\widehat{Q}
\end{bmatrix}
\succeq 0,
\qquad
\widehat{\Gamma} \in \mathcal{F}_{\mathsf{X},\mathsf{Y}} \!\left(\widehat{s}_{\rm{rec}}\right),
\end{align}
and therefore
\begin{align}\label{eq:recovery-finite-bracket}
d_L \leq L_{\mathsf{X},\mathsf{Y}} \leq \widehat{s}_{\rm{rec}}.
\end{align}
Take an eigendecomposition
\begin{align}
\widehat{\Gamma} = U_{\widehat{\Gamma}} \Lambda_{\widehat{\Gamma}} U_{\widehat{\Gamma}}^\dagger
\end{align}
and define
\begin{align}
W \coloneq U_{\widehat{\Gamma}} \Lambda_{\widehat{\Gamma}}^{1/2} =
\begin{bmatrix}
C\\D
\end{bmatrix}.
\end{align}
Then
\begin{align}\label{eq:recovery-factor-blocks}
CC^\dagger = \widehat{P}, \quad
DD^\dagger = \widehat{Q}, \quad
CD^\dagger = G.
\end{align}
Writing the columns of $C,D$ as $\mathbf{h}_{p_k},\mathbf{h}_{q_k}$ and defining
\begin{align}
p_k(x) \coloneq \mathbf{h}_{p_k}^{\,\top}\mathbf{\Phi}(x), \quad
q_k(y) \coloneq \mathbf{h}_{q_k}^{\,\top}\mathbf{\Phi}(y),
\end{align}
we obtain, for every $x,y\in\mathcal{D}$,
\begin{align}\label{eq:recovery-full-domain-factorization}
K_f(x,y)
&=
\mathbf{\Phi}(x)^\top G \mathbf{\Phi}(y)^*
\notag\\
&= \mathbf{\Phi}(x)^\top CD^\dagger \mathbf{\Phi}(y)^*
\notag\\
&= \sum_k p_k(x)q_k(y)^*.
\end{align}
Moreover,
\begin{align}
\norm{\boldsymbol{p}(x_i)}_2^2 = \mathbf{\Phi}(x_i)^\top \widehat{P} \mathbf{\Phi}(x_i)^* \leq \widehat{s}_{\rm{rec}}, \\
\norm{\boldsymbol{q}(y_j)}_2^2 = \mathbf{\Phi}(y_j)^\top \widehat{Q} \mathbf{\Phi}(y_j)^* \leq \widehat{s}_{\rm{rec}}
\end{align}
for every current active point.  Thus the recovered functions give both the full-domain factorization of the kernel and a certified finite-primal factorization on the active sets.  Lemma~\ref{lem:certified-max} in $(vi)$ then supplies the separate global certificate $L_\star \leq L_{\mathrm{up}}$; it depends only on $\widehat{P}$ and $\widehat{Q}$, not on the numerically computed factors $C,D$.

\paragraph{Complexity.}
Let $r=r_X+r_Y$.  A constant number of matrices of order $N$ or $2N$ is required at any time, while storing all sample evaluation vectors costs
$\mathcal{O}(rN)$,
\begin{align}\label{eq:recovery-complexity}
\mathrm{memory}= \mathcal{O}(N^2+rN),
\\
T_{\rm{iter}} = \mathcal{O}(N^3+rN^2).
\end{align}
The $\mathcal{O}(N^3)$ term is due to the positive-semidefinite projection and certification, while evaluating all sampled quadratic forms costs $\mathcal{O}(rN^2)$ in a direct implementation.  The spectral candidate and, when used, the guarded KKT candidate have the same or lower asymptotic cost per active-set iteration.

For $r=\mathcal{O}(N)$, the persistent memory is $\mathcal{O}(N^2)$ and both initialization/certification and each projection sweep cost $\mathcal{O}(N^3)$.  If $I$ projection iterations are required, the total recovery cost is
\begin{align}
\mathcal{O}\!\left((I+1)N^3\right),
\end{align}
and no degree-independent bound on $I$ is assumed.  If the active sets become much larger than $N$, the evaluation vectors can instead be generated from the stored nodes on demand, reducing persistent storage to $\mathcal{O}(N^2+r)$.

\subsection{Proof of $(vi)$}\label{sec:lhat-proof}
After solving the dual problem in Thm.~\ref{thm:dual-nuclear} and retrieving the functions $p_k,\, q_k$ from $(v)$, let
\begin{gather}\label{eq:gP-gQ}
    \begin{aligned}
g_P( x)\coloneq \norm{\mathbf{p}( x)}_2^2 = \sum_k \vert p_k( x) \vert^2,
\\
g_Q( y)\coloneq \norm{ \mathbf{q}( y)}_2^2 = \sum_k \vert q_k( y)\vert^2.
\end{aligned}
\end{gather}
By $(v)$, the recovered functions form a valid factorization of $K_f$ on all of $\mathcal{D}$, so Eq.~\eqref{eq:Lstar} gives $L_\star\leq\sqrt{\max_{\mathcal{D}}g_P\cdot\max_{\mathcal{D}}g_Q}$. In exact arithmetic the two maxima are attained at endpoints or at real critical points of $g_P,g_Q$, which Algorithm~\ref{alg:active-set} also uses to refine the active sets. A numerical root finder can, however, miss a maximizer, and the resulting value would then not be an upper bound. We therefore certify the maxima with the following elementary bound, which only requires samples on a fixed grid.

\begin{lemma}[Certified maximum from samples]\label{lem:certified-max}
Let $g$ be a real-valued function on $\mathcal{D}$.
\begin{enumerate}[label=(\alph*)]
\item If $\mathcal{D}=[-1,1]$ and $g$ is a polynomial of degree at most $d\geq0$, let $M>\pi d/2$ and $x_k\coloneq\cos(k\pi/M)$, $k=0,\ldots,M$. Then
\begin{align}\label{eq:certified-max-cheb}
\max_{x\in[-1,1]}\vert g(x)\vert\leq\frac{\max_{k}\vert g(x_k)\vert}{1-\pi d/(2M)}.
\end{align}
\item If $\mathcal{D}=\mathbb{T}$ and $\theta\mapsto g(e^{\ii\theta})$ is a trigonometric polynomial of degree at most $d\geq0$, let $M>\pi d$ and $z_k\coloneq e^{2\pi\ii k/M}$, $k=0,\ldots,M-1$. Then
\begin{align}\label{eq:certified-max-torus}
\max_{z\in\mathbb{T}}\vert g(z)\vert\leq\frac{\max_{k}\vert g(z_k)\vert}{1-\pi d/M}.
\end{align}
\end{enumerate}
For $g=g_P$ with $P\succeq0$ we have $g\geq0$, so the left-hand sides equal $\max_{\mathcal{D}}g_P$; the same holds for $g_Q$. For the $N$-dimensional Chebyshev basis of Eq.~\eqref{eq:Hcheb}, $g_P$ has degree at most $d=2(N-1)$; for the Laurent basis $\Phi_k(z)=z^{k-N/2}$, $k=0,\ldots,N-1$, the function $\theta\mapsto g_P(e^{\ii\theta})$ has degree at most $d=N-1$. No structure of $P$ beyond $P\succeq0$ is assumed.
\end{lemma}

\begin{proof}
(a) The function $h(\theta)\coloneq g(\cos\theta)$ is a real trigonometric polynomial of degree at most $d$, even and $2\pi$-periodic, with $\norm{h}_\infty=\max_{[-1,1]}\vert g\vert$, where the sup-norm is over $[0,2\pi]$. Bernstein's inequality gives $\norm{h'}_\infty\leq d\,\norm{h}_\infty$. Every $\theta\in[0,\pi]$ lies within $\pi/(2M)$ of a node $\theta_k=k\pi/M$, so
\begin{align*}
\vert h(\theta)\vert\leq\vert h(\theta_k)\vert+\frac{\pi}{2M}\norm{h'}_\infty\leq\max_k\vert g(x_k)\vert+\frac{\pi d}{2M}\norm{h}_\infty.
\end{align*}
By evenness and periodicity, the supremum over $\theta\in[0,\pi]$ equals $\norm{h}_\infty$. Rearranging, with $1-\pi d/(2M)>0$, gives Eq.~\eqref{eq:certified-max-cheb}.

(b) Apply the same argument to $h(\theta)\coloneq g(e^{\ii\theta})$ on $[0,2\pi]$, where every $\theta$ lies within $\pi/M$ of a node $2\pi k/M$.

For the degrees: with real $T_k$, $g_P(\cos\theta)=\sum_{k,l}P_{kl}\cos(k\theta)\cos(l\theta)$ and $\cos(k\theta)\cos(l\theta)=\tfrac12[\cos((k+l)\theta)+\cos((k-l)\theta)]$ with $k+l\leq2(N-1)$; on $\mathbb{T}$, $g_P(e^{\ii\theta})=\sum_{k,l}P_{kl}e^{\ii(k-l)\theta}$ with $\vert k-l\vert\leq N-1$.
\end{proof}

Consequently, with $\bar{a}\geq\max_{\mathcal{D}}g_P$ and $\bar{b}\geq\max_{\mathcal{D}}g_Q$ obtained from Lemma~\ref{lem:certified-max},
\begin{align}\label{eq:continuous-certified-bound}
L_\star\leq\sqrt{\max_{x\in\mathcal{D}} g_P(x)}\cdot\sqrt{\max_{y\in\mathcal{D}} g_Q(y)}\leq L_{\mathrm{up}}\coloneq\sqrt{\bar{a}\,\bar{b}}.
\end{align}
For a prescribed $\delta\in(0,1)$, the grid sizes $M=\max\{1,\lceil\pi d/(2\delta)\rceil\}$ on $[-1,1]$ and $M=\max\{1,\lceil\pi d/\delta\rceil\}$ on $\mathbb{T}$ make the factor in Lemma~\ref{lem:certified-max} at most $1/(1-\delta)$, so that $L_{\mathrm{up}}\leq\sqrt{\max g_P\max g_Q}/(1-\delta)$. Since $d=\mathcal{O}(N)$, this requires $M=\mathcal{O}(N/\delta)$ evaluations of $g_P$ and $g_Q$, each costing $\mathcal{O}(N^2)$ arithmetic operations.

Notice that $L_{\mathrm{up}}$ is generally larger than the dual optimum $L_{\mathrm{low}}$, since the maxima of $g_P$ and $g_Q$ may occur away from the active sets $\mathsf{X}$ and $\mathsf{Y}$.

\subsection{Proof of $(vii)$: solution of the finite nuclear-norm dual}\label{sec:dual-nuclear-complexity}

Consider the finite dual for fixed sets $\mathsf{X}^{(t)}=\{x_i\}_{i=1}^{r_X}$ and
$\mathsf{Y}^{(t)}=\{y_j\}_{j=1}^{r_Y}$ at iteration $t$ of the active-set.  Rescale the weights in Eq.~\eqref{eq:dual-nuclear} according to
$\mathbf{u}=2\bar{\mathbf{u}}$ and
$\mathbf{v}=2\bar{\mathbf{v}}$.  Then
\begin{align}\label{eq:dual-unit-simplex}
L_{\mathrm{low}}=\max_{\substack{\bar{\mathbf{u}},\bar{\mathbf{v}}\geq 0 \\ \mathbf{1}^\top\bar{\mathbf{u}} +\mathbf{1}^\top\bar{\mathbf{v}}=1}}
2\norm{D_{\bar{u}}^{1/2}E_{\mathsf{X}} G E_{\mathsf{Y}}^{\dagger}D_{\bar{v}}^{1/2}}_* .
\end{align}
For fixed weights, take reduced QR factorizations
\begin{align}\label{eq:dual-weighted-qr}
D_{\bar{u}}^{1/2}E_{\mathsf{X}}=Q_u R_u, \quad D_{\bar{v}}^{1/2}E_{\mathsf{Y}}=Q_vR_v .
\end{align}
Since $Q_u$ and $Q_v$ have orthonormal columns,
\begin{align}\label{eq:dual-qr-reduction}
\norm{D_{\bar{u}}^{1/2}E_{\mathsf{X}} G E_{\mathsf{Y}}^{\dagger}D_{\bar{v}}^{1/2}}_* =\norm{R_u G R_v^\dagger}_* .
\end{align}
Hence, after taking a thin SVD,
\begin{align}\label{eq:dual-core-svd}
R_u G R_v^\dagger = U\Sigma V^\dagger ,
\end{align}
the dual objective simply
\begin{align}\label{eq:dual-inner-objective}
F(\bar{\mathbf{u}},\bar{\mathbf{v}})= 2\trace(\Sigma).
\end{align}

The same QR/SVD factors give a supergradient $\mathbf{g}=(\mathbf{g}_u,\mathbf{g}_v)$ of $F$. At differentiable points this is simply the gradient, and they have the asymptotic cost as Eqs.~\eqref{eq:dual-weighted-qr}--\eqref{eq:dual-core-svd}.
Since $F$ is concave on the unit simplex, for
\begin{align}
\mathbf{w} \coloneq
\begin{bmatrix}
\bar{\mathbf{u}}\\
\bar{\mathbf{v}}
\end{bmatrix},
\quad
\mathbf{1}^\top\mathbf{w}=1,
\end{align}
any valid supergradient satisfies
\begin{align}\label{eq:dual-stationarity-gap}
0 \leq L_{\mathrm{low}}-F(\mathbf{w}) \leq
\gamma(\mathbf{w}) \coloneq \max_i g_i-\mathbf{w}^\top\mathbf{g} .
\end{align}
Hence, $\gamma(\mathbf{w})\leq\varepsilon_{\rm{dual}}$ certifies an
additive $\varepsilon_{\rm{dual}}$ solution of the finite dual for a $\varepsilon_{\rm{dual}}>0$.

When we are solving this problem inside Algorithm~\ref{alg:active-set}, the previous dual weights provide a natural warm start.  Suppose
\begin{align}
\mathsf{X}^{(t-1)} \subseteq\mathsf{X}^{(t)},
\quad
\mathsf{Y}^{(t-1)}\subseteq\mathsf{Y}^{(t)} .
\end{align}
Embed $(\bar{\mathbf{u}}^{(t-1)}, \bar{\mathbf{v}}^{(t-1)})$ in a enlarged simplex by assigning zero weight to each newly added point.  Denoting the embedded weights by $(\widetilde{\mathbf{u}},\widetilde{\mathbf{v}})$, we have exactly
\begin{align}\label{eq:dual-warm-start-value}
F_t\!\left(\widetilde{\mathbf{u}},\widetilde{\mathbf{v}}\right)
=F_{t-1}\!\left(\bar{\mathbf{u}}^{(t-1)},\bar{\mathbf{v}}^{(t-1)}\right).
\end{align}
Thus the previous solution is a feasible boundary point of the
enlarged dual problem with the same previous objective value.

The first-order update is not well defined at exactly zero weights, so we perturb the zero-padded warm start slightly toward the interior of the simplex before beginning the numerical iteration.  If $\widetilde{\mathbf{w}}$ denotes the zero-padded vector, we initialize
\begin{align}\label{eq:dual-interior-warm-start}
\mathbf{w} = (1-\rho)\widetilde{\mathbf{w}} + \frac{\rho}{r_X^{(t)}+r_Y^{(t)}}\mathbf{1},
\quad 0<\rho\ll 1 .
\end{align}
Eq.~\eqref{eq:dual-warm-start-value} may still be retained as the initial best dual value, the perturbation in Eq.~\eqref{eq:dual-interior-warm-start} affects only the starting point of the numerical iteration.

The complete inner solve is summarized in Algorithm~\ref{alg:dual-nuclear-inner} below.

\begin{algo}\label{alg:dual-nuclear-inner}
\quad\\
\indent Input:
$G$, $E_{\mathsf{X}}$, $E_{\mathsf{Y}}$,
$\varepsilon_{\rm{dual}}$, and optional weights
$(\mathbf{u}^{(t-1)},
  \mathbf{v}^{(t-1)}$) from the Algorithm~\ref{alg:active-set}. \\
\indent Output:
$\bar{\mathbf{u}}_\star$,
$\bar{\mathbf{v}}_\star$, and a finite-dual value
$L_{\mathrm{low}}$.
\begin{enumerate}[leftmargin=*]

\item If previous weights are available, embed them by zero-padding
the newly added coordinates and initialize $\mathbf{w}^{(0)}$ by
Eq.~\eqref{eq:dual-interior-warm-start}, otherwise initialize at the
uniform interior point of the combined simplex.

\item At iteration $i$, form $D_{\bar{u}^{(k)}}^{1/2}E_{\mathsf{X}}$ and $D_{\bar{v}^{(k)}}^{1/2}E_{\mathsf{Y}}$, compute the reduced QR factorizations in Eq.~\eqref{eq:dual-weighted-qr}, and take the SVD
\begin{align*}
R_u G R_v^\dagger=U\Sigma V^\dagger .
\end{align*}

\item Set $F^{(k)}=2\trace(\Sigma)$, compute a supergradient $\mathbf{g}^{(k)}$ from the QR/SVD factors, and evaluate
\begin{align*}
\gamma_k = \max_i g_i^{(k)} - (\mathbf{w}^{(k)})^\top\mathbf{g}^{(k)} .
\end{align*}
If $\gamma_k\leq\varepsilon_{\rm{dual}}$, return the current weights
and objective value.

\item Otherwise define the exponentiated-gradient path
\begin{align*}
w_i^{(k)}(\eta) \propto \max\{w_i^{(k)},w_{\min}\}
\exp\!\left( \eta\, \frac{g_i^{(k)}}{\max\{1,\|\mathbf g^{(k)}\|_\infty\}} \right),
\end{align*}
with normalization to the unit simplex.  Use a golden-section search on the prescribed interval
$[\eta_{\min},\eta_{\max}]$ to maximize $F(\mathbf{w}^{(k)}(\eta))$, set
\begin{align*}
\eta_k \approx \arg\max_{\eta\in[\eta_{\min},\eta_{\max}]} F(\mathbf{w}^{(k)}(\eta)),
\end{align*}
and take
$\mathbf{w}^{(k+1)}=\mathbf{w}^{(k)}(\eta_k)$.
Return to Step~2.
\end{enumerate}
\end{algo}

We finally count the resources of one iteration. The line search in in Step~$4$ changes only the step length along the exponentiated-gradient path.  Each trial value of $\eta$ requires an objective evaluation through Eqs.~\eqref{eq:dual-weighted-qr}--\eqref{eq:dual-inner-objective}, but not a new supergradient. Row scaling requires $\mathcal{O}((r_X+r_Y)N)$ operations.  For $r_X,r_Y\geq N$, the two reduced QR factorizations require
\begin{align}
\mathcal{O}\!\left((r_X+r_Y)N^2\right)
\end{align}
operations, while forming $R_uGR_v^\dagger$ and computing its $N\times N$ SVD require $\mathcal{O}(N^3)$.  The supergradient contractions have the same combined asymptotic cost.
Thus, if $I_{\rm{obj}}$ objective evaluations are used by the golden-section search, one first-order iteration therefore costs
\begin{align}\label{eq:dual-inner-line-search-time}
\mathcal{O}\!\left( I_{\rm{obj}}\big[(r_X+r_Y)N^2+N^3\big]\right)
\end{align}
arithmetic operations.

The largest stored objects are the evaluation and thin-$Q$ matrices
of sizes $r_X\times N$ and $r_Y\times N$, together with a constant
number of $N\times N$ coefficient matrices.  Hence
\begin{align}\label{eq:dual-inner-memory}
\mathcal{O}\!\left(
(r_X+r_Y)N+N^2
\right)
\end{align}
working memory is sufficient.

If $I_{fo}$ first-order iterations are required and each golden-section line search uses at most $I_{\rm{gs}}$ objective evaluations, the complete finite-dual solve costs
\begin{align}\label{eq:dual-inner-total-time}
\mathcal{O}\!\left( I_{fo}I_{\rm{gs}}\big[(r_X+r_Y)N^2+N^3\big] \right)
\end{align}
time and
\begin{align}\label{eq:dual-inner-memory-total}
\mathcal{O}\!\left(
(r_X+r_Y)N+N^2
\right)
\end{align}
working memory. Warm starting can reduce $I_{fo}$, whereas the golden-section search controls the number $I_{\rm{gs}}$ of coefficient-space objective evaluations used to choose each step. 

An a-priori iteration bound is also available for a safeguarded variant. Let $m\coloneq r_X+r_Y$ and, for $\delta\in(0,1)$, apply exponentiated-gradient ascent with a fixed step to the concave function $F_\delta(\mathbf{w})\coloneq F\big((1-\delta)\mathbf{w}+\delta\mathbf{1}/m\big)$. Since every evaluated point has entries at least $\delta/m$, the supergradients obey $\norm{\mathbf{g}}_\infty\le C\sqrt{m/\delta}$ with $C\coloneq\max_{i,j}\vert K_f(x_i,y_j)\vert\le L_{\mathrm{low}}$. Moreover, $\max F_\delta\ge(1-\delta)L_{\mathrm{low}}$ by concavity. The standard entropic mirror-ascent bound with $\delta=\varepsilon_{\rm rel}/2$ then shows that the averaged iterate reaches $F\ge(1-\varepsilon_{\rm rel})L_{\mathrm{low}}$ within $\lceil16\,m\log m/\varepsilon_{\rm rel}^3\rceil$ iterations, each with the cost stated below. This worst-case bound is pessimistic, and the line-search iteration of Algorithm~\ref{alg:dual-nuclear-inner}, certified by Eq.~\eqref{eq:dual-stationarity-gap}, typically terminates in far fewer iterations.

$I_{fo}$ is bounded a priori only for the safeguarded variant above. An $\varepsilon_{\rm{dual}}$-optimal value is instead certified a posteriori by Eq.~\eqref{eq:dual-stationarity-gap}, and, by Corollary~\ref{cor:finite-section}(a), every iterate already provides a valid lower bound on $L_\star$.


\section{Proof of Algorithm~\ref{alg:active-set} convergence}\label{sec:proof-alg-conv}
\begin{theorem}[Convergence of the iterative scheme]\label{thm:iterative-convergence}
Assume $t_{\max}=\infty$, $\varepsilon_{\mathrm{tol}}=0$, exact solution of the
finite dual at every iteration, exact recovery at $\widehat{s}_{\rm{rec}}^{(t)}=L_{\mathrm{low}}^{(t)}$ from Sec.~\ref{sec:recov-haagerup}, and exact global maximization in Step~3 of Algorithm~\ref{alg:active-set}, so that $\bar{a}_t=\max_{\mathcal{D}}g_P^{(t)}$ and $\bar{b}_t=\max_{\mathcal{D}}g_Q^{(t)}$.  Assume that the initial evaluation matrices $E_{\mathsf{X}}^{(0)}$ and $E_{\mathsf{Y}}^{(0)}$ have full column rank, and that the refinement rule in
Eq.~\eqref{eq:update-active-sets} inserts every minimizer of
\begin{align}
h_P^{(t)}(x)\coloneq L_{\mathrm{low}}^{(t)}-g_P^{(t)}(x), \quad
h_Q^{(t)}(y)\coloneq L_{\mathrm{low}}^{(t)}-g_Q^{(t)}(y)
\end{align}
at which the corresponding violation is negative.  Then:
\begin{enumerate}[leftmargin=*,label=(\roman*)]
    \item $\{L_{\mathrm{low}}^{(t)}\}_{t\ge0}$ is monotone nondecreasing and bounded above by $L_\star$;

    \item if the algorithm terminates at some iteration $t$, then
    \begin{align}
    L_{\mathrm{low}}^{(t)}=L_{\mathrm{up}}^{\mathrm{best},(t)}=L_\star;
    \end{align}
    \item if the algorithm does not terminate, then
    \begin{align}
    L_{\mathrm{low}}^{(t)}\rightarrow L_\star,
    \quad
    L_{\mathrm{up}}^{(t)}\rightarrow L_\star,
    \end{align}
    and therefore
    \begin{align}
    L_{\mathrm{up}}^{(t)}-L_{\mathrm{low}}^{(t)}\rightarrow 0.
    \end{align}
\end{enumerate}
\end{theorem}
\begin{proof}
We divide the proof into three steps, starting with $(i)$. Since the active sets are nested by construction,
\[
\mathsf{X}^{(t)}\subseteq \mathsf{X}^{(t+1)},
\quad
\mathsf{Y}^{(t)}\subseteq \mathsf{Y}^{(t+1)},
\]
the feasible set of the problem at iteration $t+1$ is contained in the feasible set at iteration $t$.
Strong duality therefore gives
\begin{align}
L_{\mathrm{low}}^{(t)} \leq L_{\mathrm{low}}^{(t+1)} \leq L_\star.
\end{align}
Hence
\begin{align}
d_\infty\coloneq \lim_{t\to\infty}L_{\mathrm{low}}^{(t)}
\end{align}
exists and satisfies $d_\infty\leq L_\star$, proving $(i)$.

For $(ii)$, exact recovery gives a full-domain factorization at every iteration $s\leq t$, hence
$L_\star\leq L_{\mathrm{up}}^{(s)}$ and $L_\star\leq L_{\mathrm{up}}^{\mathrm{best},(t)}$, whereas $(i)$ gives
$L_{\mathrm{low}}^{(t)}\leq L_\star$.  With
$\varepsilon_{\mathrm{tol}}=0$, the algorithm terminates either because $L_{\mathrm{up}}^{\mathrm{best},(t)}\leq L_{\mathrm{low}}^{(t)}$, or because both sets in Eq.~\eqref{eq:active-violations} are empty. In the latter case the maxima of $g_P^{(t)}$ and $g_Q^{(t)}$, which are attained at points of $\mathsf{X}^{(t)}_{\mathrm c}$ and $\mathsf{Y}^{(t)}_{\mathrm c}$, are at most $L_{\mathrm{low}}^{(t)}$, so $L_{\mathrm{up}}^{\mathrm{best},(t)}\leq L_{\mathrm{up}}^{(t)}\leq L_{\mathrm{low}}^{(t)}$. In both cases
\begin{align}
L_{\mathrm{low}}^{(t)}
=L_{\mathrm{up}}^{\mathrm{best},(t)}=L_\star.
\end{align}

It remains to prove $(iii)$.  We first establish the compactness needed for
the non-terminating sequence.  Let
\begin{align}
\Gamma^{(t)}=
\begin{bmatrix}
P^{(t)}&G\\
G^\dagger&Q^{(t)}
\end{bmatrix}
\succeq 0
\end{align}
be the exactly recovered Gram matrix at iteration $t$, and define
\begin{align*}
\alpha_X \coloneq \lambda_{\min}\!\left( E_{\mathsf{X}}^{(0)\dagger}E_{\mathsf{X}}^{(0)} \right)>0,
\\
\alpha_Y \coloneq \lambda_{\min}\!\left(E_{\mathsf{Y}}^{(0)\dagger}E_{\mathsf{Y}}^{(0)} \right)>0,
\end{align*}
where $\lambda_{\min}(X)$ is the minimum eigenvalue of $X$. Since the initial points remain in every later active set, exact finite feasibility gives
\begin{gather}\label{eq:convergence-gram-bounds}
\begin{aligned}
\alpha_X\trace(P^{(t)}) \leq r_X^{(0)}L_{\mathrm{low}}^{(t)} \leq r_X^{(0)}L_\star,
\\
\alpha_Y\trace( Q^{(t)}) \leq r_Y^{(0)}L_{\mathrm{low}}^{(t)} \leq r_Y^{(0)}L_\star.
\end{aligned}
\end{gather}
Hence $\{P^{(t)}\}$ and $\{Q^{(t)}\}$ are bounded.  This is where the
full-column-rank assumption on the initial evaluation matrices is used; the
$2N+1$ distinct points of Step~1 of Algorithm~\ref{alg:active-set} satisfy it. Now define
\begin{align}
a_t \coloneq \max_{x\in\mathcal{D}}g_P^{(t)}(x),
\quad
b_t \coloneq \max_{y\in\mathcal{D}}g_Q^{(t)}(y),
\end{align}
so that
\begin{align}
L_{\mathrm{up}}^{(t)}=\sqrt{a_t b_t},
\end{align}
and let
\begin{align}\label{eq:convergence-rho}
\rho_t \coloneq \max\!\left\{0,\, a_t-L_{\mathrm{low}}^{(t)},\,b_t-L_{\mathrm{low}}^{(t)} \right\}.
\end{align}
We claim that $\rho_t\to 0$.  Otherwise, for some $\varepsilon>0$ there is a subsequence $t_\nu$ on which one of the two violations is at least $\varepsilon$.  Without loss of generality,
\begin{align}\label{eq:convergence-positive-violation}
a_{t_\nu} \geq L_{\mathrm{low}}^{(t_\nu)}+\varepsilon.
\end{align}
Let $x_\nu$ maximize $g_P^{(t_\nu)}$.  By compactness of $\mathcal{D}$ and the bound in Eq.~\eqref{eq:convergence-gram-bounds}, after passing to a further subsequence,
\begin{align}
x_\nu\to x^\sharp, \quad P^{(t_\nu)}\to P^\sharp.
\end{align}
Since $\mathbf{\Phi}$ is continuous on the compact domain,
\begin{align}
g_P^{(t_\nu)}(x) = \mathbf{\Phi}(x)^\top P^{(t_\nu)} \mathbf{\Phi}(x)^*
\end{align}
converges uniformly to
\begin{align}
g_P^\sharp(x) \coloneq \mathbf{\Phi}(x)^\top P^\sharp\mathbf{\Phi}(x)^*.
\end{align}
Together with $L_{\mathrm{low}}^{(t_\nu)}\to d_\infty$ in (i), Eq.~\eqref{eq:convergence-positive-violation} therefore yields
\begin{align}\label{eq:convergence-limit-violation}
g_P^\sharp(x^\sharp)\geq d_\infty+\varepsilon.
\end{align}

Notwithstanding, $x_\nu$ is a global minimizer of $h_P^{(t_\nu)}$ with negative value and is therefore inserted into the next active set.  Since the active sets are nested, for every fixed $\nu$ and all sufficiently large $\mu$,
\begin{align}
g_P^{(t_\mu)}(x_\nu) \leq L_{\mathrm{low}}^{(t_\mu)}.
\end{align}
Taking $\mu\to\infty$ along the same convergent subsequence and then $\nu\to\infty$ gives
\begin{align}
g_P^\sharp(x^\sharp)\leq d_\infty,
\end{align}
contradicting Eq.~\eqref{eq:convergence-limit-violation}.  The same argument
applies to the $Q$ side, so $\rho_t \to 0$.

Finally,
\begin{align}
a_t,b_t \leq L_{\mathrm{low}}^{(t)}+\rho_t,
\end{align}
and the full-domain factorization gives
\begin{align}
L_\star \leq L_{\mathrm{up}}^{(t)} = \sqrt{a_t b_t} \leq L_{\mathrm{low}}^{(t)}+\rho_t.
\end{align}
Taking $t\to\infty$ and using $(i)$,
\begin{align}
L_\star \leq \liminf_tL_{\mathrm{up}}^{(t)} \leq \limsup_tL_{\mathrm{up}}^{(t)} \leq d_\infty \leq L_\star.
\end{align}
Hence
\begin{align}
L_{\mathrm{low}}^{(t)}
\rightarrow L_\star,
\quad
L_{\mathrm{up}}^{(t)} \rightarrow L_\star,
\end{align}
and therefore
\begin{align}
L_{\mathrm{up}}^{(t)}- L_{\mathrm{low}}^{(t)} \rightarrow 0.
\end{align}
This proves $(iii)$.
\end{proof}

For finite numerical recovery, the same argument applies with
$\widehat{s}_{\rm{rec}}^{(t)}$ in place of
$L_{\mathrm{low}}^{(t)}$ in the sampled feasibility inequalities and
in the violation polynomials, provided
\begin{align}
\widehat{s}_{\rm{rec}}^{(t)}-L_{\mathrm{low}}^{(t)}\longrightarrow0.
\end{align}
Thus a vanishing recovery tolerance preserves convergence of both certificate
endpoints.  A fixed positive recovery tolerance retains the finite lower--upper
certificates but does not imply exact asymptotic collapse.

\begin{corollary}[Running best upper bound and finite grids]\label{cor:running-best}
Let $L_{\mathrm{up}}^{\mathrm{best},(t)}$ be defined by Eq.~\eqref{eq:Lup-best}.
\begin{enumerate}[label=(\alph*),leftmargin=*]
\item In every run of Algorithm~\ref{alg:active-set} in which each $L_{\mathrm{up}}^{(s)}$ is a valid upper bound on $L_\star$, $L_\star\leq L_{\mathrm{up}}^{\mathrm{best},(t+1)}\leq L_{\mathrm{up}}^{\mathrm{best},(t)}$ for all $t$.
\item Under the hypotheses of Theorem~\ref{thm:iterative-convergence}, if the algorithm does not terminate, then $L_{\mathrm{up}}^{\mathrm{best},(t)}\to L_\star$.
\item Suppose, in exact arithmetic, that the exact maxima in Step~3 are replaced by the bounds of Lemma~\ref{lem:certified-max} with grids of size $M_t$, and let $\omega_t\coloneq\pi d/(2M_t)$ on $[-1,1]$ ($\omega_t\coloneq\pi d/M_t$ on $\mathbb{T}$), with $\omega_t<1$. Assume all other hypotheses of Theorem~\ref{thm:iterative-convergence}. If the algorithm terminates at iteration $t$, then $L_{\mathrm{low}}^{(t)}=L_\star\leq L_{\mathrm{up}}^{\mathrm{best},(t)}\leq L_\star/(1-\omega_t)$. If it does not terminate and $\omega_t\to0$, then $L_{\mathrm{low}}^{(t)}\to L_\star$ and $L_{\mathrm{up}}^{\mathrm{best},(t)}\to L_\star$.
\end{enumerate}
\end{corollary}

\begin{proof}
(a) $L_{\mathrm{up}}^{\mathrm{best},(t+1)}=\min\{L_{\mathrm{up}}^{\mathrm{best},(t)},L_{\mathrm{up}}^{(t+1)}\}$, and every term is at least $L_\star$.
(b) $L_\star\leq L_{\mathrm{up}}^{\mathrm{best},(t)}\leq L_{\mathrm{up}}^{(t)}\to L_\star$ by Theorem~\ref{thm:iterative-convergence}(iii).
(c) Let $a_t,b_t$ be the exact maxima of $g_P^{(t)},g_Q^{(t)}$. In exact arithmetic Lemma~\ref{lem:certified-max} gives $a_t\leq\bar{a}_t\leq a_t/(1-\omega_t)$ and $b_t\leq\bar{b}_t\leq b_t/(1-\omega_t)$, hence $L_\star\leq\sqrt{a_tb_t}\leq L_{\mathrm{up}}^{(t)}\leq\sqrt{a_tb_t}/(1-\omega_t)$. The proof of Theorem~\ref{thm:iterative-convergence} uses only the refinement rule, which is unchanged, and not the way the upper bound is evaluated; it therefore still gives $L_{\mathrm{low}}^{(t)}\to L_\star$ and $\sqrt{a_tb_t}\to L_\star$ in the non-terminating case, and the claim follows. If the algorithm terminates by the gap test, then $L_\star\leq L_{\mathrm{up}}^{\mathrm{best},(t)}\leq L_{\mathrm{low}}^{(t)}\leq L_\star$. If it terminates because both violation sets are empty, then, as in the proof of Theorem~\ref{thm:iterative-convergence}(ii), $\sqrt{a_tb_t}\leq L_{\mathrm{low}}^{(t)}\leq L_\star\leq\sqrt{a_tb_t}$, and $L_{\mathrm{up}}^{\mathrm{best},(t)}\leq L_{\mathrm{up}}^{(t)}\leq L_\star/(1-\omega_t)$.
\end{proof}

Theorem~\ref{thm:iterative-convergence} and Corollary~\ref{cor:running-best} are exact-arithmetic statements. In floating point, Algorithm~\ref{alg:active-set} retains the validity of the returned bracket, but not the convergence guarantee. The lower bound is valid by $(iv)$ and Corollary~\ref{cor:finite-section}(a), which bound every feasible dual value, and the upper bound by Lemma~\ref{lem:certified-max} with finite precision corrections.

\section{Resource model}\label{sec:lcu-res-model}

Here we specify the fault-tolerant resource function $R$ used in Sec.~\ref{sec:resource-bounds}. We take the resource metric to be the total number of $T$ gates and use the LCU access model of Sec.~\ref{sec:end-to-end-budget}, with the \textsc{prepare} construction of Ref.~\cite{Low_2024}. 
The error parameters $\varepsilon_{\rm trunc}$, $\varepsilon_{\rm fp}$, and $\varepsilon_{\rm synth}$ have the meanings introduced there, and their admissible values are constrained by Eq.~\eqref{eq:end-to-end-budget}.

Let the linear combination of unitaries contain $M$ terms. For a truncation tolerance $\varepsilon_{\rm{trunc}}$, define $m(\varepsilon_{\rm{trunc}})$ as the smallest number of retained terms whose discarded coefficient weight satisfies
$\norm{A-A_{\rm{trunc}}}\leq
\alpha^{-1}\sum_{j\notin S}\vert a_j\vert \coloneq \varepsilon_{\rm{trunc}}$. After ordering the coefficients such that $\vert a_{(1)}\vert \geq \cdots \geq \vert a_{(M)} \vert$, this is
\begin{align}
m(\varepsilon_{\rm{trunc}}) \coloneq
\min\left\{
m:\, \frac{1}{\alpha} \sum_{j=m+1}^{M} \vert a_{(j)} \vert \leq \varepsilon_{\rm{trunc}} \right\}.
\end{align}
Following Ref.~\cite{Low_2024}, we have to round up for a power-of-two, so we use
\begin{align}\label{eq:retained-lcu-terms}
\widetilde m(\varepsilon_{\rm{trunc}})
\coloneq 2^{\left\lceil \log_2 m(\varepsilon_{\rm{trunc}}) \right\rceil},
\end{align}
padding the coefficient list with zeros when necessary. Using these $\widetilde{m}$ terms, the \textsc{prepare} oracle has $T$-gate count
\begin{gather}\label{eq:ft-prepare-cost}
\begin{aligned}
&T_{\text{\scshape prep}}(\widetilde{m},\varepsilon_{\rm{fp}},b) \\
&\quad = \frac{4\widetilde{m}}{2^b} + 4\,2^b
\left[ \left\lceil \log_2\!\left(\frac{1}{\varepsilon_{\rm{fp}}}\right) \right\rceil
+ \log_2(\widetilde{m}) \right] \\
&\qquad + 10 \left\lceil \log_2\!\left(\frac{1}{\varepsilon_{\rm{fp}}}\right) \right\rceil -11,
\end{aligned} 
\end{gather}
where $b\in \mathbb{N}_+$ defines the tradeoff between T gates and qubits. For the \textsc{select} oracle, we have
\begin{align}\label{eq:ft-select-cost}
T_{\text{\scshape sel}}(\widetilde{m}) = 4\widetilde{m}.
\end{align}
Therefore, since one block-encoding query contains one \textsc{select} and both \textsc{prepare} and $\textsc{prepare}^{\dagger}$, its cost is
\begin{gather}\label{eq:ft-block-encoding-cost}
\begin{aligned}
T_{\rm{be}}(\widetilde{m},\varepsilon_{\rm{fp}},b)  =
T_{\text{\scshape sel}}(\widetilde{m})+2T_{\text{\scshape prep}}(\widetilde{m},\varepsilon_{\rm{fp}},b) 
\end{aligned}
\end{gather}
For fixed $\widetilde{m}$ and $\varepsilon_{\rm{fp}}$, we choose the
integer $b$ that minimizes this cost subject to the available number of qubits,
\begin{align}\label{eq:ft-b-opt}
b_\star(\widetilde{m},\varepsilon_{\rm{fp}}) \coloneq \arg\min_{b\in\mathcal{B}_{\rm{hw}}} T_{\rm{be}} (\widetilde{m},\varepsilon_{\rm{fp}},b),
\end{align}
where $\mathcal{B}_{\rm{hw}}$ denotes the set of values allowed by the
hardware. Without any additional restriction on $b$, the
continuous relaxation of Eq.~\eqref{eq:ft-b-opt} has a minimum near
\begin{align}\label{eq:ft-b-continuous}
2^{b_\star} \approx \sqrt{\frac{\widetilde{m}}{\left\lceil\log_2(1/\varepsilon_{\rm{fp}})\right\rceil+\log_2(\widetilde{m})}}.
\end{align}

The QSP/QSVT sequence associated with a polynomial $f$ requires $S[f]$ rotations. We assume that its single-qubit rotations are synthesized to error $\varepsilon_{\rm{synth}}$ using
\begin{align}\label{eq:ft-rz-cost}
\# T(\varepsilon_{\rm{synth}})= 3\log_2 \left(\frac{1}{\varepsilon_{\rm{synth}}}\right)
\end{align}
$T$ gates \cite{ross2016}. Thus their total contribution is
\begin{align}\label{eq:ft-qsp-rotation-cost}
T_{\phi}(f,\varepsilon_{\rm{synth}})=3S[f]\log_2\left(\frac{1}{\varepsilon_{\rm{synth}}} \right).
\end{align}

Combining the oracle and phase-synthesis costs gives the resource
function
\begin{gather}\label{eq:ft-resource-explicit}
\begin{aligned}
&R\left(f,\varepsilon_{\rm{trunc}},\varepsilon_{\rm{fp}},\varepsilon_{\rm{synth}}\right) \\
&=N\left[4\widetilde{m}+\frac{8\widetilde{m}}{2^{b_\star}} +8\,2^{b_\star}\left(\left\lceil \log_2\!\frac{1}{\varepsilon_{\rm{fp}}} \right\rceil + \log_2\widetilde{m} \right) \right.  \\
&\qquad\left.+20\left\lceil \log_2\!\frac{1}{\varepsilon_{\rm{fp}}} \right\rceil -22 \right]+3S[f]\log_2\!\frac{1}{\varepsilon_{\rm{synth}}},
\end{aligned}
\end{gather}
with $\widetilde{m} = \widetilde{m}(\varepsilon_{\rm{trunc}}), \quad b_\star=b_\star\left(\widetilde{m}(\varepsilon_{\rm{trunc}}), \varepsilon_{\rm{fp}} \right)$.

\section{Details of the fixed-degree coefficient search}
\label{sec:searchN-details}

Here we describe the realization of $\mathsf{Search}_N$ used in Sec.~\ref{sec:res-aware-poly-sel}. The purpose of the search is to modify the coefficients of a fixed-degree polynomial while preserving the approximation accuracy and preferentially moving toward polynomials with smaller operator-Lipschitz sensitivity. The complete fault-tolerant resource cost is evaluated for polynomials produced by this search, as described in Algorithm~\ref{alg:res-aware-poly-design}.

\paragraph{Approximation-feasible coefficient space.}
For the fixed degree $N$, we restrict the search to
\begin{align}\label{eq:search-design-domain}
    \mathcal{A}_{N,\varepsilon}^{\rm{design}}
    \coloneq
    \left\{\mathbf{a}\in\mathcal{A}_{N,\varepsilon}:\varepsilon_{\rm{app}}[f_{\mathbf{a}}] \leq\varepsilon_{\rm{design}} \right\},
\end{align}
where $\varepsilon_{\rm{design}}$ is fixed before the search begins. The initial polynomial must already belong to this set. In our implementation, the certified minimax error of a Remez polynomial \cite{remez1934} provides a reference for choosing $\varepsilon_{\rm{design}}$.

It is convenient to write the free coefficients as
\begin{align}\label{eq:search-affine-coefficients}
    \mathbf{a}(\mathbf{z}) = \mathbf{a}_{\rm{fixed}}+A\mathbf{z},
\end{align}
so that parity and other linear admissibility conditions are enforced by construction. The search then acts only on the free coordinates $\mathbf{z}$.

Consider the coefficient matrix $G(\mathbf{a})$ associated with
$f_{\mathbf{a}}$ and a feasible dual point $(\mathbf{u},\mathbf{v},X)$ obtained during an operator-Lipschitz calculation. If these dual variables are held fixed, then
\begin{align}\label{eq:search-linear-lower-bound}
    \ell_{u,v,X}(\mathbf{a}) \coloneq \mathrm{Re}
    \trace\! \left[X^\dagger G(\mathbf{a})\right]
\end{align}
is affine in $\mathbf{a}$ and satisfies
\begin{align}\label{eq:search-linear-lower-bound-valid}
    \ell_{u,v,X}(\mathbf{a}) \leq L_\star[f_{\mathbf{a}}].
\end{align}
Thus, every operator-Lipschitz evaluation leaves behind a simple linear lower bound that can be reused when considering other polynomials. We retain a finite collection of these bounds, denoted by $\ell_1,\ldots,\ell_R$, and
\begin{align}\label{eq:search-lower-model}
    L_R(\mathbf{a}) \coloneq \max\left\{0,\ell_1(\mathbf{a}),\ldots,\ell_R(\mathbf{a})\right\}.
\end{align}
This quantity is inexpensive to evaluate and gives the search a direction in coefficient space without requiring a new full operator-Lipschitz calculation at every trial point.

\paragraph{Proposing a new polynomial.}
Let $\mathbf{z}_r$ denote the current accepted coefficient vector. As in the Remez algorith~\cite{remez1934}, we maintain a finite set $\mathsf{E}_r\subset\mathcal{D}$ of points at which the approximation constraint is explicitly imposed. Starting from Chebyshev--Lobatto points on every connected component of $\mathcal{D}$, the set is enlarged whenever the continuum approximation check identifies a new violating point.

The next coefficient proposal is obtained from the linear program (LP)
\begin{align}
    \min_{\mathbf{z},t}\quad & t  \label{eq:search-master} \\
    \mathrm{s.t.}\quad & \left\vert f_{\mathbf{a}(\mathbf{z})}(x)-f_{\rm T}(x) \right\vert \leq \varepsilon_{\rm{design}},
    \qquad x\in\mathsf{E}_r, \label{eq:search-master-approx} \\
    & \ell_j(\mathbf{a}(\mathbf{z})) \leq t, \qquad \qquad \qquad j=1,\ldots,R \label{eq:search-master-op} \\
    & \vert z_k-z_{r,k}\vert \leq \Delta_r s_k, \qquad k=1,\ldots,m_N, \label{eq:search-master-trust} \\ 
    &t \geq 0 . \nonumber
\end{align}
Here $\Delta_r$ limits the size of one coefficient update and the positive factors $s_k$ account for different coefficient scales. The objective therefore asks for a nearby polynomial for which all currently available operator lower bounds are as small as possible, while the sampled approximation constraints remain satisfied. The trust region $\Delta_r s_k$ is important because the linear bounds accumulated from previous operator calculations are most informative near the part of coefficient space that has already been explored. It also prevents a single LP step from producing an unnecessarily large change in the polynomial.

\paragraph{Continuum approximation check.}
Satisfying Eq.~\eqref{eq:search-master-approx} on the finite set $\mathsf{E}_r$ does not guarantee $\varepsilon_{\rm{app}}\leq\varepsilon_{\rm{design}}$ on the full domain. We therefore do not accept the LP solution directly.

Instead, let $\mathbf{a}_r$ be the current feasible polynomial and $\mathbf{a}_{\rm LP}$ the polynomial proposed by Eq.~\eqref{eq:search-master}. We consider the segment
\begin{align*}
    \mathbf{a}(\alpha) = (1-\alpha)\mathbf{a}_r + \alpha\mathbf{a}_{\rm LP}, \qquad 0 \leq \alpha \leq 1,
\end{align*}
and determine how far along this segment the approximation budget remains valid. The trial polynomial is chosen on the feasible side of this boundary. Any newly detected approximation extrema are added to $\mathsf{E}_r$ and are therefore enforced explicitly in subsequent LPs.

\paragraph{Deciding when to run the full operator calculation.}
The candidate polynomial then goes through Algorithm~\ref{alg:active-set} to obtain
\begin{align*}
    L_{\mathrm{low}}[f_{\mathbf{a}}] \leq L_\star[f_{\mathbf{a}}] \leq L_{\mathrm{up}}[f_{\mathbf{a}}].
\end{align*}
The resulting dual information supplies additional linear lower bounds of the form Eq.~\eqref{eq:search-linear-lower-bound}, which are added to the model for subsequent iterations.

\paragraph{Updating the search.}
Let $L_{{\mathrm{up}}_r}$ be the certified upper bound of the current accepted polynomial. After going through Algorithm~\ref{alg:active-set}, define
\begin{align}
    \Delta_{\rm{act}} &\coloneq L_{{\mathrm{up}}_r}-L_{\mathrm{up}}[f_{\mathbf{a}}], \\
    \Delta_{\rm{pred}}&\coloneq L_{{\mathrm{up}}_r}-L_R(\mathbf{a}).
\end{align}
The first quantity is the improvement established by the full certificate, while the second is the improvement predicted by the inexpensive linear model. The candidate replaces the current polynomial when it produces a genuine decrease of the certified upper bound and the ratio
\begin{align*}
\frac{\Delta_{\rm{act}}}{\Delta_{\rm{pred}}}
\end{align*}
shows sufficient agreement between prediction and certification. When the prediction is reliable the allowed coefficient step can be increased, when it is unsuccessful the allowed step is reduced.

A rejected polynomial is not necessarily wasted. The operator calculation can still contribute new lower bounds to Eq.~\eqref{eq:search-lower-model}, thereby improving the next LP.

The process is repeated until the search ceases to produce useful updates.

\paragraph{Resource-cost selection.}
The preceding steps are used only to explore coefficient space efficiently. They do not replace the resource objective of Sec.~\ref{sec:res-aware-poly-sel}. Let $\mathcal{P}_N$ contain the initial and every polynomial encountered during the search for which both the approximation error and the operator-Lipschitz bounds are available. For each $f\in\mathcal{P}_N$, we independently solve the implementation-error allocation problem and evaluate
\begin{align*}
    \mathcal C_N[f,L_{\mathrm{up}}[f],\varepsilon].
\end{align*}
The fixed-degree output is then the polynomial in $\mathcal{P}_N$ with the smallest certified resource cost, as in Eq.~\eqref{eq:opt-poly}.

Thus the internal search uses the operator-Lipschitz quantity to find promising regions of coefficient space, while the final selection among the polynomials it actually certifies is made with the full fault-tolerant resource model. No global optimality over
$\mathcal{A}_{N,\varepsilon}$ is implied.

\end{document}